\documentclass[a4paper,11pt]{article}
\usepackage{jheppub} 
\usepackage{lineno}
\usepackage{amsmath}
\usepackage{amssymb}
\usepackage{amsthm}

\newtheorem{lemma}{Lemma}[section]
\newtheorem{theorem}{Theorem}[section]
\newtheorem{corollary}{Corollary}[section]

\usepackage{bbold}
\usepackage{enumitem}
\usepackage{url}
\usepackage{verbatim}
\usepackage[dvipsnames]{xcolor}
\definecolor{TBlueLight}{HTML}{DCECFF}
\definecolor{TRedLight}{HTML}{FFAAAA}
\definecolor{TOrangeLight}{HTML}{FFD6B6}
\definecolor{TPurpleLight}{HTML}{EAC4FF}
\definecolor{TGreyLight}{HTML}{E0E0E0}

\usepackage{microtype}
\usepackage{tikz}
\usetikzlibrary{bending, positioning, calc, decorations.pathmorphing, decorations.markings, arrows.meta}
\newcommand{\tbeta}{\widetilde{\beta}}
\newcommand{\wh}{\widetilde{h}}
\newcommand{\wR}{\widetilde{R}}
\newcommand{\ba}{\overline{a}}
\newcommand{\bb}{\overline{b}}
\newcommand{\bq}{\mathfrak{q}}
\newcommand{\ket}[1]{\left| #1 \right>}

\newcommand{\VS}[1]{{\color{CornflowerBlue}\textbf{VS:} #1}}
\newcommand{\AB}[1]{{\color{Maroon}(\textbf{AB:} #1)}}
\newcommand{\beq}{\begin{equation}}
\newcommand{\eeq}{\end{equation}}
\newcommand{\beqn}{\begin{eqnarray}}
\newcommand{\eeqn}{\end{eqnarray}}

\newcommand{\mat}[4]{\begin{pmatrix}
\newcommand{\ssb}{\mathsf{b}}
  #1 & #2 \\
  #3 & #4
\end{pmatrix}}

\def\al{\alpha}

\def\lm{\lambda}
\def\eps{\epsilon}
\def\om{\omega}
\title{A semiclassical Hilbert space for random matrix theory}

\author[a]{Abhirup Bhattacharya,}
\author[a]{Onkar Parrikar,}
\author[a]{Vivek Singh}

\affiliation[a]{Department of Theoretical Physics, 
Tata Institute of Fundamental Research, 1 Homi Bhabha Road, Mumbai 400005, India}

\abstract{We study the one-sided time evolution of thermofield double (TFD) states in random matrix theory, where the Hamiltonian is taken to be a $D\times D$ random matrix drawn from a unitarily invariant emsemble of Hermitian matrices. We argue that the Krylov basis with the maximally entangled state taken as the initial vector gives a semiclassical Hilbert space description of these TFD states in random matrix theory at any $O(1)$ temperature and time in the $D\to \infty$ limit, very analogous to the ``chord Hilbert space'' construction in the double-scaled SYK (DSSYK) model. We study this semiclassical description in detail  for ensembles where the spectral density in the $D\to \infty$ limit is even, compactly supported on an interval and has square root edges. With a few more conditions on the analytic structure of the spectral density, we observe that the semiclassical Hamiltonian has the same asymptotic behavior at large Krylov depth as that of DSSYK, with the corresponding parameter $\mathfrak{q}=e^{-\lambda}$ being related to the location of the nearest zero of the spectral density away from the spectral cut. Furthermore, in a large class of models corresponding to ultraviolet deformations of the DSSYK spectral density, i.e., where the spectrum in the UV is modified while leaving the near-edge behavior unchanged, we show that the semiclassical effective Hamiltonian in the Krylov basis reduces to the Liouville Hamiltonian of JT gravity in a low-energy, continuum limit. This suggests that our semiclassical Hilbert space should be interpreted as the bulk Hilbert space of a dual gravity description.      } 

\begin{document}

\maketitle

\parskip=10pt

\section{Introduction}
In recent years, the Sachdev-Ye-Kitaev (SYK) model \cite{Sachdev_1993, Kitaev1, Kitaev2, Maldacena:2016hyu, Kitaev:2017awl, Sarosi:2017ykf} has revolutionized our understanding of the AdS/CFT correspondence in lower dimensions. A central insight to come out of this body of work was that instead of looking at one specific instance of a chaotic Hamiltonian, mileage can be gained from studying the average over an ensemble of chaotic Hamiltonians, and that ensemble-averaged quantities tend to admit semi-classical descriptions in terms of emergent collective variables in the large-$N$ limit. What is more, the emergent semi-classical description matches in some cases with the dynamics of a two-dimensional model of gravity called JT gravity, with negative cosmological constant \cite{Almheiri:2014cka, Maldacena:2016upp}. This insight was made extremely sharp by Saad, Shenker and Stanford \cite{Saad:2019lba} who proposed a precise AdS/CFT-like duality between JT gravity with its genus expansion on the one hand, and a random matrix theory in the double-scaling limit with its t'Hooft expansion on the other hand. The random matrix in question was identified with the Hamiltonian of the putative boundary quantum mechanics, and the natural observables that the duality allows one to match correspond to ensemble averaged products of the thermal partition function: 
$$ \langle Z(\beta_1)Z(\beta_2)\cdots Z(\beta_m)\rangle,\;\;\cdots\;\;Z(\beta) = \text{Tr}\,e^{-\beta H}.$$ 

More recently, the double-scaled SYK (DSSYK) model has led to a lot of interesting progress in this field, particularly at the disc level (i.e., the leading contribution as $N \to \infty$). The DSSYK model was originally solved exactly in the large-$N$, double scaling limit by Berkooz et al \cite{Berkooz:2018qkz, Berkooz:2018jqr} (building on earlier work in \cite{Erd_s_2014, Cotler:2016fpe}). It has a large-$N$ spectral density that is compactly supported on an interval. The model was solved using a ``chord diagram'' picture, where by chords one means Hamiltonian lines which connect up with each other pairwise on account of the Gaussian random averaging over the couplings. This leads to a \emph{chord Hilbert space}, spanned by an orthonormal basis of states corresponding to chord number eigenstates. Perhaps most interestingly, it was suggested in \cite{Berkooz:2018qkz, Lin:2022rbf} that the chord Hilbert space should be interpreted as a bulk gravitational Hilbert space dual to DSSYK in the large-$N$ limit (see also \cite{Blommaert:2024ymv, Blommaert:2025avl} for a realization of the chord-Hamiltonian in sine-dilaton gravity). Indeed, in a further scaling limit (called the triple-scaling limit in this context) which corresponds to zooming in near a large value of the chord number and correspondingly zooming in near the edge of the spectral density in energy, the chord number becomes a continuous, real-valued variable and the Hamiltonian of the model written in the chord basis reduces to the Liouville Hamiltonian of JT gravity for a particular gravitational mode, namely the length of the Einstein-Rosen wormhole connecting the two asymptotic boundaries of the black hole \cite{Bagrets:2016cdf, Harlow:2018tqv, Yang:2018gdb}. In this sense, DSSYK at large-$N$ provides a UV completion of the spectral density of JT gravity, and the chord number basis provides a ``discretization'' of the geometric length basis of JT gravity. The DSSYK model is rich enough to also describe matter operators and their correlation functions \cite{Berkooz:2024lgq} resulting in a non-trivial bulk Hilbert space \cite{Lin:2022rbf}. 

While many of these results were originally derived in the context of the double-scaling limit of the SYK model, it was later pointed out in \cite{Jafferis:2022wez} that they could as well be interpreted in terms of a unitarily invariant random matrix theory -- sometimes referred to as the ``ETH matrix model'' -- with the potential engineered so as to reproduce the correct leading large-$N$ spectral density. The matrix model perspective has the advantage that it admits a standard t'-Hooft expansion in the size of the matrices, which fits nicely with the SSS matrix model/JT gravity duality, and has since been fruitfully used to make important progress on higher genus corrections in the ETH matrix model \cite{Okuyama:2023aup, Okuyama:2023kdo, Okuyama:2024eyf, Miyaji:2025ucp}. It is interesting to ask how the chord Hilbert space of DSSYK is to be interpreted from the matrix model point of view. A hint towards this was originally suggested in \cite{Lin:2022rbf} and worked out in detail in \cite{Rabinovici:2023yex, Ambrosini:2024sre} -- in the absence of matter chords, the chord-number basis in DSSYK agrees with the \emph{Krylov basis} with respect to the maximally-entangled state on two copies of the system. The Krylov basis\footnote{Krylov/recursion methods have been studied in the many-body context for many years now, see \cite{viswanath1994recursion}. It is only relatively recently -- starting from \cite{Parker:2018yvk} -- that they have been appreciated in the quantum chaos and quantum gravity context.} \cite{Parker:2018yvk, Balasubramanian:2022tpr} (see also \cite{Baiguera:2025dkc, Nandy:2024evd, Rabinovici:2025otw} for recent reviews and further references) is the orthonormal basis obtained from the set of states $\{|\Omega\rangle, H|\Omega\rangle, H^2 |\Omega\rangle,\cdots\}$ after Gram-Schmidt, where $|\Omega\rangle$ is the maximally entangled state on two copies of the system. It is thus natural to expect that from the matrix model viewpoint, we should interpret the chord Hilbert space as the subspace spanned by the Krylov basis vectors up to fixed $O(1)$ Krylov depth in the limit where the size of the matrices goes to infinity. In the presence of matter, the story is more subtle, although it may still be possible to define some analog of the Krylov basis in that case as well (see \cite{Lin:2022rbf}). 

In the present paper, we will apply the lessons learnt in the DSSYK (and the ETH matrix model) context to more general one-cut random matrix models. In particular, we will study time-evolved thermo-field double (TFD) states of the form:
\beq \label{eq:TFD0}
|\Psi(z)\rangle = e^{-z\,H}\,|\Omega\rangle,\;\;\;\cdots\;\;(z= \tau + it),
\eeq 
where $H$ is the one-sided Hamiltonian (acting, say, on the right factor) drawn randomly from a unitary ensemble over $D\times D$ Hermitian matrices.\footnote{We will use the notation $D$ for the dimension of the one-sided Hilbert space, or equivalently the size of the Hamiltonian matrix. In DSSYK, $D=2^{N/2}$.} The GKPW dictionary in AdS/CFT \cite{Witten:1998qj, Gubser:1998bc} works at the level of Euclidean path integral states \cite{Marolf:2017kvq} and the family of states $\Psi(z)$ in equation \eqref{eq:TFD0} constitutes the simplest prototypical example of such states. We will argue that in a general, one-cut random matrix model, a semi-classical ``bulk'' Hilbert space description for such TFD states emerges at large $D$ in the \emph{Krylov basis}, analogous to the chord-Hilbert space picture in the DSSYK model. The essential idea is that a randomly drawn Hamiltonian can be put in a tri-diagonal form by writing it in terms of its Krylov basis with respect to the maximally entangled state $\Omega.$ The coefficients which enter this tri-diagonal matrix are called \emph{Lanczos coefficients}, and at least at fixed Krylov depth $n$ in the $D\to \infty$ limit, these coefficients can be replaced with their ensemble averaged values (i.e., their fluctuations are suppressed by $\frac{1}{D}$) which turn out to be smooth functions $n$. Thus, by restricting to the subspace spanned by Krylov basis vectors of fixed $O(1)$ depth in the $D\to \infty$ limit, one obtains a ``small Hilbert space'' with a local (i.e., tri-diagonal) effective Hamiltonian acting upon it. This small Hilbert space is of particular importance to us because the wavefunction $\langle n| \Psi(z)\rangle$ of the TFD state in equation \eqref{eq:TFD0} remains effectively \emph{confined} to the small Hilbert space for any $O(1)$ value of $z$ in the $D \to \infty$ limit. Thus, the small Hilbert space furnishes a smooth, local and semi-classical\footnote{Here by semi-classical we mean that the wavefunction is effectively spread over an $O(1)$ subspace, as opposed to an exponentially large number of states in the $D\to \infty$ limit.} description of time-evolved TFD states in the $D\to \infty$ limit, much like the chord Hilbert space of DSSYK. It is natural to wonder whether this small Hilbert space should be interpreted as the analog the ``bulk'' Hilbert space, obtained from the canonical quantization of some putative gravitational theory. Our goal here is to gather evidence towards this interpretation. 

\subsection{Main results}
We will study the semiclassical bulk Hilbert space and the corresponding Krylov effective Hamiltonian for unitarily invariant random matrix models where the ensemble-averaged spectral density $\omega(E)$ in the $D\to \infty$ limit is even under $E\to -E$, compactly supported on the interval $E\in [-1,1]$ and has \emph{square root} edges. Our main results are as follows: 
\begin{enumerate}
\item The effective bulk Hamiltonian can be extracted efficiently by using the mathematical machinery of orthogonal polynomials with respect to the spectral density $\omega(E)$. Standard results on the asymptotics of orthogonal polynomials \cite{Fokas:1991za, bleher, Deift_et_al, deift1993steepest, kuijlaars2004riemann, kuijlaars2003riemann} then imply that the corresponding (ensemble-averaged) Lanczos coefficients in the semiclassical effective Hamiltonian all approach a constant, $n$-independent value at large Krylov depth $n$.\footnote{Note that by large $n$ we mean the limit $D \to \infty$ followed by $n\to \infty$, in that order.} Thus, at large Krylov depth, the effective Hamiltonian becomes that of a free particle hopping on the one-dimensional Krylov lattice:
\beq 
-H_{\text{eff}} \sim \frac{1}{2}\sum_n \left(|n\rangle\langle n+1| + |n+1\rangle\langle n|\right),\;\;\;\cdots\;\;\;(\text{at large}\;n).
\eeq 
\item With some mild assumptions on the analytic structure of the spectral density, we observe that the approach of the Lanczos coefficients to this asymptotic value is exponential in $n$, very much like in the chord Hamiltonian of the DSSYK model. In DSSYK, the rate of asymptotic approach is controlled by the parameter $\mathfrak{q}=e^{-\lambda}$; in our context, the role of this parameter is played by the location of the nearest pair of zeros of the spectral density away from the interval $[-1,1]$. 
\item We argue that in a large class of matrix models which are ultraviolet (UV) deformations of DSSYK, i.e., where the spectral density agrees with that of DSSYK near the edges but is deformed in the UV, the effective ``bulk'' Hamiltonian in the Krylov basis reduces to the Liouville Hamiltonian of JT gravity \cite{Bagrets:2016cdf, Harlow:2018tqv} in a low-energy, continuum limit (where $\lambda \to 0$ and we zoom-in near one edge of the spectrum). We demonstrate this explicitly in a large class of models corresponding to \emph{rational Christoffel deformations} of the DSSYK spectral density, i.e.,
\beq 
\omega(E) = \mathcal{N}\frac{\prod_{i=1}^P(p_i^2 - E^2)}{\prod_{j=1}^Q(q_j^2 - E^2)}\,\omega_{\text{DSSYK}}(E),
\eeq 
with $p_i,q_j$ real and greater than 1, and $\mathcal{N}$ being a normalization constant.\footnote{Note that the DSSYK spectral density can be realized as the large-$D$ spectral density of the ETH matrix model \cite{Jafferis:2022wez}. Christoffel deformations of the DSSYK spectral density can presumably be translated to deformations of the potential of the ETH matrix model. It would also be interesting to understand the origin of these deformations from the perspective of UV deformations of the SYK model \cite{Jiang:2019pam, Anninos:2022qgy, Berkooz:2024ifu, Berkooz:2024ofm, Aguilar-Gutierrez:2026nmd, Aguilar-Gutierrez:2026ogo}.} As mentioned before, it is known that in the DSSYK/ETH matrix model, the effective Hamiltonian in the chord-number/Krylov basis provides a discretization of JT gravity away from the low-energy limit \cite{Berkooz:2018qkz, Lin:2022rbf, Rabinovici:2023yex}. Our construction shows that this result is robust with respect to choice of ultraviolet completion by generalizing it to a much larger set of matrix models which agree with DSSYK near the edges (and in particular, contain the $\sinh(\sqrt{E})$ behavior of JT near the edges) but differ significantly from DSSYK in the ultraviolet.  
\end{enumerate}
So to summarize, at least for the dynamics of TFD states, the construction of the semi-classical ``bulk'' Hilbert space associated to the discretized wormhole length in the DSSYK model generalizes -- via the identification of chord-number states with Krylov basis states -- to a large class of random matrix models. The bulk theory dual to the DSSYK model is apparently a theory of discrete wormhole lengths, and our construction provides evidence suggesting an analogous discrete gravity dual for more general matrix models away from the low-energy limit in the $D\to \infty$ limit.    

The rest of this paper is organized as follows: in section \ref{sec:prelim}, we review some background material on the DSSYK model, Krylov basis etc., and show why the Krylov basis gives a good semiclassical effective description in the $D\to \infty$ limit. In section \ref{sec:largen}, we study this semiclassical description in the context of random matrix theory, and discuss how the mathematical machinery of orthogonal polynomials can be used to efficiently obtain the effective Hamiltonian. In section \ref{sec:OP}, we briefly review some standard facts about orthogonal polynomials and discuss their implications for our bulk effective theory. In section \ref{sec:continuum}, we discuss the low-energy continuum limit and show the emergence of JT gravity from our effective Hamiltonian in a large class of matrix models. We end with a brief summary in section \ref{sec:discussion}. Further mathematical details on the theory of orthogonal polynomials are relegated to the appendices. 

\section{Preliminaries}\label{sec:prelim}
\subsection{The DSSYK model}
The SYK model consists of $N$ Majorana fermions with the Hamiltonian given by:
\beq
H=i^{p/2} \sum_{i_1<....<i_p} J_{i_1i_2....i_p} \ \psi_{i_1} \psi_{i_2} ....\psi_{i_p},\;\;\;\left\{\psi_i,\psi_j\right\} = 2\delta_{ij},
\eeq
where $J_{i_1i_2....i_p}$ are i.i.d gaussian random variables with zero mean, and variance given by
\beq 
\langle J_{i_1\cdots i_p}^2\rangle =\left(\begin{matrix}N\\p\end{matrix}\right)^{-1}\mathcal{J}^2.
\eeq 
The double scaling limit is achieved by taking $N,p \to \infty$ while keeping $\lambda=\frac{2p^2}{N}$ fixed \cite{Berkooz:2024lgq}. The model is exactly solvable in the double scaling limit; for instance, the thermal partition function can be calculated exactly in this limit using the technique of ``chord diagrams''. In calculating these quantities using chord diagrams, one encounters an \emph{effective Hilbert space} of chords, $\mathcal{H}_{\text{eff}}$. The effective Hilbert space is spanned by orthonormal states $\{|n\rangle\}_{n=0}^{\infty}$ which are to be interpreted as chord number states, where the $|0\rangle$ state is the (normalized) maximally entangled state. The Hamiltonian can then be expressed in terms of these chord number states and takes the following tri-diagonal form:
\beq\label{eq:Heff1}
-H_{\text{eff}}\ket{n}= b_n \ket{n-1} +  b_{n+1} \ket{n+1},
\eeq
where the coefficients $b_n$ are given by
\beq \label{eq:Heff2}
b_n=\sqrt{\frac{1-\bq^n}{1-\bq}},\;\;\;\; \bq := e^{-\lambda}.
\eeq 
Using this effective Hamiltonian, one can calculate the partition function:
\beq 
Z(\beta) = \text{Tr}\,e^{-\beta H} = \langle 0|e^{-\beta H_{\text{eff}}}|0\rangle,
\eeq 
for any fixed $O(1)$ value of $\beta$ in the $N\to \infty$ limit, where the trace appearing in the above expression is a ``renormalized trace'' defined as $$\mathrm{Tr}(\mathbb{1}) = 1,$$
and is related to the standard trace by an overall factor of $2^{N/2}$. Performing an inverse Laplace transform with respect to the temperature, one obtains the spectral density $\rho(E)$. In order to state the result, it is convenient to rewrite the energy $E$ in terms of an angle variable $\theta \in (0, \pi)$ as 
\beq \label{eq:energy}
E(\theta)=-\frac{2 \ \text{cos}(\theta)}{\sqrt{1-\bq}}.
\eeq 
Then, the spectral density is given by:
\beq 
dE\,\rho(E) = d\theta\,\mu(\theta),
\eeq 
where
\beq \label{eq:density}
\mu(\theta)=\frac{1}{2\pi}(\bq;\bq)_{\infty} (e^{2i\theta};\bq)_{\infty} (e^{-2i\theta};\bq)_{\infty},
\eeq
and the $q$-Pochhammer symbol $(a;\bq)_{n}$ is defined as
\beq 
(a;\bq)_{n} := \prod_{k=0}^{n-1}(1- a\,\bq^k).
\eeq 
 Note that the spectrum is compactly supported in the range $-\frac{2}{\sqrt{1-\bq}} \leq E \leq \frac{2}{\sqrt{1-\bq}}$. It is convenient to scale out the overall factor of $\frac{2}{\sqrt{1-\bq}}$ in equation \eqref{eq:energy} by defining a rescaled Hamiltonian:
\beqn \label{eq:rescaledH}
-H_{\text{eff}}^{(\text{new})} &=& -\frac{1}{2}\sqrt{1-\bq}\,H_{\text{eff}}^{(\text{old})}\nonumber\\
&=& \frac{1}{2}\sqrt{1-\bq^n}|n-1\rangle + \frac{1}{2}\sqrt{1-\bq^{n+1}}|n+1\rangle,
\eeqn  
so that the rescaled spectrum is supported in the range $-1 \leq E \leq 1$.   For the rest of this paper, we will work with this convention. 

\subsection{JT gravity in the low-energy limit}
A striking feature of the semi-classical description of the DSSYK model in terms of chord-number states is that this description reduces to JT gravity in the so called \emph{triple-scaling limit}, which involves taking the $\lambda\to 0$ limit and zooming in near the edge of the spectrum (i.e., taking a low-energy limit) \cite{Berkooz:2018qkz, Lin:2022rbf}. To see how this happens, first consider the spectral density in the $\lm\to 0$ limit. Expanding $\rho(E)$ near the edge by defining 
\beq \label{eq:Escaling}
E = -1 + \lambda^2 \widetilde{E},
\eeq
and sending $\lambda \to 0$ while holding $\widetilde{E}$ fixed, one finds \cite{Berkooz:2018qkz, Berkooz:2024lgq}
\beq
 \rho(E)\,dE
= C(\lambda)\ \text{sinh}\left(2\pi \sqrt{2\widetilde{E}} \right) d\widetilde{E},
\eeq
where $C$ 
is a constant independent of the rescaled energy, but vanishing in the $\lambda \to 0$ limit. This was expected -- the spectral density of the original DSSYK model is normalized to have unit integral over the range $-1\leq E\leq 1$, but now we are zooming in close to the edge, and only a vanishingly small part of the total eigenvalue content survives near the edge. Nevertheless, the zoomed-in spectral density has the right dependence on the rescaled energy to match with the spectral density of JT gravity (which, of course, is unbounded above and thus has no natural normalization). In this way, the spectral density of DSSYK approaches that of JT gravity in the low-energy limit. 

Next, let's look at the effective Hamiltonian $H_{\text{eff}}$ in equation \eqref{eq:rescaledH}. The Lanczos coefficients are given by
\beq \label{eq:DSSYKLC}
b_n=\frac12\sqrt{1-e^{-\lambda n}}.
\eeq
To get to the continuum limit, we zoom in to a region at large $n$ by defining: 
\beq \label{eq:scaled variables}
\ell = \lambda(n-n_*),\;\;\;n_* = \frac{1}{\lambda}\log(\frac{1}{\lambda^2}),
\eeq 
and taking $\lambda \to 0$ while holding $\ell$ fixed. Furthermore, in this limit we must also scale the inverse temperature as:
\beq 
\beta = \frac{1}{\lambda^2}\tbeta,
\eeq 
in accordance with the corresponding energy rescaling given in equation \eqref{eq:Escaling}. The new variable $\ell$ should be regarded as a continuous ``renormalized length'' variable. In this sector of the Hilbert space, the Lanczos coeffcients are given by:
\beq \label{eq:LanczosDSSYK}
b_n =\frac{1}{2}\sqrt{1-\lambda^2e^{-\ell}} \simeq \frac{1}{2} - \frac{1}{4}\lambda^2 e^{-\ell}+O(\lambda^4).
\eeq
In order to derive the effective Hamiltonian in the scaling limit, we look at the Schrodinger equation for the wavefunction in the Krylov/chord-number basis:
\beq 
-\partial_{\beta} \langle n| \Psi(\beta)\rangle = \frac{1}{2} \langle n | H_{\text{eff}} |\Psi(\beta)\rangle.
\eeq 
Assuming the wavefunction is a sufficiently smooth function of $n$ such that it descends to a function of $\ell$ in the scaling limit, it is natural to define:\footnote{With this redefinition, the inner product in the $n$-basis translates to the standard $L^2$ norm on functions of $\ell$.} 
\beq 
\psi_{\tbeta}(\ell) : = \frac{1}{\sqrt{\lambda}} \langle n| \Psi(\beta)\rangle.
\eeq 
Using the form of $H_{\text{eff}}$ in equation \eqref{eq:rescaledH}, we get a Schrodinger equation for the continuum wavefunction $\psi_{\tbeta}(\ell)$:
\beq 
-\partial_{\tbeta} \psi_{\tbeta}(\ell) = \frac{1}{2} H_{\text{cont.}}\psi_{\tbeta}(\ell),
\eeq 
where
\beq
H_{\text{cont.}}=-\frac{1}{\lambda^2}+\frac{1}{2} \left(-\partial_{\ell}^2+e^{-\ell}\right).
\eeq
This matches the JT gravity Hamiltonian where $\ell$ has the interpretation of the renormalized length of the wormhole connecting the two sides of the eternal black hole.

\subsection{Krylov basis and emergent semi-classicality}
The discussion so far has been about a specific theory, namely the DSSYK model. Our goal in this paper is to generalize these ideas to a large class of one-cut random matrix models. To this end, let us now momentarily change gears, and consider a general quantum system with a Hilbert space $\mathcal{H}$ of dimension $D$; we will return to DSSYK at the end of the section. For the purposes of studying thermal states (potentially also evolved in Lorentzian time), it is actually natural to consider the doubled Hilbert space $\mathcal{H}\otimes \mathcal{H}^\star$, which can also be thought of as the space of operators acting on $\mathcal{H}$. We take the initial state to be the canonical maximally entangled state $|\Omega\rangle$ proportional to the identity operator:
\beq 
|\Omega\rangle = \frac{1}{\sqrt{D}}\sum_i |i\rangle \otimes |i^\star\rangle,
\eeq 
where $\{|i\rangle\}$ is some orthonormal basis for $\mathcal{H}$, and correspondingly $\{|i^\star\rangle\}$ is the conjugate basis. We are interested in constructing a semi-classical description of thermo-field double states of the form:
\beq \label{eq:TFD}
|\Psi(z)\rangle = e^{-z\,H}|\Omega\rangle,
\eeq 
where $z= (\frac{\beta}{2}+it)$ and  $H$ is the one-sided Hamiltonian acting on the right Hilbert space, which from the operator point of view corresponds to right multiplication. Indeed, the thermal partition function is simply the norm of the above state:
\beq 
Z(\beta) = \mathrm{Tr}\,e^{-\beta H} = \langle \Psi(z)|\Psi(z)\rangle,
\eeq 
where once again the trace appearing above is the renormalized trace. For the initial state $|\Omega\rangle$ and Hamiltonian $H$, the \emph{Krylov basis} is defined \cite{viswanath1994recursion, Parker:2018yvk, Balasubramanian:2022tpr} by starting with the ordered set of states 
\beq 
|\Omega\rangle,\, H |\Omega\rangle,\;H^2|\Omega\rangle,\cdots,
\eeq 
and orthonormalizing them using the Gram-Schmidt procedure. The resulting set of orthonormal states form a basis for the \emph{Krylov subspace}, i.e., the subspace of the Hilbert space spanned by the Euclidean/Lorentzian time evolution of the initial state $\Omega$. For sufficiently chaotic Hamiltonians, we expect the Krylov subspace to be a $D$-dimensional subspace of $\mathcal{H}\otimes \mathcal{H}^\star$ spanned by states with the ``same energy'' on both sides. It is a  general fact that the Hamiltonian, when expressed in the Krylov basis, takes a tri-diagonal form:
\beq 
-H |n\rangle = -a_n |n\rangle + b_n|n-1\rangle + b_{n+1}|n+1\rangle,
\eeq
where the coefficients $\{a_n, b_n\}$ are called \emph{Lanczos} or \emph{recursion} coefficients. Note that the Krylov basis can be redefined using the ``gauge freedom'' to multiply individual basis vectors by local phases:
\beq \label{eq:gauge}
|n\rangle' = e^{i\phi_n}|n\rangle. 
\eeq 
Under such a redefinition, we get
\beq 
a_n'=a_n,\;\;b_n' = e^{i(\phi_n-\phi_{n-1})}b_n.
\eeq 
It is natural to go to a gauge where the coefficients $\{b_n\}$ are all real, but their signs are not fixed. Later, we will fix a particular sign convention which is natural from the point of view of taking the low-energy limit.  

The point of introducing the Krylov basis is that the states $|\Psi(z)\rangle$ in equation \eqref{eq:TFD} are naturally contained in the Krylov subspace with respect to the one-sided Hamiltonian, and thus can naturally be expressed in the Krylov basis. It has been argued that the Krylov basis minimizes the growth of the complexity of classical simulation under Hamiltonian evolution $e^{-zH}$, and is thus an ideal candidate for a dual semiclassical description of the family of states $\Psi(z)$ \cite{Balasubramanian:2022tpr, Basu:2024tgg, Basu:2025mmm, Balasubramanian:2026klv}. This description becomes particularly useful in the limit where the Hilbert space dimension goes to infinity with $z$ fixed to be $O(1)$. In this limit, it is natural to consider the \emph{effective} Hilbert space spanned by Krylov basis states with $n$ fixed as $D\to \infty$:
\beq 
\mathcal{H}_{\text{eff}} = \text{span}\,\{|n\rangle\},\;\;\;(n\;\text{fixed},\, D\to \infty).
\eeq 
We can think of this as the analog of the ``small'' Hilbert space or ``code subspace'' defined in \cite{Papadodimas:2013wnh, Papadodimas:2013jku, Almheiri:2014lwa}, but now for the Hamiltonian operator. The reason the description in terms of the Krylov basis becomes useful is encapsulated by the following three properties of this effective description in general quantum systems with a smooth, compactly supported spectral density in the large $D$ limit: 
\begin{enumerate} 
     \item \textbf{Locality}: The \emph{effective Hamiltonian}, i.e., the restriction of $H$ to $\mathcal{H}_{\text{eff}}$ is tri-diagonal, and hence has emergent locality on the Krylov chain. 
      \item \textbf{Smoothness}: The Lanczos coefficients $\{a_n,b_n\}$ of $H_{\text{eff}}$ approach smooth functions of $n$ in the $D\to \infty$ limit, with any erratic behavior being suppressed in $\frac{1}{D}$.
     \item \textbf{Confinement}: The wavefunction of a TFD state $\Psi(z)$ remains \emph{confined} within the effective Hilbert space $\mathcal{H}_{\text{eff}}$ in the $D\to\infty$ limit for any $O(1)$ value of $z$.
\end{enumerate}  

 The first point is evident. To see smoothness, note that for a given Hamiltonian $H$, the Lanczos coefficients can be extracted from the moments of $H$ as follows:    
\beqn
\mu_1 &=& \langle 0 | H | 0\rangle = a_0,\nonumber\\
\mu_2 &=& \langle 0 | H^2 |0\rangle = a_0^2 + b_1^2,\nonumber\\
\mu_3 &=& \langle 0 | H^3|0\rangle = a_0^3 + 2 a_0 b_1^2 + b^2_1a_1,
\eeqn
and so on, where $\langle 0 | H^n |0\rangle = \text{Tr}(H^n)$, and by trace we mean the one-sided, normalized trace. Furthermore, if we are only interested in the coefficients $\{a_n,b_n\}$ for fixed $n$ in the $D\to \infty$ limit, then these recursion coefficients are given in terms of the moments $\{\mu_n\}$ with fixed $n$ in the $D\to \infty$ limit. For a system with a smooth large-$D$ spectral density, the moments are given by smooth functions:
\beq 
\mu_n = \int dE\,\rho(E)\,E^n,
\eeq 
 from which follows the smoothness of the Lanczos coefficients. Of course, the spectral density at any finite $D$ is not a smooth function, but one way to define a smooth spectral density is to compute the inverse Laplace transform of the partition function $Z_{\infty}(\beta) = \lim_{D\to \infty} Z_D(\beta)$.\footnote{Had we computed the inverse Laplace transform before taking the $D\to \infty$ limit, we would obtain a sum over delta functions.} Equivalently, we could coarse-grain over energy windows at energy scales $\frac{1}{D}\ll \frac{\delta E}{E} \ll 1$. In the context of random matrix theory where $H$ is drawn from an ensemble, it is also natural to define the ensemble averaged spectral density $\overline{\rho(E)}.$ In that context, a related fact is that the Lanczos coefficients for $n\sim O(1)$ are \emph{self-averaging}. This is because the moments $\mu_n$ with $n$ fixed in the $D\to \infty$ limit are self-averaging, i.e., they can be written as 
\beq 
\mu_n = \overline{\mu}_n + \delta \mu_n,
\eeq 
where $\overline{\mu}_n$ are the ensemble-averaged values, while the fluctuations $\delta \mu_n$ are $O(\frac{1}{D})$. Correspondingly, the Lanczos coefficients in the large-$D$ limit are also essentially given by their ensemble-averaged values, with fluctuations of the order $\delta a_n,\,\delta b_n \sim O(\frac{1}{D})$. Thus, in the $D\to \infty$ limit the recurssion coefficients approach their the ensemble-averaged values $\{\overline{a}_n, \overline{b}_n\}$, with erratic fluctuations of $O(\frac{1}{D})$.

 Finally, the confinement property follows from a basic result in linear algebra \cite{linalg0, linalg, linalg2}. Let $P_{\Lambda}$ be the projector onto the Krylov subspace with $n\leq \Lambda-1$, and let us assume that the Hamiltonian $H$ has a bounded spectrum such that $-E_{\text{max}}\leq  E \leq E_{\text{max}}$. We would like to estimate how much support the state $|\Psi(z)\rangle$ has outside the Krylov subspace of depth $\Lambda$. We would thus like to compute $||\left(\mathbb{1} - P_{\Lambda}\right)\Psi(z)||$, where $\| X \| = \sqrt{\langle X|X\rangle}$ is the standard Hilbert space norm. Note that
\beq 
(\mathbb{1} - P_{\Lambda}) e^{-zH}|0\rangle  =(\mathbb{1} - P_{\Lambda}) \sum_{n = \Lambda}^{\infty}\frac{(-z H)^n}{n!}|0\rangle.
\eeq
Thus,
\beqn
\| (\mathbb{1} - P_{\Lambda}) e^{-zH}|0\rangle \| &=& \|(\mathbb{1} - P_{\Lambda}) \sum_{n = \Lambda}^{\infty}\frac{(-z H)^n}{n!}|0\rangle\|\nonumber\\
&\leq & \left\|\sum_{n = \Lambda}^{\infty}\frac{(-z H)^n}{n!}|0\rangle\right\|\nonumber\\
&\leq &  \sum_{n = \Lambda}^{\infty}\left\| \frac{(-z H)^n}{n!}|0\rangle \right\|\nonumber\\
&=& \sum_{n = \Lambda}^{\infty}\frac{|z|^n}{n!} E_{\text{max}}^n.
\eeqn 
From standard bounds on the convergence of Taylor series, we thus find that
\beq 
\| \Psi(z) - P_{\Lambda} \Psi(z) \| \leq e^{\xi} \,\frac{(|z|E_{\text{max}})^{\Lambda}}{\Lambda!}\simeq e^{\xi}\,\frac{e^{-\Lambda \left(\log(\frac{\Lambda}{|z|E_{\text{max}}})-1\right)}}{\sqrt{2\pi\Lambda}},
\eeq 
 where $\xi \in (0,E_{\text{max}})$ is arbitrary, and in the last line we have used Stirling's approximation. Thus, for $\Lambda > E_{\text{max}}|z|$, the norm of the tail of the wavefunction decays super-exponentially in $\Lambda$. So, if we first send $D \to \infty$ and then $\Lambda \to \infty$ while making sure that $|z| E_{\text{max}}$ stays fixed, then the state remains effectively confined to the effective Hilbert space $\mathcal{H}_{\text{eff}}$. This is essentially the same reason why Krylov subspace techniques are useful for efficient numerical calculations of exponentials of large matrices \cite{linalg0, linalg, linalg2}.

For the reasons explained above, \emph{the Krylov basis provides a local, smooth, semiclassical effective description of (time-evolved) thermal states $\Psi(z)$ for quantum systems with a smooth and compactly-supported spectral density in the $D\to \infty$ limit}. Here, by semiclassical we mean that the state $\Psi(z)$ is not spread over a subspace whose dimension is scaling in some way with $D$ in the $D\to \infty$ limit. This is an extremely important feature of the effective description in the Krylov basis -- the time evolution of a generic initial state by a random Hermitian Hamiltonian generically becomes exponentially hard to simulate classically within an $O(1)$ amount of time evolution. It is only in an extremely fine-tuned choice of basis, such as the Krylov basis, that the dynamics becomes classically simulable \cite{Basu:2024tgg, Basu:2025mmm, Balasubramanian:2026klv}.  

Coming back to the DSSYK model, it was pointed out in \cite{Lin:2022rbf} (see also \cite{Rabinovici:2023yex, Ambrosini:2024sre}) that the chord length basis is identical to the Krylov basis with respect to the maximally entangled state. Thus, it appears that the DSSYK Hamiltonian written in the chord length basis is precisely the emergent, semiclassical description in Krylov space that we have discussed above. In the rest of this paper, we will take this as our guiding principle in extending the results of DSSYK to a larger class of theories.

\section{Semi-classical dynamics in random matrix theory}\label{sec:largen}
In this section, we will study the semiclassical effective description which emerges in the Krylov basis in a specific class of theories, namely one-cut random matrix models. Our goal is to explain how the effective Hamiltonian $H_{\text{eff}}$ can be extracted from the spectral density of the model in the $D\to \infty$ limit. 
\subsection{Setup}
Consider a theory defined on a Hilbert space of dimension $D$, where the Hamiltonian is drawn randomly from a unitary-invariant ensemble over $D\times D$ Hermitian matrices with a probability measure of the form:
\beq 
\mu(H) = \frac{1}{Z}e^{-D\,\text{Tr}\,V(H)},
\eeq 
for some potential function $V(H)$. We will particularly be interested in potentials such that the corresponding ensemble-averaged spectral density of eigenvalues:
\beq 
\om(E) := \overline{\rho(E)},\;\;\rho(E)= \frac{1}{D}\sum_{i=1}^D \delta (E- E_i),
\eeq 
in the $D\to\ \infty$ limit is compactly supported on an interval. For simplicity, we will focus on the case where the potential is an even function of $H$, and correspondingly the spectral density is symmetric under $E\to -E$. We can then appropriately rescale the couplings in the potential to make sure that the spectrum is supported on the interval 
$$E \in [-1,1].$$ We will further restrict our attention to the class of random matrix theories where the averaged spectral density in the $D\to \infty$ limit takes the form:
\beq 
\omega(E)=(1-E^2)^{1/2} h(E),
\eeq 
where $h(E)$ is some \emph{positive}, real-analytic function over the interval $[-1,1]$. The choice of the square-root edges greatly simplifies the analysis, and turns out to be natural from the gravity point of view.

As before, we will be interested in thermo-field double states $\Psi\in\mathcal{H}\otimes \mathcal{H}^\star$ of the form 
\beq 
|\Psi(z) \rangle = e^{-z H}|\Omega\rangle,\;\;\;(z= \tau + it),
\eeq 
where $z=(\frac{\beta}{2} + it)$ with $\beta \geq 0$, $H$ is the one-sided Hamiltonian (which we think of acting on the right subsystem) and $\Omega$ is the maximally entangled state. We wish to construct a semi-classical description for such states using the Krylov basis. Given a specific Hamiltonian from the random matrix ensemble, we can construct the corresponding Krylov basis $\{|n\rangle\}_{n=0}^{D-1}$ with respect to the maximally entangled state $\Omega$ by applying the Gram-Schmidt process to the ordered set of states $\{H^n|\Omega\rangle\}_{n=0}^{D-1}$.\footnote{In random matrix theory, we expect the dynamics to be ergodic enough so that the Krylov subspace should be of dimension $D$. Note that we can at best get a set of $D$ basis states since the Krylov subspace is spanned by states of the form $|E_i\rangle \otimes |E_i^\star\rangle$, i.e., with $E_L=E_R$.} As before, the Hamiltonian expressed in the Krylov basis takes a tri-diagonal form:
\beq\label{eq:HK2}
-H\ket{n}= -a_n\ket{n}+b_n\ket{n-1}+b_{n+1}\ket{n+1}.
\eeq 
Using the gauge freedom explained in equation \eqref{eq:gauge}, we will henceforth fix the phases of the Krylov basis vectors so as to make the coefficients satisfy $b_n \geq 0$ for all $n$. This choice for the signs turns out to be a convenient one in taking the low-energy, continuum limit, so that in this gauge, low-energy states also end up having slowly-varying wavefunctions in the Krylov basis. 

As discussed previously, the Lanczos coeffcients in the limit $D\to \infty$ with $n$ fixed are essentially given by their ensemble averaged values, with fluctuations of the order $\delta a_n,\;\delta b_n \sim O(\frac{1}{D})$. Thus, at large $D$, we obtain an \emph{effective Hamiltonian} acting on a semi-classical Hilbert space comprising of the subspace $\mathcal{H}_{\text{eff}} = \text{span}\{|n\rangle,\; n\; \text{fixed}, D\to \infty\}$:
\beq \label{eq:Heff3}
-H_{\text{eff}} |n\rangle = -\overline{a}_n |n\rangle + \overline{b}_{n}|n-1\rangle + \overline{b}_{n+1}|n+1\rangle.
\eeq 
In fact, for random matrix models with even potentials, ensemble averaged odd moments vanish, and this implies that $\overline{a}_n=0$ for all $n$; thus, we only need to worry about the $\overline{b}_n$. It is natural to view this restriction of the Hamiltonian to the semiclassical Hilbert space in the $D\to \infty$ limit as the analog of the ``chord basis'' description of DSSYK. Indeed, as explained in \cite{Lin:2022rbf, Rabinovici:2023yex}, the chord basis of DSSYK precisely matches with the Krylov basis, at least in the absence of matter. Our task now is to understand how to construct $H_{\text{eff}}$, or in other words, how to extract the ensemble-averaged Lanczos coefficients $\{\overline{b}_n\}$ from properties of the spectral density of the model at large $D$.

\subsection{Orthogonal polynomials}\label{sec:poly}
There is a convenient way to obtain the ensemble averaged Lanczos coefficients $\overline{b}_n$ using the method of orthogonal polynomials \cite{Kar:2021nbm, Muck:2022xfc}. To see how this works, consider the overlap
\beq 
\phi_n(E_i) = \sqrt{D}\,\langle \widehat{E}_i|n\rangle,
\eeq 
for the moment at finite $D$, where  $|\widehat{E}_i\rangle = |E_i\rangle\otimes |E_i^\star\rangle$ for an energy eigenstate $|E_i\rangle$ of the Hamiltonian. From equation \eqref{eq:HK2}, we see that 
\beq \label{eq:rec1}
-E_i \phi_n(E_i) = -a_n \phi_n(E_i) + b_{n}\phi_{n-1}(E_i) + b_{n+1}\phi_{n+1}(E_i).
\eeq 
In addition, 
\beq \label{eq:rec2}
\phi_0(E_i) = \langle \widehat{E}_i | \Omega\rangle = 1.
\eeq 
Equations \eqref{eq:rec1} and \eqref{eq:rec2} imply that $\phi_n(E_i)$ are polynomials in $E_i$. Furthermore, the states $|\widehat{E}_i\rangle$ form a complete set of states in the Krylov subspace of interest, and so inserting these in the relation $\langle n|n'\rangle=\delta_{n,n'}$, we get
\beq 
\int dE\,\rho(E)\, \phi_n^*(E)\, \phi_{n'}(E) = \delta_{n,n'},\;\;\,\rho(E) = \frac{1}{D}\sum_{i}\delta(E-E_i).
\eeq 
Taking the ensemble average of this equation in the $D\to \infty$ limit and using large-$D$ factorization, we get
\beq 
\int_{-1}^1 dE\,\omega(E)\, \psi_n^*(E)\, \psi_{n'}(E) = \delta_{n,n'},
\eeq 
where $\omega(E) = \overline{\rho(E)}$, and the ensemble-averaged polynomials $\psi_n(E)$ satisfy the Schrodinger equation with respect to the effective Hamiltonian \eqref{eq:Heff3}:
\beq \label{eq:rec3}
-E \psi_n(E) = -\overline{a}_n \psi_n(E) + \overline{b}_{n}\psi_{n-1}(E) + \overline{b}_{n+1}\psi_{n+1}(E),
\eeq 
\beq 
\psi_0(E) = 1.
\eeq 
Thus, the effective wavefunctions $\psi_n(E)$ are orthogonal polynomials (or more precisely, orthonormal polynomials) with respect to the weight function $\omega(E)$, i.e., the $D\to \infty$ limit of the spectral density of the matrix model. It is an important property of orthogonal polynomials that they satisfy a three-term recursion relation -- in our context, this is related to the tri-diagonal nature of the effective Hamiltonian. The coefficients in the recursion relation of the orthogonal polynomials are then precisely the ensemble-averaged Lanczos coefficients in the $D\to \infty$ limit. This allows us to efficiently extract the effective Hamiltonian from knowledge of the spectral density of the matrix model, via its orthogonal polynomials. Note that orthogonal polynomials are famously useful in random matrix theory \cite{Mehta, DiFrancesco:1993cyw, Eynard:2015aea}, but in the usual context, one considers orthogonal polynomials with respect to the the measure of $e^{-V}$, where $V$ is the potential that defines the random matrix ensemble. In contrast, here we are considering orthogonal polynomials with respect to the $D\to \infty$ limit of the spectral density.
\subsection{Properties of the effective theory}
It turns out that while the $\overline{b}_n$ coefficients can in general depend on the detailed form of the weight function at finite $n$, the asymptotic behavior of the $\overline{b}_n$ in the $n\to \infty$ limit is quite universal and depends only on behavior near the edges and some basic analytic properties of the weight $\omega(E)$. This allows us to say some universal things about the effective Hamiltonian:
\begin{enumerate}
\item \textbf{Free particle at large $n$}: For any one-cut matrix model with compactly supported spectral density, the coefficients $\overline{b}_n$ asymptotically approach constant values:
\beq 
\overline{b}_n \sim \frac{1}{2},
\eeq 
in the $n\to \infty$ limit \cite{kuijlaars2003riemann, kuijlaars2004riemann}. This implies that in any such random matrix theory, the effective Hamiltonian in the Krylov basis approaches that of free-particle hopping on a one-dimensional chain at large enough $n$:
\beq 
-H_{\text{eff}}\sim \frac{1}{2}\sum_n \left(|n\rangle\langle n+1| + |n+1\rangle\langle n|\right),\;\;\;\cdots \;\;\;(n\to \infty).
\eeq 
\item \textbf{Exponential potential}: Furthermore, if the spectral density has square-root edges, then the Lanczos coefficients approach their asymptotic values exponentially:
\beq \label{eq:basymp}
\overline{b}_n \sim \frac{1}{2}- \alpha \,e^{-c\,n}+\cdots,
\eeq
where the constant $c$ is related to the location of the closest zero of the spectral density away from the cut; this will be explained in more detail in section \ref{sec:OP} (see theorem \ref{thm:main}). The same exponential-approach-to-saturation for the $\overline{b}_n$ coefficients is present in the DSSYK model as well (see equation \eqref{eq:Heff2}) and plays a crucial role in the effective chord Hamiltonian reducing to the Liouville Hamiltonian of JT gravity with the exponential potential. 

\item \textbf{Emergence of JT gravity:} It is thus tempting to conjecture that for any matrix model that constitutes a UV completion of the spectral density of JT gravity -- namely, that the spectral density in the $D\to \infty$ limit is even, compactly supported and approaches $\sinh(\sqrt{2E})$ near the edges but is unconstrained away from the edges -- the semiclassical effective Hamiltonian in the Krylov basis should reduce to the Liouville Hamiltonian of JT gravity in a low-energy limit. In section \ref{sec:continuum}, we will give evidence supporting this conjecture by showing this for a large class of UV deformations of the DSSYK spectral density. This then makes it natural to interpret the Krylov effective theory as a dual gravitational description.
\end{enumerate}
\subsection{A remark about entropy}
So far, we have emphasized that given the (ensemble-averaged) spectral density $\omega(E)$ of the matrix model in the $D\to \infty$ limit, we can extract the Krylov effective Hamiltonian from it via the orthogonal polynomials of $\omega$. Interestingly, we can also go in the reverse direction -- given the effective Hamiltonian in the Krylov basis, we can first compute the (renormalized) ensemble-averaged partition function:
\beq 
\overline{Z(\beta)} = \langle 0|e^{-\beta H_{\text{eff}}}|0\rangle,
\eeq
where the left hand side is the boundary partition function while the right hand is the ``bulk'' overlap between the Hartle-Hawking states $e^{-\beta/2 H_{\text{eff}}}|0\rangle$. From here, we can obtain the large-$D$ spectral density, and hence the (renormalized) entropy $S(E)\sim \log\,\omega(E)$, by computing the inverse Laplace transform of the (renormalized) partition function:
\beq 
e^{S(E)} = \int_{-\infty}^{\infty}dt\,e^{itE}\langle 0| e^{-itH_{\text{eff}}}|0\rangle.
\eeq 
Crucially, the use of the effective description in computing these quantities is justified because the wavefunction of a TFD state at $O(1)$ temperature or time is \emph{confined} within the effective Hilbert space. But at the same time, since the Krylov effective theory is defined in the $D\to \infty$ limit with $\beta$ or $t$ held fixed to be $O(1)$, the above calculation should naturally give a smooth spectral density (and not a sum of delta functions). This explains why the semiclassical effective theory ``knows'' about the entropy $S(E)$ in the microscopic theory, but only in a coarse-grained sense; this is consistent with what we know from AdS/CFT, and can be thought of as toy model towards a logical explanation of the Euclidean path integral calculations of entropy in gravity \cite{Gibbons:1976ue, Lewkowycz:2013nqa}.

\section{Some relevant mathematical results on orthogonal polynomials}\label{sec:OP}
In this section, we will briefly review some mathematical results on orthogonal polynomials and their asymptotics that will be relevant for us.
\subsection{Large $n$ limit: Riemann Hilbert approach}
 
Here, we will review a fundamental theorem about the large-$n$ asymptotics of orthogonal polynomials for one-cut weights with square-root edges. We will largely follow the presentation in \cite{kuijlaars2003riemann, kuijlaars2004riemann}; further details will be relegated to Appendix \ref{app:asymptotics} (see also \cite{Murugan:2026rfa} for a recent paper applying similar techniques in the context of Krylov space in random matrix theory). It turns out to be convenient to work with the ``monic'' polynomials
\beq 
\pi_n(E)=\frac1{\gamma_n} \psi_n(E),
\eeq 
where $\pi_n(E)$ is defined such that the coefficient of the $E^n$ is equal to 1. Substituting this in equation \eqref{eq:rec3}, we learn that
\beq 
\bb_{n+1} =- \frac{\gamma_n}{\gamma_{n+1}},
\eeq 
and we get the following recursion relations for monic polynomials:
\beq\label{eq:rec pi}
E\, \pi_n(E)= \ba_n \pi_n(E)+\bb_n^2 \pi_{n-1}(E)+\pi_{n+1}(E).
\eeq\label{eq-Y}
A clever way to obtain the asymptotic behavior of these polynomials is to construct a corresponding Riemann-Hilbert problem. Introduce the matrix $Y(z)$ defined as:
\beq
Y(z)=\mat{\pi_n(z)}{C[\pi_n \omega](z)}{c_n \pi_{n-1}(z)}{c_n\, C[\pi_{n-1}\omega](z)},
\eeq
where $\pi_n(z)$ is the analytic continuation of $\pi_n(E)$ to the complex plane $z = (x+iy)\in\mathbb{C}$, $c_n=-2 \pi i \gamma_{n-1}^2$ and $C[f]$ is the Cauchy transform defined as:
\beq\label{eq:cauchy}
C[f](z)= \int_{-1}^1\frac{dx}{2 \pi i}\, \frac{f(x)}{x-z}.
\eeq
Strictly speaking, we should be labeling the matrix $Y$ with the subscript $n$, but we will avoid doing this in order to simplify notation. The matrix $Y$ constructed out of the monic orthogonal polynomials for the weight $\omega(x)$ has the following properties in the complex $z$ plane:
\begin{enumerate}
    \item $Y(z)$ is analytic in the complex $z$ plane everywhere except on $[-1,1]$ interval.
    \item $Y(z) \to (\mathbb{1}+O(1/z))\cdot \mat{z^n}{0}{0}{z^{-n}}$ as $z \to \infty$.
    \item It has following jump matrix across $[-1,1]$:
    \beq
    Y_+(x)=Y_-(x) \mat{1}{\omega(x)}{0}{1},
    \eeq
    where $Y_\pm = \lim_{\eps \to 0} Y(x\pm i\eps)$.
    \item As $z \to \pm1$, we get:
    \beqn
    Y(z) &\to&  O\mat{1}{1}{1}{1} .
    \eeqn
\end{enumerate}
It turns out that $Y(z)$ is the unique complex-valued matrix which satisfies the above properties (see appendix \ref{app:unique}). These set of conditions on the analytic and asymptotic behavior of $Y$ are said to constitute a \emph{Riemann-Hilbert problem} (RHP). The point is that given a weight function, we can construct the corresponding RHP, from which we can extract the  corresponding orthogonal polynomials. But the RHP can be made more amenable to a large-$n$ asymptotic analysis by doing a series of invertible matrix transformations $Y(z) \to T(z) \to S(z) \to R(z)$, with each matrix in this chain satisfying a corresponding RHP of its own. These transformations simplify the RHP enough, so that the RHP for $R(z)$ can be solved order by order in an expansion at large $n$. After obtaining $R(z)$, we then simply use the inverse transformations to get $Y(z)$, thus obtaining the asymptotic form for $\pi_n(z)$, from which, of course, one can obtain the asymptotic formulas for the Lanczos coefficients. For the reader's convenience, we have reviewed these RHP transformations in some detail in Appendix \ref{app:asymptotics} following the presentation in \cite{kuijlaars2003riemann, kuijlaars2004riemann}. We only state the final result here:\footnote{The requirement in point number 4 that the zeros lie on the real line can be relaxed without too much difficulty.}

\begin{theorem}[Large-$n$ asymptotics] \label{thm:main}
Consider a unitarily invariant random matrix model such that the spectral density $\om(E)$ in the $D\to \infty$ limit has the following properties:
\begin{enumerate}
    \item $\om(E)$ is compactly supported on the interval $\Sigma_0 = [-1,1]$.
    \item $\om(-E) = \om(E).$
    \item $\om$ has square-root edges, i.e., we can write 
    $$\om(E)=\sqrt{1-E^2} \ h(E),$$
    where $h(E)$ is a positive real-analytic function on $\Sigma_0$, such that $h$ can be analytically continued to a function $h(z)$ in the complex plane.
    \item  The zeros of the function $h(z)$ are all simple, come in pairs $\{z_i=\pm x_i\}$ and all lie on the real line, outside $\Sigma_0$, i.e., $x_i \in \mathbb{R} \setminus \Sigma_0$ and $x_i>1$. 
    \item As $z \to \infty$, $h(z)$ is such that the following integral vanishes:
    \beq
    \oint_{\Gamma_{\infty}} \frac{dz}{z-s} \frac{z^{-2n-1}}{h(z)} 
    \eeq
    where $\Gamma_{\infty}$ is the closed contour at infinity, for any s.\footnote{For the purposes of studying the asymptotics of Lanczos coefficients as $n\to \infty$, it suffices to impose that $\frac1{h(z)}$ should not grow super-polynomially as $z \to \infty$. We will only consider such functions $h(z)$ throughout this paper.}
\end{enumerate} 
Let $\pm x_0$ be the zeroes of $h(z)$ closest to $\Sigma_0$, where $x_0 > 1$. Then, we have the following asymptotic formula for the ensemble averaged Lanczos coefficients:
\beq \label{eq:asymp1}
\ba_n=0, \ \ \bb_n=\frac12 -\alpha \,e^{-2 c n}+\cdots,
\eeq
where,
\beq \label{eq:asymp2}
c = \log\left(x_0+\sqrt{x_0^2-1}\right), \;\;\alpha = 4 \ |g(x_0)|,
\eeq 
and
\beq \label{eq:asymp3}
g(x_0)= \frac{\ D_h(x_0)^2}{4\phi(x_0)}\operatorname{Res}_{z=+x_0} \left(\frac1{h(z)}\right) .
\eeq
\end{theorem}
The functions $D_h(z)$ and $\phi(z)$ are defined in appendix \ref{app:conventions}.

The details of the proof are given in Appendix \ref{app:asymp_lanc}, but let us make some remarks: firstly, note from equation \eqref{eq:asymp1} that the Lanczos coefficients saturate to a constant at large $n$. This is actually a general property of all one-cut random matrix models where the spectral density has compact support on an interval. Thus, at large enough $n$, the effective Hamiltonian in the Krylov basis becomes that of a free particle hopping on the Krylov chain. Secondly, the approach to saturation is exponential in $n$. The parameter $c$ that controls the rate of the leading exponential approach is given in equation \eqref{eq:asymp2} in terms of the location of the zeroes of $\omega$ closest to the cut $\Sigma_0$. The overall coefficient $\alpha$ in front of the exponential is given by the residue of $\frac{1}{h(z)}$ at $z=x_0$ times some function evaluated at $x_0$. Note that this exponential approach to saturation is exactly what we observed for the coefficients in the chord-Hamiltonian of DSSYK:
\beq\label{eq:DSSYKLC2}
\overline{b}_n^2 = \frac{1}{4}\left(1- e^{-\lm n}\right).
\eeq 
In DSSYK, this formula is exact and there are no further corrections. But in more general matrix models, the exponential correction in equation \eqref{eq:asymp1} is only the leading contribution to the exponential approach obtained from the closest pair of zeros at first order in perturbation theory; there are subleading exponentials corresponding to the other zeros of $h$ (farther away than $x_0$) which have to be accounted for consistently at higher orders in perturbation theory (see Appendix \ref{app:asymp_lanc} for details).

Let us consider some examples. For the simplest case where $h(E)=1$ corresponding to the spectral density of the Gaussian unitary ensemble (GUE), the averaged Lanczos coefficients are exactly given by $\overline{a}_n=0, \,\overline{b}_n=1/2$ in our conventions \cite{Dumitriu:2002ntg, Balasubramanian:2022dnj}. This matches with equation \eqref{eq:asymp1} because we can regard this example as the case where the closest zeros move off to infinity, i.e., $c\to \infty.$ We will now consider more non-trivial examples.

\subsubsection*{Example 1: Two zeros} 
Consider the weight function:
\beq \label{eq:SD2zeros}
\omega(E)=\mathcal{N} \sqrt{1-E^2} (p^2-E^2),
\eeq
where $p=\text{cosh}\left(\frac{\lm}{2}\right)$ and $\mathcal{N}=\frac{2}{\pi(p^2-\frac{1}{4})}$ is the normalization constant. As shown in appendix \ref{app:asymp_lanc}, equation \eqref{eq:asymp1} for the asymptotics of the Lanczos coefficients in this case gives:
\beq
\overline{b}_n=\frac12-\frac{\left( e^{\lm}-e^{-\lm}\right)^2}{8 \ e^{\frac{3\lm}{2}} \ \text{cosh}\left(\frac{\lm}{2}\right)} e^{-\lm n}+ \cdots.
\eeq
Actually, for the spectral density in equation \eqref{eq:SD2zeros}, it is possible to write down the Lanczos coefficients exactly (see theorem \ref{thm:Christoffel}); as we will see below, the above expression agrees with this exact formula up to higher order corrections. 

\subsubsection*{Example 2: Multiple zeros} 
The above simple example can be generalized to include multiple zeros (and poles as well):
\beq
\om(E)=\mathcal{N} \sqrt{1-E^2}  \frac{\prod_{i=0}^{P-1}(p_i^2-E^2)}{\prod_{j=0}^{Q-1}(r_j^2-E^2)}.
\eeq
For this general class of weights, the asymptotic behavior of the Lanczos coefficients is again given by:
\beq\label{eq:example1}
\overline{b}_n=\frac12- \al e^{-\lm_0 n}+\cdots,
\eeq
with
\beqn\label{eq:example2}
\al &=& \frac{\left( e^{\lm_0}-e^{-\lm_0}\right)^2}{8 \ e^{\frac{3\lm_0}{2}} \ \text{cosh}\left(\frac{\lm_0}{2}\right)} \ \frac{\prod_{i=1}^{M-1} \left(\frac{1-e^{-\lm_0} e^{-\lm_i}}{1-e^{\lm_0} e^{-\lm_i}}\right)}{\prod_{j=0}^{L-1} \left(\frac{1-e^{-\lm_0} e^{-\eta_j}}{1-e^{\lm_0} e^{-\eta_j}}\right)},
\eeqn
where we have defined $p_i = \text{cosh}\left(\frac{\lambda_i}{2}\right)$, $r_j = \text{cosh}\left(\frac{\eta_j}{2}\right)$.

\subsubsection*{Example 3: DSSYK} 
The spectral density of the DSSYK model (see equation \eqref{eq:density}) is an extension of the class of theories we considered in the previous example with \emph{infinitely} many zeros and no poles. Furthermore, the zeros are at
\beq
p_k= \pm \cosh\left(\frac{k \lm}{2}\right). 
\eeq
Using our asymptotic formulas in (\ref{eq:asymp1}) and (\ref{eq:asymp2}), we find:
\beqn
\overline{a}_n&=&0\\
\overline{b}^2_n&=&\frac{1}{4}(1-e^{-\lm n})+\cdots.
\eeqn
This agrees with the large-$n$ behavior of the exact Lanczos coefficients given in equation \eqref{eq:DSSYKLC2}. Remarkably, from the exact formula \eqref{eq:DSSYKLC2}, it appears that there are no higher order corrections to the Lanczos coefficients, and that the full answer comes from the leading order correction from the first pair of zeros. It would be interesting to understand this in more detail. 

\subsection{Christoffel deformations} \label{sec:CD}

There is another useful tool for computing orthogonal polynomials and their recursion coefficients which we will use, and so we briefly review it now. Consider a random matrix theory with the spectral density $\omega(E) = \overline{\rho(E)}$. A \emph{polynomial Christoffel deformation} of $\omega$ is a new spectral density $\widetilde{\omega}$ defined as
\begin{equation} \label{eq:CDPoly}
    \widetilde{\omega}(E) = \mathcal{N} \prod_{i=1}^{k} (p_i - E)\, \omega(E).
\end{equation}
where $\{p_i\}_{i=1}^k$ is a collection of real numbers such that $p_i > 1$, and $\mathcal{N}$ is a normalization constant that depends on the $\{p_i\}$. Note that the zeros $p_i$ need not be distinct, but for simplicity we will only consider the case where each zero is simple. The deformed density differs from the original density in that the deformation introduces $k$ new zeros at the points $p_i$. Note that the deformed density is non-negative on $\text{supp}(\omega)$ since each $p_i > 1$.

Likewise, we define a \emph{rational Christoffel deformation}\footnote{Rational Christoffel deformations are often referred to as \emph{Christoffel-Uvarov deformations} in the literature \cite{ismail2005classical}.} of $\omega$ as
\begin{equation} \label{eq:CDRational}
    \widetilde{\omega}(E) = \mathcal{N} \frac{\prod_{i=1}^{k}(p_i - E)}{\prod_{j=1}^l (q_j - E)} \omega(E).
\end{equation}
where the $p_i$ and $q_j$ are real numbers that are larger than $1$. In addition to $k$ new zeros at the points $p_i$, a rational Christoffel deformation introduces $l$ new poles at the points $\{q_j\}_{j=1}^l$, each of which will be assumed to be simple. As before, note that deformed density is positive on $\text{supp}(\omega)$ since each $p_i, q_j > 1$. 

Interestingly, if the orthogonal polynomials with respect to the original density $\omega$ are known, it is straightforward to compute the orthogonal polynomials, and the associated Lanczos coefficients, with respect to the deformed density $\widetilde{\omega}$ \cite{szeg1939orthogonal, ismail2005classical}. For simplicity, in the main text we only state the relevant theorems for polynomial deformations. The analogous theorems for rational deformations (i.e., including poles) can be found in Appendix \ref{app:RationalCD}.

\begin{theorem}[Polynomial Christoffel deformations] \label{thm:Christoffel}
    Let $\widetilde{\omega}$ be a polynomial Christoffel deformation of $\omega$ that introduces a collection of simple zeros at the points $\{p_i\}_{i=1}^k$ with $p_i > 1$. Let $\pi_n(E)$ be the (monic) orthogonal polynomials with respect to $\omega$. Then the (monic) orthogonal polynomials with respect to $\widetilde{\omega}$ are given by
    \begin{equation} 
        \widetilde{\pi}_n(E) = \frac{(-1)^k}{ \prod_{i=1}^k (p_i - E)}
        \frac{\det\begin{pmatrix}
            \pi_n(p_1) & \pi_{n+1}(p_1) & \cdots & \pi_{n+k}(p_1) \\
            \vdots & \vdots & \ddots & \vdots \\
            \pi_{n}(p_k) & \pi_{n+1}(p_k) & \cdots & \pi_{n+k}(p_k) \\
            \pi_n(E) & \pi_{n+1}(E) & \cdots & \pi_{n+k}(E) 
            \end{pmatrix}}{\det \left[ \pi_{n + j -1}(p_i) \right]_{i,j = 1}^k}.
    \end{equation}
\end{theorem}

In this paper, since we only consider even weights, we will primarily focus on Christoffel deformations that are symmetric about the origin. An even Christoffel deformation that introduces $2k$ new zeros at the points $\{ \pm p_i \}_{i=1}^k \subset \mathbb{R} \setminus \operatorname{supp}(\omega)$ can be written as
\begin{equation} \label{eq:CDEven}
    \widetilde{\omega}(E) = \mathcal{N} \prod_{i=1}^k (p_i^2 - E^2) \omega(E).
\end{equation}
We then have the following corollary.
\begin{corollary} \label{cor:Cor1}
    If both $\omega$ and its Christoffel deformation $\widetilde{\omega}$ are even, then the deformed orthogonal polynomials are given by
    \begin{equation} \label{eq:CDPolyEven}
    \widetilde{\pi}_n(E) = \frac{(-1)^k}{\prod_{i=1}^k (p_i^2 - E^2)} \frac{\det 
    \begin{pmatrix}
        \pi_n(p_1) & \pi_{n+2}(p_1) & \cdots & \pi_{n+2k}(p_1) \\
        \vdots & \vdots & \ddots & \vdots \\
        \pi_n(p_k) & \pi_{n+2}(p_k) & \cdots & \pi_{n+2k}(p_k) \\
        \pi_n(E) & \pi_{n+2}(E) & \cdots & \pi_{n+2k}(E) 
    \end{pmatrix}}{\det[ \pi_{n + 2(j-1)}(p_i)]_{i,j=1}^k}.
\end{equation}
\end{corollary}
These are standard results in the theory of orthogonal polynomials, and the proofs are reviewed in appendix \ref{app:ChristoffelProof}. 

As a simple example, consider a polynomial Christoffel deformation of $\omega$ that adds a pair of new zeros at the points $\pm p$ (with $p > 1$):
\begin{equation}
    \widetilde{\omega}(E) = \mathcal{N} (p^2 - E^2)\, \omega(E).
\end{equation}
Using equation \eqref{eq:CDPolyEven}, it is easy to check that the deformed orthogonal polynomials satisfy
\begin{equation} \label{eq:Christoffel}
    \widetilde{\pi}_n(E) = \frac{1}{E^2 - p^2} \left( \pi_{n+2}(E) - \frac{\pi_{n+2}(p)}{\pi_n(p)} \pi_{n}(E) \right).
\end{equation}
Note that although the right hand side has a denominator of the form $E^2 - p^2$, the quantity in brackets has zeros at $\pm p$, and therefore $\widetilde{\pi}_n$ is indeed a polynomial of degree-$n$.

Having expressed the deformed polynomials $\widetilde{\pi}_n$ in terms of the undeformed polynomials $\pi_n$, we are now in a position to compute the deformed Lanczos coefficients. Let the deformed and the undeformed polynomials satisfy recursion relations with Lanczos coefficients $\bar{\mathsf{b}}_n$ and $\bar{b}_n$ respectively, i.e.
\begin{align}
    E\widetilde{\pi}_n(E) &= \widetilde{\pi}_{n+1}(E) + \bar{\mathsf{b}}_n^2 \widetilde{\pi}_{n-1}(E), \label{eq:DefLanczos} \\
    E\pi_n(E) &= \pi_{n+1}(E) + \bar{b}_n^2 \pi_{n-1}(E) \label{eq:UndefLanczos}.
\end{align}
Using equation \eqref{eq:CDPolyEven} and the above recursion relations, a short computation leads to the following result (see appendix \ref{app:ChristoffelProof} for a proof):
\begin{corollary} \label{cor:Cor2}
Let $\widetilde{\omega}$ be an even polynomial Christoffel deformation of $\omega$, as defined in eq.\ \eqref{eq:CDEven}. Let $\bar{\mathsf{b}}_n$ and $\bar{b}_n$ be the deformed and undeformed Lanczos coefficients respectively. Then we have
\begin{equation} \label{eq;bnRatioGeneral}
    \left( \frac{\bar{\mathsf{b}}_n}{\bar{b}_n} \right)^2 = \frac{ \det[\pi_{n+2j}(p_i)]_{i,j=1}^k \, \det[ \pi_{n-1+2(j-1)}(p_i)]_{i,j=1}^k}{ \det[\pi_{n+2(j-1)}(p_i)]_{i,j=1}^k\, \det[\pi_{n-1+2j}(p_i)]_{i,j=1}^k}.
\end{equation}
\end{corollary}

For the case of single pair of zeros at $\pm p$ (with $p > 1$), this formula simplifies to
\begin{equation} \label{eq:bnRatio1}
    \bar{\mathsf{b}}_n^2 = \frac{\pi_{n+2}(p) \pi_{n-1}(p)}{\pi_n(p) \pi_{n+1}(p)}\, \bar{b}_n^2.
\end{equation}


\subsubsection*{Example: Two zeros revisited} 
Let us revisit the example with two zeros considered in the previous section. The weight function of interest is given by:
\beq
\widetilde{\omega}(E)=\mathcal{N} \sqrt{1-E^2}\; (p^2-E^2).
\eeq
The monic polynomials for $\omega(E)=\frac{2}{\pi} \sqrt{1-E^2}$ are the Chebyshev polynomials of the second kind, and the corresponding Lanczos coefficients are given by $\overline{b}_n=1/2$ for all $n$. We can now use equation \eqref{eq:bnRatio1} to obtain the exact form of the Lanczos coefficients for the deformed weight:
\beq
\overline{\mathsf{b}}_n^2 = \frac14\frac{U_{n+2}(p)\, U_{n-1}(p)}{U_n(p)\, U_{n+1}(p)}.
\eeq
The Chebyshev polynomials for $p>1$ are given by:
\beq
U_n(\text{cosh}(\eta))=\frac{1}{2^{n-1}}\frac{\text{sinh}((n+1)\eta)}{\text{sinh}((\eta)}.
\eeq
Using this and the parametrization, $p=\text{cosh}(\lm/2)$, we get:
\beq
\overline{\mathsf{b}}_n^2 = \frac14 \frac{(1-e^{-\lm n}) (1-e^{-\lm (n+3)})}{(1-e^{-\lm (n+1)}) (1-e^{-\lm (n+2)})}.
\eeq
Expanding this to the leading order in $e^{-\lm n}$, we get:
\beq
\overline{\mathsf{b}}_n^2=\frac{1}{4}-\frac{\left( e^{\lm}-e^{-\lm}\right)^2}{8 \ e^{\frac{3\lm}{2}} \ \text{cosh}\left(\frac{\lm}{2}\right)} e^{-\lm n}+\cdots.
\eeq
This agrees with the result that was obtained in the previous section using theorem \ref{thm:main}. The more general formula in equation \eqref{eq;bnRatioGeneral} immediately allows one to write down exact formulas for the Lanczos coefficients corresponding to even Christoffel deformations of the GUE spectral density with arbitrary (but finite) number of zeros.

\section{Continuum limit and JT gravity} \label{sec:continuum}
In the previous section, we discussed how the effective Hamiltonian in the Krylov basis at large $n$ universally seems to become that of a free particle hopping on the Krylov chain in the presence of an exponential potential. The exponential potential is, of course, reminiscent of the Liouville Hamiltonian of JT gravity \cite{Bagrets:2016cdf, Jafferis:2022wez}. In light of this, it is natural to conjecture the following: \emph{for any one-cut matrix model that forms a UV completion of JT gravity -- i.e., where the spectral density of the matrix model agrees with that of JT gravity near the edges -- the effective Hamiltonian in the Krylov basis should reduce to the Liouville Hamiltonian of JT gravity in a low-energy, continuum limit}. If this is true, then the Krylov basis provides a UV completion of the geometric ``wormhole-length'' basis of JT gravity. 

In this section, we will give evidence towards the above conjecture by using the DSSYK model as our starting point. In DSSYK, it is known that the Hamiltonian in the Krylov basis reduces to JT gravity in the low-energy, continuum limit \cite{Berkooz:2018qkz, Lin:2022rbf, Rabinovici:2023yex}. Recall that in this limit (sometimes referred to as the ``triple-scaling'' limit in the context of DSSYK), the Krylov depth $n$ is replaced by a continuous variable $\ell$:
\beq \label{eq:scaling}
n = n_{\star} + \frac{\ell}{\lambda},\;\;n_{\star} = \frac{1}{\lambda}\log\left(\frac{1}{\lambda^2}\right),\;\;\;\cdots\;\;(\text{continuum limit}),
\eeq
where we take the limit $\lm\to 0$ with $\ell$ fixed. Then, the effective Hamiltonian in terms of the $\ell$ variable becomes the Liouville Hamiltonian of JT gravity. Here, we will show that for a large class of ultraviolet deformations of the spectral density of the DSSYK model -- where the spectral density is modified in the UV, leaving the spectral density near the edges unchanged -- the semiclassical effective description in the Krylov basis continues to reduce to JT gravity in the low-energy, continuum limit.  

Note that the limit in equation \eqref{eq:scaling} does two things: (i) it corresponds to taking $n$ large but not independently of $\lm$ (on account of the shift by $n_{\star}=\frac{1}{\lm} \log(\frac{1}{\lm^2})$), and (ii) it corresponds to zooming-out to get a continuum approximation. In the context of the DSSYK model, the parameter $\mathfrak{q}=e^{-\lambda}$ is a parameter of the theory, but we can also think of $\mathfrak{q}$ as the location of the nearest zero of $\omega$ away from the cut (i.e., not including the zeros at $E=\pm 1$). We emphasize that the continuum limit is \emph{not} the same as the large-$n$ limit considered in theorem \ref{thm:main}, because we are taking the limit $n\to \infty$ in a coordinated way with $\lm \to 0$. Much of the analysis in this section is devoted to treating this limit carefully. 

In what follows, we will study the continuum limit described above in two classes of random matrix theories: 
\begin{enumerate}
    \item where the function $h(E)$ in the spectral density has a finite number of zeros and poles in the complex-$E$ plane, 
    \item where the function $h(E)$ has an infinite number of zeros in the complex-$E$ plane. 
\end{enumerate}    
In this latter case, we will specialize to a particular class of models corresponding to rational Christoffel deformations of the DSSYK model with a finite number of \emph{additional} (i.e., in addition to the DSSYK zeros) pairs of zeros and poles. Crucially, the locations of these additional zeros will \emph{not} be scaled in the $\lm \to 0$ limit. In both these cases, we will show that the Lanczos coefficients in the $\lambda\to 0$ limit take the form:
\beq 
\overline{b}_n = \frac{1}{2}+\lambda^2b_{(2)}(\ell)+\lambda^3 b_{(3)}(\ell)+\cdots,
\eeq 
where the leading correction term turns out to be $O(\lambda^2)$ because of the shift by $n_{\star}$ in equation \eqref{eq:scaling}. The wavefunction $\psi_n(E)$, which was previously a function of the discrete index $n$, now becomes a function of the continuous variable $\ell$:
\beq \label{eq:scaling2}
\psi(\ell) := \frac{1}{\sqrt{\lambda}} \psi_{n}.
\eeq 
Substituting equations \eqref{eq:scaling} and \eqref{eq:scaling2} into the tri-diagonal form of the effective Hamiltonian:
\beq
-H_{\text{eff}} \psi_n(E)=\overline{b}_{n+1} \psi_{n+1}(E)+\overline{b}_n \psi_{n-1}(E),
\eeq
and doing a Taylor series expansion in $\lambda$, we get a new continuum Hamiltonian:
\beq \label{eq:contHam}
 H_{\text{cont.}} =-\frac{1}{\lm^2} + \left(-\frac{1}{2} \partial_{\ell}^2-2b_{(2)}(\ell) \psi(\ell) \right)+O(\lambda).
\eeq

\subsection{Case (i): Finite number of zeros and poles}\label{sec:finitepoles}
Consider the following class of models with the spectral density given by:
\beq
\om(E)=\mathcal{N} \sqrt{1-E^2} \prod_{i=0}^{Q-1}(q_i^2-E^2) \ \frac{\prod_{i=0}^{P-1}(p_i^2-E^2)}{\prod_{i=0}^{R-1}(r_i^2-E^2)},
\eeq
where now we have separated out some zeros $\{\pm q_i\}$ of $h(E)$ as special. We will consider a limit where the $\{q_i\}$ start approaching the cut $[-1,1]$ in the $\lambda \to 0$ limit, while the $\{p_i\}$ and $\{r_i\}$ remain fixed away from the cut. We parameterize the scaled zeros as 
$$q_i=\text{cosh}\left(\frac{\lm f(i)}{2}\right),$$ where, $f(i)>f(j)$ for $i>j$ and $f(0)=1$. Note that the scaling parameter $\lambda$ is related to the position of the closest pair of zeros as $q_0 = \cosh(\lm/2)$. For this class of theories, we show in appendix \ref{app:asymp_lanc}, that:
\beq
b_{(2)}(\ell)=0,
\eeq
in the continuum limit. Heuristically, this happens because the coefficient in front of the leading exponential $e^{-\lm n}$ in the large-$n$ limit is itself $O(\lambda^2)$, and so when we take a further scaling limit, this contribution ends up being of $O(\lm^4)$. However, we emphasize that this is merely a heuristic explanation; as we have noted, the continuum limit is not the same as the large-$n$ limit, and in Appendix \ref{app:asymp_lanc}, we do the analysis carefully in the continuum limit. Substituting $b_{(2)}=0$ in equation \eqref{eq:contHam}, the continuum Hamiltonian becomes:
\beq
H_{\text{cont.}}=-\frac{1}{\lm^2}-\frac{1}{2}\partial_{\ell}^2+O(\lm^2),
\eeq
which is just the free particle Hamiltonian.

\subsection{Case (ii): Christoffel deformations of DSSYK} \label{sec:CDTS}
Consider now the following class of spectral density functions:
\beq\label{eq:CD}
\om(E)= \mathcal{N}\,\om_{\text{DSSYK}} (E) \frac{\prod_{i=0}^{P-1}(p_i^2-E^2)}{\prod_{j=0}^{Q-1}(q_j^2-E^2)},
\eeq
where $\omega_{\text{DSSYK}}$ is the spectral density of the DSSYK model, and $\mathcal{N}$ is a normalization constant which depends on $\{p_i\}$ and $\{q_j\}$. These are (even) rational Christoffel deformations of the spectral density of the double scaled SYK model, with additional zeros at $E=\pm p_i$ and poles at $E=\pm q_j$ (with $p_i, q_j > 1$). One can think of the spectral density in equation \eqref{eq:CD} as a large class of ``ultraviolet'' deformations of the spectral density. The locations of the additional zeros and poles will \emph{not} be scaled in the $\lm \to 0$ limit and thus the near-edge JT behavior of the spectral density is left unaffected; but the spectral density in the UV (i.e., in the middle of the spectrum) can be more or less arbitrarily modified using such deformations (see figure \ref{fig:sat_value}). Note that the DSSYK spectral density can be realized as the large-$D$ spectral density of a Hermitian matrix model called the ETH matrix model \cite{Jafferis:2022wez}; the Christoffel deformations of the DSSYK spectral density can, in principle, also be translated to deformations of the potential of the ETH matrix model using the formula \cite{DiFrancesco:1993cyw}:
\beq 
V'(E) = - 2\pi i\Big( C[\rho](E+i\epsilon) + C[\rho](E-i\epsilon)\Big),\;\;C[\rho](z) = \int_{-1}^1\frac{dx}{2\pi i}\frac{\rho(x)}{x-z}.
\eeq

\begin{figure}[t]
 \centering
 \includegraphics[width=0.6\textwidth]{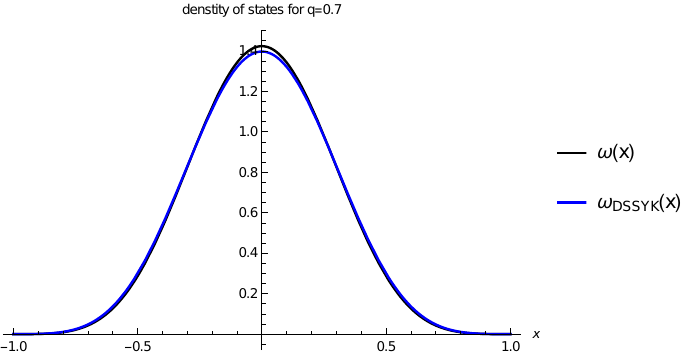}
    \caption{The spectral density $\omega(E)$ for a Christoffel deformation of DSSYK with two additional zeros at $p=\pm2$ for $q=0.7$. We see that the spectrum differs from the DSSYK density only near the center and approaches the DSSYK density at the edges.}  
    \label{fig:sat_value}
\end{figure}

One of our main results is that in this class of models, one gets
\beq \label{eq:CT2}
b_{(2)}(\ell) = -\frac{1}{4}e^{-\ell}
\eeq 
in the continuum limit, exactly the same as in the undeformed DSSYK model. This implies that the effective Hamiltonian in the low-energy, continuum limit agrees with the Liouville Hamiltonian of JT gravity: 
\beq
H_{\text{cont.}}=-\frac{1}{\lm^2}+\frac{1}{2}\left(-\partial_{\ell}^2+ e^{-\ell}\right)+O(\lm).
\eeq
We can derive this result in several ways: here we will use theorem \ref{thm:Christoffel} directly in the continuum limit to give an elementary proof for polynomial Christoffel deformations, i.e., with additional zeros, but no poles. The case of poles is conceptually similar, but more technically involved, and is discussed in Appendix \ref{sec:RationalCDScaling}. A second derivation of the same results using the RHP method is presented in Appendix \ref{app:asymp_lanc}. 

\subsubsection*{Proof for polynomial Christoffel deformations}
We will proceed by method of induction on the number of additional zeros. First, consider the simplest case of a Christoffel deformation with one pair of added zeros:
\beq
\om(E)=\mathcal{N}\,\om_{\text{DSSYK}} (E) (p^2-E^2).
\eeq 
In this case, we can simply use the results from theorem \ref{thm:Christoffel}. Recall from eq.\ \eqref{eq:LanczosDSSYK} that in the DSSYK model, the undeformed Lanczos coefficients are given by
\begin{equation}
    \bar{b}_n^2 = \frac{1 - \bq^n}{4}, \qquad \bq =  e^{-\lambda}.
\end{equation}
Our goal now is to evaluate the ratio $\bar{\mathsf{b}}_n / \bar{b}_n$, where $\mathsf{b}_n$ are the deformed Lanczos coefficients. The main novelty is that we will compute this ratio directly in the continuum limit, i.e. the limit $\lambda \to 0$ while keeping the continuum length $\ell$ and the location $p$ of the additional pair of zeros fixed. To this end, we define the ratios
\begin{equation}
    R_n(p) = \frac{\pi_{n+2}(p)}{\pi_n(p)}, \qquad r_n(p) = \frac{\pi_{n+1}(p)}{\pi_n(p)}, 
\end{equation}
where the $\pi_n$ are monic orthogonal polynomials corresponding to the DSSYK spectral density. Clearly, from the above definitions $R_n(p) = r_{n+1}(p) r_n(p)$, and we have from equation \eqref{eq:bnRatio1}
\begin{equation} \label{eq:bnratio}
    \left(\frac{\bar{\mathsf{b}}_n}{\bar{b}_n}\right)^2 = \frac{R_{n}(p)}{R_{n-1}(p)} = \frac{r_{n+1}(p)}{r_{n-1}(p)}.
\end{equation}
Note that from equation \eqref{eq:UndefLanczos}, we get 
\beq 
p = \frac{\pi_{n+1}(p)}{\pi_n(p)} + \overline{b}_n^2\frac{\pi_{n-1}(p)}{\pi_{n}(p)},
\eeq 
which implies the following recursion relation for $r_n$:
\beq 
p = r_{n}(p) + \frac{\overline{b}_n^2}{r_{n-1}(p)}.
\eeq 
We now assume that in the continuum limit with $p$ fixed, $r_n(p)$ admits a good limit, and define
\beq 
r(\ell,p) := r_{n =n_{\star} + \frac{\ell}{\lm}}(p),\;\;n_{\star}=\frac{1}{\lm}\log(\frac{1}{\lm^2}).
\eeq 
This object satisfies the following recursion relation:
\begin{equation} \label{eq:r_TS}
    p = r(\ell, p) + \frac{(1 - \lambda^2 e^{-\ell})}{4} \frac{1}{r(\ell - \lm, p)}.
\end{equation} 
The recursion relation implies that there cannot be any singular behavior for $r$ as $\lm\to0$, so we postulate the following expansion in powers of $\lambda$:
\begin{equation} \label{eq:rExp}
    r(\ell, p) = \sum_{n=0}^{\infty} r_{(n)}(\ell, p) \lambda^n.
\end{equation}
Likewise, Taylor expanding $r(\ell - \lambda, p)$ and substituting eq.\ \eqref{eq:rExp}, we get:
\begin{equation} \label{eq:rShiftedExp}
    r(\ell - \lambda, p) = \sum_{n,m=0}^{\infty} \frac{(-1)^n}{n!} \lambda^{n+m} \partial_{\ell}^n r_{(m)}(\ell, p).
\end{equation}
The idea now is to solve eq.\ \eqref{eq:r_TS} perturbatively in $\lambda$. To that end, we expand out the right hand side of eq.\ \eqref{eq:r_TS} and plug in the expansions eqs.\ \eqref{eq:rExp} and \eqref{eq:rShiftedExp}. At $O(\lambda^0)$, we get
\begin{equation}
    r_{(0)}^2 - p\, r_{(0)} + \frac{1}{4} = 0,
\end{equation}
whose solutions are
\begin{equation}
    r_{(0)}(\ell, p) = \frac{1}{2} \left( p \pm \sqrt{p^2 - 1} \right).
\end{equation}
The physical branch is the one with the + sign since it follows from definition that $r_{(0)}(\ell, p) \sim p$ as $p \to \infty$.\footnote{This is also consistent with the large-$n$ asymptotics of orthogonal polynomials:
\begin{equation}
    \pi_n(x) \sim \left( \frac{\phi(x)}{2} \right)^n, \qquad \phi(x) = x + \sqrt{x^2 - 1},
\end{equation}
for $n \to \infty$ and fixed $x$; see equation \eqref{eq:leadpi} in Appendix \ref{app:asymptotics}.} Therefore, at $O(\lambda^0)$ we get
\begin{equation} \label{eq:assumpt}
    r_{(0)}(\ell, p) = \frac{\phi(p)}{2}, \qquad \phi(p) = p + \sqrt{p^2 - 1}.
\end{equation}
Note that $r_{(0)}$ is independent of the length variable $\ell$. The $\ell$-dependence appears at first subleading order. Using the fact that $r_{(0)}$ is independent of $\ell$, at $O(\lambda)$ we get
\begin{equation}
    r_{(1)} \left(1-  \frac{1}{4r_{(0)}^2}\right) = 0,
\end{equation}
which implies that $r_{(1)}(\ell, p) = 0$ (since $p > 1$, the other factor can never vanish). Next, using the relations $r_{(0)} = \phi(p) / 2$ and $r_{(1)} = 0$, it easy to check that we get the following solution for $r_{(2)}$:
\begin{align} 
    r_{(2)}(\ell, p) = \frac{1}{2}\frac{\phi(p)}{ \phi(p)^2 - 1 } e^{-\ell}. \label{eq:r_2}
\end{align}
In principle, we can compute the higher order terms as well, but we will not need them here. In Appendix \ref{sec:Hermite}, we will reproduce the above formulas for $r_{(0)}$ and $r_{(2)}$ directly from the known formulas for $\bq$-Hermite polynomials. 

It now follows that $\bar{\mathsf{b}}_n$ also admits a similar expansion in the triple scaling limit. From eq.\ \eqref{eq:bnratio}, we get
\begin{equation}
    \left( \frac{\bar{\mathsf{b}}_n}{\bar{b}_n} \right)^2 = \frac{r_{n+1}(p)}{r_{n-1}(p)} \to \frac{r(\ell + \lambda, p)}{r(\ell - \lambda, p)}.
\end{equation}
Substituting the expansion for $r$ above, we find that
\begin{equation}
    \frac{\bar{\mathsf{b}}_n}{\bar{b}_n} \to 1 + O(\lambda^3).
\end{equation}
Hence, to $O(\lm^2)$, we see that the Lanczos coefficients are the same as in the DSSYK model with corrections appearing only at $O(\lambda^3)$. Actually, we can compute the $O(\lm^3)$ correction as well -- note that the $r_{(3)} \lambda^3$ term does not contribute to the $O(\lambda^3)$ correction to the Lanczos ratio since $r_{(3)}$ appears with the same sign in both the numerator and denominator, and therefore cancels once the denominator is expanded. So, we get 
\beq 
\left( \frac{\bar{\mathsf{b}}_n}{\bar{b}_n} \right)^2 = 1 + 2 \lm^3 \frac{\partial_{\ell}r_{(2)}}{r_{(0)}} = 1 -  \lm^3 \frac{2e^{-\ell}}{(\phi^2-1)}+ O(\lm^4).
\eeq 
In Appendix \ref{app:asymp_lanc}, this formula is reproduced using the RHP method (see equation \eqref{eq:conf}). 

Having shown that the Lanczos coefficients are unaffected with one pair of zeros, we now proceed to the more general case using induction. Consider an even Christoffel deformation of $\omega_{\text{DSSYK}}$ that adds $P$ pairs of zeros at the points $\pm p_i$:
\begin{equation}
    \omega^{(P)}(E) = \mathcal{N}_P\prod_{i=0}^{P-1} (p_i^2 - E^2)\ \omega_{\text{DSSYK}}(E).
\end{equation}
Let $\bar{\mathsf{b}}_n^{(P)}$ be the Lanczos coefficients associated to $\omega^{(P)}$, and let $\pi_n^{(P)}$ be the corresponding orthogonal polynomials. Let us assume that
\begin{equation} \label{eq:bnInd}
    \bar{\mathsf{b}}_n^{(P)} = \bar{b}_n+ O(\lambda^3),
\end{equation}
where $\bar{\mathsf{b}}_n^{(P=0)} = \bar{b}_n$ are the Lanczos coefficients in the undeformed DSSYK model. Consequently, the orthogonal polynomials corresponding to $\omega^{(P)}$ satisfy the recursion relation
\begin{equation}
    E \pi_n^{(P)}(E) = \pi_{n+1}^{(P)}(E) + \frac{(\bar{b}_n + O(\lambda^3))}{4} \pi_{n-1}^{(P)}(E).
\end{equation}
Now consider adding one more pair of zeros to the Christoffel deformation: 
\beq 
\omega^{(P+1)}(E) \propto (p^2 - E^2)\ \omega^{(P)}(E).
\eeq 
Then, from equation \eqref{eq:bnRatio1} we have
\begin{equation}
    \left( \frac{\bar{\mathsf{b}}_n^{(P+1)}}{\bar{\mathsf{b}}_n^{(P)}} \right)^2 = \frac{r_{n+1}^{(P)}(p)}{r^{(P)}_{n-1}(p)},
\end{equation}
where the ratio $r_{n}^{(P)}(p)$ is defined as $r^{(P)}_n(p) = \pi_{n+1}^{(P)}(p) / \pi_{n}^{(P)}(p)$. As in equation \eqref{eq:rExp}, we assume that $r^{(P)}_n$ admits a well-defined triple scaling limit which we call $r^{(P)}(\ell, p)$. This quantity satisfies the recursion relation:
\beq 
p = r^{(P)}(\ell,p) + \frac{1}{4}(1-\lm^2 e^{-\ell}+O(\lm^3))\frac{1}{r^{(P)}(\ell-\lm,p)}.
\eeq 
The recursion relation implies that there are no singularities at $\lm=0$, and thus $r^{(P)}$ can be expanded in powers of $\lambda$ as
\begin{equation}
    r^{(P)}(\ell, p) = \sum_{n = 0}^{\infty} r^P_{(n)}(\ell, p) \lambda^n.
\end{equation}
The coefficients $r_{(n)}^P$ can be computed using the recursion relation. This calculation is essentially identical to the $P=1$ case and shows that the coefficients remain unchanged up to $O(\lambda^2)$:
\begin{align}
    r^P_{(0)}(\ell, p) &= \frac{\phi(p)}{2}, \\
    r^P_{(1)}(\ell, p) &= 0, \\
    r^P_{(2)}(\ell, p) &= \frac{\phi(p)}{2 (\phi(p)^2 - 1)} e^{-\ell}.
\end{align}
Therefore, in the triple scaling limit
\begin{align*}
    \left(\frac{\bar{\mathsf{b}}^{(P+1)}_n}{\bar{\mathsf{b}}^{(P)}_n}\right)^2 &\to \frac{r^{(P)}(\ell + \lambda, p)}{r^{(P)}(\ell - \lambda, p)} \\
    &= 1 + O(\lambda^3).
\end{align*}
Therefore, by induction on $P,$ equation \eqref{eq:bnInd} holds for any $P \geq 1$. This completes the proof for polynomial Christoffel deformations. We can go a step further and compute the $O(\lm^3)$ correction to the Lanczos ratio:
\beq \label{eq:CDzeros}
\left( \frac{\bar{\mathsf{b}}^{(P)}_n}{\bar{b}_n} \right)^2  = 1 -  \lm^3\sum_{i=0}^{P-1}\frac{2e^{-\ell}}{(\phi^2(p_i)-1)}+ O(\lm^4).
\eeq 
As a consistency check, this formula will be reproduced using the RHP method in Appendix \ref{app:asymp_lanc}. The more general case of rational Christoffel deformations with zeros and poles is discussed in Appendix \ref{sec:RationalCDScaling}.

To summarize, we have argued that for any one-cut matrix model where the spectral density corresponds to an even, rational Christoffel deformation (with arbitrary but finite number of extra pairs of zeros and poles) of the DSSYK spectral density, the effective Hamiltonian in the Krylov basis reduces in the low-energy continuum limit to the Liouville Hamiltonian of JT gravity. Together with the fact that such Christoffel deformations modify the spectral density in the middle of the spectrum leaving the near-edge JT behavior unchanged, we see that in this large class of matrix models the semiclassical effective description in the Krylov basis reduces to JT gravity in the low-energy, continuum limit. So, at least in this large class of UV completions of the spectral density of JT gravity, the Krylov basis provides a natural UV version of the geometric ``wormhole-length'' basis of JT gravity. This statement was, of course, well-known in the specific context of the DSSYK model \cite{Berkooz:2018qkz, Lin:2022rbf, Rabinovici:2023yex}, but here we have shown that it is more robustly true in a way that is independent of the details of the UV completion. We also take our results as evidence that the semiclassical effective description in the Krylov basis should be interpreted as a dual gravitational description.

\section{Summary}\label{sec:discussion}
In this paper, we have explained how the construction of the semiclassical bulk Hilbert space associated to the wormhole-length in DSSYK can be carried out in more general one-cut random matrix models satisfying certain minimal conditions. Our construction gives evidence pointing towards a semiclassical gravity dual for all such matrix models at the disc level, i.e., at leading order in the $D\to \infty$ limit. At least without matter, the bulk description is obtained by writing the Hamiltonian in the Krylov basis with respect to the maximally entangled state on two copies. The corresponding Hamiltonian is tri-diagonal, with the ensemble-averaged Lanczos coefficients given by the recursion coefficients of the orthogonal polynomials of the planar spectral density. For any one-cut model with an even spectral density compactly supported on an interval and with square root edges, the corresponding effective Hamiltonian in the Krylov basis behaves similarly to the chord Hamiltonian of DSSYK at large Krylov depth. Furthermore, in a large class of models corresponding to UV deformations of the DSSYK spectral density, the low-energy continuum limit of the Krylov effective Hamiltonian reduces to JT gravity. This reinforces the picture that the semiclassical Hilbert space in the Krylov basis should be interpreted as a bulk gravitational Hilbert space. The gravitational mode in the low-energy limit becomes the length of the wormhole connecting the two sides of the TFD state, but away from the low-energy limit, the bulk description apparently has a discrete wormhole length. It would be interesting to explore and extend this work in several directions, such as including matter, higher-genus contributions, supersymmetry etc.



\acknowledgments
We would like to thank Vijay Balasubramanian, Pawel Caputa, Abhijit Gadde, Ohad Mamroud, Shiraz Minwalla, Harshit Rajgadia, Pratik Rath, Joan Simon, Sandip Trivedi and Zhenbin Yang for helpful discussions. We are grateful to Ohad Mamroud and Joan Simon for helpful comments on an earlier version of the draft. This research was supported by the Department of Atomic Energy, Government of India, under Project Identification Number RTI-4012 and the Infosys Endowment for the study of the Quantum Structure of Spacetime. OP acknowledges fruitful discussions during the workshop ``Observers, wormholes and complex saddles in cosmology", organized at the Bernoulli Center for Fundamental Studies (EPFL, Lausanne) from 18--22 May 2026.

\appendix 

\section{Asymptotics of Lanczos coefficients from RHP} \label{app:asymptotics}
In this appendix, we will present a detailed review of the Riemann-Hilbert approach to extracting the asymptotic behavior of the Lanczos coefficients. We will largely follow the presentation in \cite{kuijlaars2004riemann} (see also \cite{kuijlaars2003riemann}). 

\subsection{Notation and conventions}\label{app:conventions}
In this section we will set up the notation and define various functions to be used later.
\begin{itemize}
    \item $\boldsymbol{\omega(z)}$: This is the analytic continuation of the weight function $\om(x)$ from $[-1, 1]$ to $\mathbb{C}$. Recall that $\om(x)=(1-x)^{\al} (1+x)^{\beta} \ h(x)$ and therefore has branch points at $\{-1, 1\}$ for non-integer values of $\al$ and $\beta$. We will choose the branch cuts to lie along $(-\infty,1]\cup[1,\infty)$ so that $\omega(x)$ is unambiguously defined for $x \in [-1, 1]$. For the particular case that is considered throughout this paper, namely $\al=\beta=\frac12$, we have:
    \beq\label{eq:def_omega}
    \om(z)= \begin{cases}-i \sqrt{z^2-1} \ h(z), & \text{Im}(z)>0 \\
    i \sqrt{z^2-1} \ h(z), &  \text{Im}(z)<0.
    \end{cases}
    \eeq
    Moreover, the branch of $\sqrt{z^2-1}$ is chosen such that $\sqrt{z^2-1} \to z$ as $z \to \infty$ throughout the appendix.
    \item $\boldsymbol{a(z)}$: We define $a(z)=(\frac{z-1}{z+1})^{1/4}$, with the branch cut chosen to lie along the real interval $[-1,1]$.
    \item $\boldsymbol{D(z)}$: This is the so-called Szeg\H{o} function, and is given by
    \beq
    D(z)=\exp \left({\frac{\sqrt{z^2-1}}{2 \pi}\int_{-1}^{1}\frac{dx}{z-x}\frac{\log\omega(x)}{\sqrt{1-x^2}}}\right).
    \eeq
    It is analytic everywhere except for a branch cut along the interval $[-1,1]$. For the weight function corresponding to GUE, i.e., $\om(x)=\frac2{\pi}\sqrt{1-x^2}$, it is given by:
    \beq\label{eq:D_J}
D_J(z)=\frac{(z^2-1)^{1/4}}{\sqrt{\phi(z)}}.
    \eeq
   For any other weight of the form $\om(x)=\sqrt{1-x^2} \ h(x)$, the Szeg\H{o} function has the decomposition
    \beq\label{eq:fact_J}
    D(z)=D_J(z) D_h(z),
    \eeq
    where 
    \beq\label{eq:Dhz}
    D_h(z)=\exp \left(\frac{\sqrt{z^2-1}}{2 \pi}\int_{-1}^{1}\frac{dx}{z-x}\frac{\text{log}\,h(x)}{\sqrt{1-x^2}}\right).
    \eeq
    We will use this notation throughout the appendix. More generally, the Szeg\H{o} function satisfies the following \emph{factorization property:} if the density of states is of the form:
    \beq
      \om(x)=\mathcal{N}\sqrt{1-x^2} h_1(x) h_2(x)
    \eeq
    where, $h_1,\ h_2$ are positive real analytic functions on the interval $[-1,1]$, then
    \beq
      D_h(z)=D_{h_1}(z) D_{h_2}(z).
    \eeq
    Another important property of $D_h(z)$ that is useful later is that it is even under $z \to -z$. To see this, note that:
    \beq
       D_h(-z)=\exp \left(-\frac{\sqrt{z^2-1}}{2 \pi}\int_{-1}^{1}\frac{dx}{-z-x}\frac{\text{log}\,h(x)}{\sqrt{1-x^2}}\right).
    \eeq
    under $x\to-x$, this becomes $D_h(z)$ as $h(x)$ is an even function. 
    \item $\boldsymbol{\phi(z)}$: this is the conformal map from $\mathbb{C} \setminus [-1,1]$ to the exterior of the unit circle and is given by $\phi(z)=z+\sqrt{z^2-1}$. The branch cut is chosen to lie along the real interval $[-1,1]$ and $\sqrt{z^2-1} \to z$ as $z \to \infty$. For $x \in (-1, 1)$, it is useful to define $\phi_{\pm}(x) = x \pm i \sqrt{1-x^2}$ so that $\phi_+(x) \phi_-(x) = 1$. Note that $\phi(z)$ is the natural analytic continuation of $\phi_+(x)$ to the upper half plane and $\phi_-(x)$ to the lower half plane. In other words, $\phi(x \pm i 0^+) = \phi_{\pm}(x)$ for $x \in (-1, 1)$. 
\end{itemize}

\subsection{Solving the RHP of section \ref{sec:poly}}\label{app:unique}
Let us begin by recalling the Riemann-Hilbert problem defined in section \ref{sec:poly}. We seek a matrix valued function $Y: \mathbb{C} \setminus [-1, 1] \to \mathbb{C}_{2\times 2}$ that satisfies the following properties:
\begin{enumerate}
    \item $Y$ is analytic on $\mathbb{C} \setminus [-1, 1]$.
    \item $Y$ satisfies the following jump condition on $(-1, 1)$:
    \begin{equation} \label{eq:JumpY}
        Y_+(x) = Y_{-}(x) \begin{pmatrix}
            1 & \omega(x) \\
            0 & 1
        \end{pmatrix}
    \end{equation}
    where we have defined $Y_\pm(x) = Y(x \pm i 0^+)$.
    \item $Y$ satisfies the asymptotic condition
    \begin{equation} \label{eq:YAsymp}
        Y(z) = \left( \mathbb{I} + O(1/z) \right) \begin{pmatrix}
            z^n & 0 \\
            0 & z^{-n}
        \end{pmatrix}   
    \end{equation}
    as $|z| \to \infty$.
    \item Near the endpoints $z \to \pm 1$, $Y_{ij}(z) = O(1)$ for every $1 \leq i,j \leq 2$.
\end{enumerate}
Let us first show that the solution to this problem, if it exists, is unique. We will then prove existence by explicitly constructing a solution.

Uniqueness is a straightforward consequence of Liouville's theorem. To begin with, let us take the determinant on both sides of eq.\ \eqref{eq:JumpY}. Clearly, $\det Y_+(z) = \det Y_{-}(z)$ and therefore $\det Y: \mathbb{C} \to \mathbb{C}$ is an entire function. Now, eq.\ \eqref{eq:YAsymp} implies that $\lim_{|z| \to \infty} \det Y(z) = 1$ and by analyticity, $\det Y(z) = 1$ everywhere. This shows that $Y$ must be invertible everywhere on $\mathbb{C}$. Now, suppose $Y$ and $\tilde{Y}$ are two solutions to the above RHP. Then, it is easy to check that
\begin{align*}
    Y_+(z) \tilde{Y}_+^{-1}(z) &= Y_{-}(z) \tilde{Y}_{-}^{-1}(z),
\end{align*}
which shows that $Y \tilde{Y}^{-1}$ is an entire function on $\mathbb{C}$. Moreover, from eq.\ \eqref{eq:YAsymp} we see that $Y(z) \tilde{Y}^{-1}(z) \to \mathbb{1}$ as $|z| \to \infty$ and therefore by Liouville's theorem $Y = \tilde{Y}$.\footnote{Note that the endpoint conditions are necessary to ensure uniqueness. Roughly speaking, without these conditions one can always multiply $Y$ by a meromorphic function with isolated singularities at $\pm 1$ and obtain a valid solution to the RHP. The endpoint conditions fix this ambiguity.} 

We will now construct a solution to the above RHP. The jump condition eq.\ \eqref{eq:JumpY} is equivalent to four additive RHPs on $(-1, 1)$:
\begin{align}
    Y^+_{11}(x) &= Y^{-}_{11}(x) \label{eq:JumpY11} \\
    Y^+_{21}(x) &= Y^{-1}_{21}(x) \label{eq:JumpY21}\\
    Y_{12}^{+}(x) - Y_{12}^{-}(x) &= \omega(x) Y^{-}_{11}(x) \label{eq:JumpY12}\\
    Y_{22}^{+}(x) - Y_{22}^{-}(x) &= \omega(x) Y^{-}_{21}(x) \label{eq:JumpY22}
\end{align}
 The asymptotic condition can be equivalently be written as
\begin{equation}
    Y(z) \begin{pmatrix}
        z^{-n} & 0 \\
        0 & z^n
    \end{pmatrix}
    \sim \mathbb{1} + O(1/z), \quad \text{as } |z| \to \infty.
\end{equation}
From here, it is easy to see that we have the following asymptotic conditions on the entries of $Y(z)$:
\begin{align}
    Y_{11}(z) &\sim z^n \label{eq:AsympY11}\\
    Y_{21}(z) &\sim O(z^{n-1}) \label{eq:AsympY21}\\
    Y_{12}(z) &\sim O(z^{-n - 1}) \label{eq:AsympY12}\\
    Y_{22}(z) &\sim z^{-n} \label{eq:AsympY22}.
\end{align}
Clearly, it follows from eqs.\ \eqref{eq:JumpY11} and \eqref{eq:JumpY21} that $Y_{11}$ and $Y_{21}$ are entire functions. Therefore, the asymptotic conditions \eqref{eq:AsympY11} and \eqref{eq:AsympY21} and Liouville's theorem  imply that $Y_{11}(z)$ and $Y_{21}(z)$ are polynomials of degree $n$ and $n-1$ respectively with $Y_{11}$ being monic. Let $Y_{11}(z) = P_n(z)$ and $Y_{21}(z) = Q_{n-1}(z)$. From conditions \eqref{eq:JumpY12} and \eqref{eq:JumpY22}, we have
\begin{align}
    Y_{12}(z) &= \frac{1}{2\pi i} \int_{-1}^{1} dx\, \frac{\omega(x) P_n(x)}{x - z} \equiv C[P_n \omega](z) \\
    Y_{22}(z) &= \frac{1}{2\pi i} \int_{-1}^{1} dx\, \frac{\omega(x) Q_{n-1}(x)}{x - z} \equiv C[Q_{n-1} \omega](z).
\end{align}
Now, recall that Cauchy transforms of typical functions decay as $1/z$ for large $|z|$, whereas in our case we require $Y_{12}(z) \sim O(z^{-n - 1})$. To ensure that $Y_{12}$ has the correct asymptotics we need further conditions on $P_n$. Performing an expansion on $Y_{12}$, it is easy to see that these extra conditions are
\begin{equation}
    \int_{-1}^{1}dx\, \omega(x) P_n(x)\, x^k = 0, \quad \text{for all}\ 0 \leq k \leq n-1.
\end{equation}
This shows that $P_n$ are orthogonal with respect to $\omega$, and since $P_n$ are monic, we have $P_n = \pi_n$.

A similar argument shows that $Q_{n-1} = c_n \pi_{n-1}$. The coefficient $c_n$ is fixed by the asymptotic condition $Y_{22}(z) \sim z^{-n}$ and turns out to be $c_n = -2\pi i \gamma_{n-1}^2$. Finally note that the endpoint conditions are trivially satisfied since the weight function $\omega(x) = \sqrt{1-x^2} h(x)$ is positive and bounded on $(-1, 1)$ and the limit $\lim_{z \to 1} Y_{ij}(z)$ exists. This completes the proof.

\subsection{Steepest descent method}
In this section, we will see how to extract the large-$n$ asymptotics of the orthogonal polynomials using a version of the Deift-Zhou steepest descent method \cite{deift1993steepest} developed in \cite{kuijlaars2004riemann}. The idea is to perform a sequence of transformations on the matrix $Y(z)$ whose purpose is to simplify the problem enough so that the large-$n$ limit becomes tractable. We will then perform the appropriate sequence of inverse transformations to extract the asymptotics of $\pi_n$ and the Lanczos coefficients.

\subsubsection*{$\boldsymbol{Y} \to \boldsymbol{T}$}
We define the matrix $T(z)$ as
\beq
T(z)=\mat{2^n}{0}{0}{2^{-n}} Y(z) \mat{\phi(z)^{-n}}{0}{0}{\phi(z)^{n}},
\eeq
where 
\beq 
\phi(z) = z + \sqrt{z^2-1}.
\eeq 
The function $\phi$ is the conformal map from $\mathbb{C}/[-1,1]$ to the exterior of the unit circle and was introduced in \ref{app:conventions}. Also recall from \ref{app:conventions} that for $x \in (-1, 1)$ we defined $\phi_{\pm}(x) = x \pm i \sqrt{1 - x^2}$ so that $\phi_+(x) \phi_-(x) = 1$. Then, a straightforward calculation shows that $T(z)$ satisfies the following RHP:

\noindent\textbf{RHP for $T(z)$:}
\begin{enumerate}
    \item $T(z)$ is analytic on $\mathbb{C} \setminus [-1, 1]$.
    \item $T$ has the following jump matrix across $(-1,1)$:
    \[
    T_+(x)=T_-(x) \mat{\phi_+(x)^{-2n}}{\omega(x)}{0}{\phi_-(x)^{-2n}}.
    \]
    \item $T(z) \to \mathbb{1}+O(1/z) $ as $z \to \infty$.
    \item $T$ has the same behavior as $Y$ when $z \to \pm 1$.
\end{enumerate}
 This transformation serves two purposes. First, it simplifies the boundary conditions at $z = \infty$. Second, although the jump matrix looks complicated, the problem considerably simplifies if we can somehow continue $x$ away from the real axis. Recall from \ref{app:conventions} that $\phi(z)$ is the analytic continuation of $\phi_{\pm}(x)$ in the sense that $\phi(x \pm i0^+) = \phi_{\pm}(x)$. Now, for $z \in \mathbb{C} \setminus [-1, 1]$, $|\phi(z)| > 1$ and $\phi(z)^{-2n}$ is exponentially suppressed, and in the large-$n$ limit the jump matrix approaches the identity. This observation leads us to our next transformation.

\subsubsection*{$\boldsymbol{T} \to \boldsymbol{S}$}
This transformation is based on the following factorization of the jump matrix for $T$ (sometimes called \emph{Bruhat decomposition}):\footnote{More generally, for any matrix $M = \begin{pmatrix} a & b \\ c & d \end{pmatrix} \in GL_2(\mathbb{C})$ we have $M = \begin{pmatrix} 1 & 0\\ \dfrac{d}{b} & 1 \end{pmatrix} \begin{pmatrix} 0 & b\\ -\dfrac{\det(M)}{b} & 0 \end{pmatrix} \begin{pmatrix} 1 & 0\\ \dfrac{a}{b} & 1 \end{pmatrix}$.}
\begin{equation} \label{eq:TJumpDecomp}
    \begin{pmatrix}
        \phi_+^{-2n} & \omega \\
        0 & \phi^{-2n}_{-}
    \end{pmatrix}
    = 
    \begin{pmatrix}
        1 & 0\\
        \frac{\phi_{-}^{-2n}}{\omega} & 1
    \end{pmatrix}
    \begin{pmatrix}
        0 & \omega \\
        -\frac{1}{\omega} & 0
    \end{pmatrix}
    \begin{pmatrix}
        1 & 0\\
        \frac{\phi_{+}^{-2n}}{\omega} & 1
    \end{pmatrix}.
\end{equation}
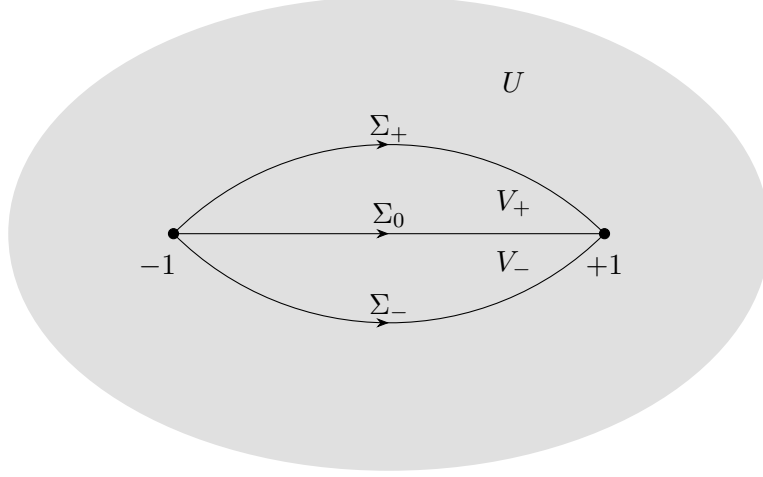
\begin{figure}[t]
    \centering
    \begin{tikzpicture}[scale=0.95]
        \coordinate (O) at (0, 0);
        \coordinate (A) at (-3, 0);
        \coordinate (B) at (3, 0);
        \coordinate (C) at (-4, 0);
        \coordinate (D) at (4, 0);
        \fill[fill=TGreyLight] (O) ellipse (5.3 and 3.3);
        \draw[postaction={decorate, decoration={markings, mark=at position 0.5 with {\arrow{Stealth}}}}] (A) -- node[midway, yshift=7pt] {$\Sigma_0$} (B);
        \filldraw[draw=black, fill=black] (A) circle (2pt) node[xshift=-6pt, yshift=-12pt] {$-1$};
        \filldraw[draw=black, fill=black] (B) circle (2pt) node[xshift=0pt, yshift=-12pt] {$+1$};
        \draw (A) edge[bend left=45, postaction={decorate, decoration={markings, mark=at position 0.5 with {\arrow{Stealth}}}}] node[midway,yshift=6pt] {$\Sigma_+$} (B);
        \draw (A) edge[bend left=-45, postaction={decorate, decoration={markings, mark=at position 0.5 with {\arrow{Stealth}}}}] node[midway, yshift=6pt] {$\Sigma_-$} (B);
        \node[xshift=-1.2cm, yshift=2cm] at (B) {$U$};
        \node[xshift=-1.2cm, yshift=0.4cm] at (B) {$V_+$};
        \node[xshift=-1.2cm, yshift=-0.4cm] at (B) {$V_-$};
    \end{tikzpicture}
    \caption{The grey shaded region is the neighbourhood $U$ and $\Sigma_0 = [-1, 1]$. The lens shaped region is the region bounded by $\Sigma_+ \cup \Sigma_-$. We will take $V_{\pm}$ to be the regions bounded between $\Sigma_0 \cup \Sigma_{\pm}$.}
    \label{fig:lens}
\end{figure}
The idea now is to interpret the jump condition on $T$ as a sequence of three jumps governed by the above factorization as one crosses the interval $[-1, 1]$ from the upper half plane. This is achieved by breaking up the contour $\Sigma_0 \equiv [-1, 1]$ in to three distinct contours $\Sigma_{+} \cup \Sigma_0 \cup \Sigma_-$, as shown in figure \ref{fig:lens}. To this end, let us choose some open neighbourhood $U$ of the strip $[-1,1]$ where the function $h(z)$ is analytic and does not have any zeros. Such a neighbourhood always exists since $h$ is real-analytic and positive on $[-1, 1]$. In particular, $\operatorname{Re} h(z) > 0$ for all $z \in U$. Inside $U$ we now draw a lens shaped region as in figure \ref{fig:lens} and define a new matrix $S(z)$ as follows:
\beq 
S(z) = \begin{cases}  T(z), & z \in \text{outside the lens} \\
 T(z) \mat{1}{0}{-\frac{1}{\omega(z)} \phi^{-2n}}{1}, & z \in  V_+ \\
 T(z) \mat{1}{0}{\frac{1}{\omega(z)} \phi^{-2n}}{1},  & z \in V_-
\end{cases}.
\eeq
This leads to the following Riemann-Hilbert problem for $S$:

\noindent\textbf{RHP for $S(z)$:}
\begin{enumerate}
    \item $S(z)$ is analytic in $\mathbb{C} \setminus (\Sigma_+ \cup \Sigma_0 \cup \Sigma_-)$.
    \item $S$ has following jump matrices across $\Sigma_+$, $\Sigma_-$ and $\Sigma_0$:
    \beqn
   S_+(z) &=& S_-(z) \mat{1}{0}{\frac{1}{\omega(z)} \phi(z)^{-2n}}{1}, \quad \text{across } \Sigma_{\pm} \\
    S_+(x)&=&S_-(x) \mat{0}{\omega(x)}{-\frac{1}{\omega(x)}}{0}, \quad \text{across } \Sigma_0.
    \eeqn
    \item $S(z) = \mathbb{1}+O(1/z) $ as $z \to \infty$.
    \item As $z\to \pm 1$,
     \begin{eqnarray}
         S(z)&=& O\mat{1}{1}{1}{1}, \quad \text{from outside the lens} \\
         &=&O\mat{|z\mp1|^{-1/2}}{1}{|z\mp1|^{-1/2}}{1}, \quad \text{from inside the lens.}
     \end{eqnarray}
\end{enumerate}
In deriving the jump conditions, we have used the factorization formula eq.\ \eqref{eq:TJumpDecomp}.
The benefit of this transformation is that now $\phi(z)^{-2n}$, in the jump matrix, is an exponentially suppressed quantity since it lies outside the strip $[-1,1]$. Furthermore, the jump matrix across $\Sigma_0$ is significantly simpler. 

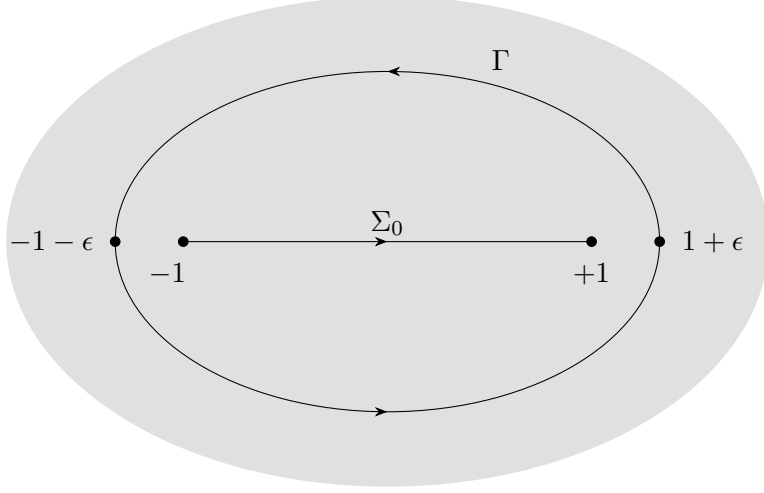
\begin{figure}[t]
    \centering
    \begin{tikzpicture}[scale=0.9]
        \coordinate (O) at (0, 0);
        \coordinate (A) at (-3, 0);
        \coordinate (B) at (3, 0);
        \coordinate (C) at (-4, 0);
        \coordinate (D) at (4, 0);
        \fill[fill=TGreyLight] (O) ellipse (5.6 and 3.6);
        \draw[postaction={decorate, decoration={markings, mark=at position 0.5 with {\arrow{Stealth}}}}] ($(O)+(4,0)$) arc[start angle=0, end angle=180, x radius=4, y radius=2.5];
        \draw[postaction={decorate, decoration={markings, mark=at position 0.5 with {\arrow{Stealth}}}}] ($(O)+(-4,0)$) arc[start angle=180, end angle=360, x radius=4, y radius=2.5];
        \draw[postaction={decorate, decoration={markings, mark=at position 0.5 with {\arrow{Stealth}}}}] (A) -- node[midway, yshift=7pt] {$\Sigma_0$} (B);
        \filldraw[draw=black, fill=black] (A) circle (2pt) node[xshift=-6pt, yshift=-12pt] {$-1$};
        \filldraw[draw=black, fill=black] (B) circle (2pt) node[xshift=0pt, yshift=-12pt] {$+1$};
        \filldraw[draw=black, fill=black] (C) circle (2pt) node[xshift=-24pt] {$-1 - \epsilon$};
        \filldraw[draw=black, fill=black] (D) circle (2pt) node[xshift=20pt] {$1 + \epsilon$};
        \node[xshift=-1.2cm, yshift=2.4cm] at (B) {$\Gamma$};
    \end{tikzpicture}
    \caption{Contour deformation of $\Sigma_0 \cup \Sigma_+ \cup \Sigma_-$ to $\Sigma_0 \cup \Gamma$. Note that we have dropped the segments $(-1 - \epsilon, -1)$ and $(1, 1 + \epsilon)$ since $S(z)$ is continuous across those segments.}
    \label{fig:lens-deformed}
\end{figure}

For the special case of square root edges that we are working with, we can further deform the lens shaped region as shown and explained in figure \ref{fig:lens-deformed}. By deforming the region in this way, we obtain a new RHP for $S(z)$:

\noindent\textbf{Deformed RHP for $S(z)$:}
\begin{enumerate}
    \item $S(z)$ is analytic in $\mathbb{C} \setminus (\Gamma \cup \Sigma_0 )$.
    \item $S(z) = \mathbb{1}+O(1/z) $ as $z \to \infty$.
    \item $S$ has following jump matrices across $\Gamma$ and  $\Sigma_0$:
    \begin{eqnarray*}
    S_+(z)&=&S_-(z) \mat{1}{0}{\pm\frac{1}{\omega(z)} \phi(z)^{-2n}}{1} \ \ \ \text{for $z \in \Gamma$ ... (+ for $\text{Im}(z)<0$)}  \\ 
    S_+(x)&=&S_-(x) \mat{0}{\omega(x)}{-\frac{1}{\omega(x)}}{0} \ \ \ \text{for $x\in\Sigma_0$}.
    \end{eqnarray*}
    \item As $z \to \pm 1$:
    \[
    O\mat{|z\mp1|^{-1/2}}{1}{|z\mp 1|^{-1/2}}{1}.
    \]
\end{enumerate}

\subsubsection*{$\boldsymbol{S \to R}$}
We are now ready to introduce the final transformation. This is given by:
\beq\label{eq:StoR}
R(z)=S(z)N(z)^{-1},
\eeq
where the matrix $N$ satisfies the following RHP:

\noindent\textbf{RHP for $N(z)$:}
\begin{enumerate}[topsep=0pt]
    \item $N(z)$ is analytic in $\mathbb{C}- \Sigma_0$.
    \item $N(z) = \mathbb{1}+O(1/z) $ as $z \to \infty$.
    \item $N$ has the following jump matrix across $ \Sigma_0 $:
    \begin{eqnarray*}
    N_+(x)&=&N_-(x) \mat{0}{\omega(x)}{-\frac{1}{\omega(x)}}{0} \ \ \ \text{for $x\in \Sigma_0$}.
    \end{eqnarray*}
    \item $N$ has the same behavior as $S$ for $z \to \pm 1$.
\end{enumerate}
By a similar argument to the one that guaranteed uniqueness of the solution for the RHP for $Y$, the above RHP for $N$ must also have a unique solution. Moreover, direct calculation shows that:
\beq
N(z)=\mat{D_{\infty}}{0}{0}{D_{\infty}^{-1}} \mat{\frac{a(z)+a(z)^{-1}}{2}}{\frac{a(z)-a(z)^{-1}}{2i}}{\frac{a(z)-a(z)^{-1}}{-2i}}{\frac{a(z)+a(z)^{-1}}{2}} \mat{D(z)^{-1}}{0}{0}{D(z)},
\eeq
is a solution for this RHP, and hence by uniqueness this is the only solution. The functions $a(z)$ and $D(z)$ and their analytic properties have been previously defined in \ref{app:conventions}. Due to the appearance of $N^{-1}$ in the transformation \eqref{eq:StoR}, the jump matrix across the $[-1,1]$ cut becomes $\mathbb{1}$. We thus get the following RHP for $R$:

\noindent\textbf{RHP for $R(z)$:}
\begin{enumerate}
    \item $R(z)$ is analytic everywhere in the complex plane except on $\Gamma$.
    \item $R(z)\to \mathbb{1}$ as $z\to\infty$.
    \item $R$ has the following jump matrix across $\gamma$:
    \beqn
    R_+(s)&=&R_-(s)V(s)
    \eeqn
    where,
    \beqn
    V(s)&=&N(s) \mat{1}{0}{\frac{1}{w(s)} \phi(s)^{-2n}}{1} N(s)^{-1}\;\cdots \;( s \in \Gamma \cap {\text{Im}(z)<0} )\\
    &=& N(s) \mat{1}{0}{\frac{1}{-w(s)} \phi(s)^{-2n}}{1} N(s)^{-1} \;\cdots\;( s \in \gamma \cap {\text{Im}(z)>0}).
    \eeqn
\end{enumerate}
Using the explicit form of $N$, we can write $V(s)$ as:
\beq
V(s)=\mathbb{1}+\Delta(s)
\eeq
where,
\beq
\Delta(s)=\pm \frac{D(s)^2 \phi(s)^{-2n}}{w(s)}\mathcal{A}(s),
\eeq\label{eq:delta}
with $+$ sign for $s \in \Gamma \cap {\text{Im}(z)<0}$ and $-$ sign for $s \in \Gamma \cap {\text{Im}(z)>0}$, and 
\beq 
\mathcal{A}(s)=\mat{\frac{a(s)^2-a(s)^{-2}}{4i}}{\frac{(a(s)-a(s)^{-1})^2}{4} D_{\infty}^{2}}{\frac{(a(s)+a(s)^{-1})^{2}}{4} D_{\infty}^{-2}}{\frac{a(s)^2-a(s)^{-2}}{-4i}}.
\eeq
Written in this form, $\Delta(s)$ appears to have a branch cut along $(-\infty,1] \cup [1,\infty)$. But this is not true as this apparent branch cut gets canceled by the branch cut in the analytic continuation of $\om(s)$ away from $[-1,1]$ \footnote{This is the reason why the simplification was possible for $\al,\beta=\pm\frac12$}. Therefore, using equations (\ref{eq:def_omega}), (\ref{eq:D_J}) and (\ref{eq:fact_J}) we get:
\beq\label{eq:def delta}
\Delta(s)=-\frac{i \ D_h(s)^2 \phi(s)^{-2n-1}}{h(s)}\mathcal{A}(s),
\eeq
In the large $n$ limit, $\Delta(s)$ is small and of $O(e^{-cn})$. The jump matrix for $R(z)$ is identity as $n\to\infty$, so $R(z)$ at $n=\infty$ (which we will denote as $R^{(0)}(z)$) is an entire function, and therefore by Liouville's theorem $R^{(0)}(z)=\mathbb{1}$. In what follows, we will show that tracking this back to $Y$ implies that the Lanczos coefficients as $n\to \infty$ are given by $\ba_n=0$ and $\bb_n=1/2$. 


\subsection{Solving the RHP for $R$ perturbatively} \label{app:RHP for R}
In this subsection, we show that the RHP for $R(z)$ can be solved perturbatively at large $n$ in various cases. First, we will derive a general formula for the leading perturbative correction to $R$ in the large $n$ limit, but without any low-energy/continuum limit. After this, we will show how to solve the RHP for $R$ direclty in the low-energy continuum limit $\lambda\to 0$, where $\lambda$ is related to the location of the closest pair of zeros of $h$. Here we will consider the two cases discussed in the main text: in case (i), the number of zeros of $h(z)$ in the complex plane are finite and $h(z)$ falls off at most polynomially at complex infinity. In case (ii), the spectral density takes the form:
\beq
\om(E)=\om_{\text{DSSYK}} (E) \ \mathcal{N} \  \frac{\prod_{i=1}^{P}(p_i^2-E^2)}{\prod_{j=1}^{Q}(q_j^2-E^2)},
\eeq
where the zeros $\{p_i\}$ and the poles $\{q_j\}$ are not scaled in the $\lm \to 0$ limit.
\subsubsection*{Large $n$ limit}
At large $n$, we can solve for $R$ order by order in $\Delta$. This is possible since $\Delta$ is uniformly suppressed in $e^{-n}$ outside the interval $[-1,1]$. We can solve for $R(z)$ by expanding it as:
\beq
R(z) = \mathbb{1}+ R^{(1)}(z) + \cdots.
\eeq 
By substituting this in RHP for $R(z)$ and keeping the terms that are $O(\Delta(s))$ in the jump matrix, we get the following RHP for $R^{(1)}(z)$.

\textbf{RHP for $R^{(1)}(z)$:}
\begin{enumerate}
    \item $R^{(1)}(z)$ is analytic everywhere in the complex plane except on $\Gamma$.
    \item $R^{(1)}(z)\to 0$ as $z\to\infty$.
    \item It has the following jump matrix across $\Gamma$:
    \beqn
     R^{(1)}_+(s)&=&R^{(1)}_-(s)+\Delta(s).
    \eeqn
\end{enumerate}
The main simplification that this provides is that $R^{(1)}(z)$ follows an additive RHP instead of the multiplicative RHP satisfied by $R(z)$. Any additive RHP admits the following simple solution:
\beq \label{eq:R1integrand}
R^{(1)}(z)=\frac{1}{2 \pi i}\oint_{\Gamma} ds \frac{\Delta(s)}{s-z}.
\eeq
In order to calculate this integral, we will do a contour deformation in the $s$ plane. To this end, let us first look at the analytic properties of the function $\Delta(s)$. $\Delta(s)$ has a branch cut across $[-1,1]$ inside $\Gamma$ coming from branch cuts of $D_h(s)$, $\phi(s)$ and $a(s)$. The analytic property of this function outside $\Gamma$ depends on the analytic continuation of $h(s)$. If $h(s)$ has zeros in the complex $s$ plane outside $\Gamma$, then $\frac{1}{h(s)}$ will have poles at those points, and while deforming the $\Gamma$-contour to $\infty$ we will pick up contributions from the residues of these poles. In addition, we have also assumed that $\frac{1}{h}$ vanishes sufficiently fast as $s \to \infty$ so that we can drop the contour at infinity. More precisely, the integrand in equation (\ref{eq:R1integrand}) goes as $\frac1{h(s) \ s^{2n+2}}$ as $s \to \infty$. This implies that $h(s)$ should not decay faster than $s^{-2n-1}$ at infinity. Moreover, since $n$ can be taken to be arbitrarily large in the present analysis, any polynomial decay at infinity would fall under this category. For any $h(s)$ of this type, we have: 
\beq\label{eq:solR1}
R^{(1)}(z)=\sum_{s_i} \text{Res}_{s=s_i} \left(\frac1{h(s)} \frac{i \ D_h(s)^2}{(s-z)} \phi(s)^{-2n-1} \mathcal{A}(s)\right) - \frac{i \ D_h(z)^2 \phi(z)^{-2n-1} \mathcal{A}(z)}{h(z)} ,
\eeq
where the first term above comes from the zeros of $h$ and the second term comes from the pole at $s=z$, assuming $z$ lies outside $\Gamma$; we will later be interested in taking $z$ large. In the large $n$ limit, the leading term comes from the residue at the set of points $s_j$ for which $\phi(s_j)$ has the least magnitude, i.e., the closest zeros of $h$. Therefore,
\beq\label{eq:exp R}
R^{(1)}(z)=\sum_{s_j=\pm s_0} \text{Res}_{s= s_j} \left(\frac1{h(s)} \frac{i \ D_h(s)^2}{\phi(s)^{2n+1}(s-z)}  \mathcal{A}(s)\right).
\eeq
It is clear from equation (\ref{eq:exp R}) that if the poles at $\pm s_0$ are simple, then $R^{(1)}(z)$ would be suppressed by a factor of $e^{-2cn}$, this is what we are after in this paper. However, if the order of pole $a$ is greater than $1$, then $R^{(1)}(z)$ would go as $n^{a-1} e^{-2 cn}$. Of course, this does not create any trouble for the large $n$ limit, but could potentially be problematic for the continuum limit. For this reason, it is convenient to restrict ourselves to the case of simple zeros in $h$. In this case:
\beq\label{eq:exp R simple}
R^{(1)}(z)=e^{-2 c n}\sum_{s_j = \pm s_0} \text{Res}_{s=  s_j} \left(\frac1{h(s)}\right) \frac{i \ D_h(s_j)^2}{\phi( s_j)( s_j-z)}  \mathcal{A}( s_j),
\eeq
where, $c=\text{log}\left(s_j+\sqrt{s_j^2-1}\right)$. This leads to an exponential correction in the monic polynomials, and in turn, an exponential correction in the Lanczos coefficients.

We now turn to the solution of the RHP in the continuum limit. 

\subsubsection*{Continuum limit, case (i): finitely many zeros and poles}

While the arguments in the previous section work well for large $n$, the continuum limit is somewhat more delicate because here we are taking $n$ large but at the same time also sending $\lambda \to 0$ in a coordinated way, so it is not completely immediate that $\Delta(s)$ is uniformly small on and outside $\Gamma$. In order to make progress, let us first consider the following specific class of spectral density functions considered in section \ref{sec:finitepoles}:
\beq
\om(E)=\mathcal{N} \sqrt{1-E^2} \prod_{k=0}^{Q-1}(q_k^2-E^2) \ \frac{\prod_{i=0}^{P-1}(p_i^2-E^2)}{\prod_{j=0}^{R-1}(r_j^2-E^2)}.
\eeq
Here, the $q_k$'s are to be scaled in the $\lm\to 0$ limit, while the $\{p_i\}$ and $\{r_j\}$ do not scale. We parameterize the scaled zeros as $q_i=\text{cosh}\left(\frac{\lm f(i)}{2}\right)$, where, $f(i)>f(j)$ for $i>j$ and $f(0)=1$.
From equation \eqref{eq:def delta}, we get:
\beq \label{eq:finiteDelta}
\Delta(s)=-i \ \phi(s)^{-2n-1} \prod_{i=1}^Q k(q_i,s) \frac{\prod_{i=1}^P k(p_i,s)}{\prod_{i=1}^R k(r_i,s)} \mathcal{A}(s),
\eeq
where we have used the factorization property of the Szeg\H{o} function and 
\beq 
k(x,s):=\frac{\left(\phi(x)^2-\phi(s)^{-2} \right)^2}{4 \phi(x)^2 (x^2-s^2)}.
\eeq 
In the continuum limit, the scaled zeros of $h$ at $\{q_k\}$ start approaching the cut $[-1,1],$ and consequently, so also does the contour $\Gamma$, since $\Gamma$ was chosen to lie before we hit the first zero. It is convenient to zoom in near the cut by introducing the new variable $t$ defined as:
\beq
s=\text{cosh}\left(\frac{\lm t}{2} \right).
\eeq
Therefore, in the limit $\lm\to0$, we get
\beq
\phi(s)^{-2 n}=e^{-\lambda n t}= \lm^{2 t}e^{-t\ell },
\eeq
where, in the second equality we have used equation \eqref{eq:scaling}. In terms of the new variable, $\Delta$ becomes:
\beq
\Delta(t) = -i \lm^{2 t} e^{-t\ell} \prod_{i=1}^Q \frac{1-e^{-\lm t} e^{-f(i)\lm}}{1-e^{\lm t} e^{-f(i)\lm}} \ \frac{\prod_{i=1}^P k(p_i)}{ \prod_{i=1}^R k(r_i) \phi(s)} \mathcal{B}(t),
\eeq
where,
\beq
\mathcal{B}(t)=\frac{1}{\lm t} \mat{i}{1}{1}{-i}.
\eeq
Upon simplifying in the $\lm \to 0$ limit, this gives
\beq\label{eq:finite q delta}
\Delta(t)=-i \lm^{2t-1} e^{-t\ell} \prod_{i=1}^Q \left(\frac{f(i)+t}{f(i)-t}\right) \frac{1}{t} \mat{i}{1}{1}{-i}.
\eeq
So, it is clear that if $Q$ is finite, then $\Delta$ is uniformly small on $\Gamma$.\footnote{Note that this is true for all $0<t<1$; the -1 in the power of $\lm^{2t-1}$ is compensated for by the fact that $\Delta$ is a density, i.e., $ds \Delta(s) = dt \Delta(t)$.} This implies that the perturbative approach to solving the RHP is valid and the leading order contribution should come from the nearest zeros. Hence to the leading order we have:
\beq \label{eq:scaled Rsol}
R^{(1)}(z)= \lm^{2}e^{-\ell} \sum_{s_j=\pm s_0} \text{Res}_{s=s_j} \left(\frac1{h(s)}\right) \frac{i \ D_h(s_j)^2}{\phi(s_j)(s_j-z)}  \mathcal{A}(s_j).
\eeq

\subsubsection*{Continuum limit, case (ii): Christoffel deformations of DSSYK}
We will now work out the Christoffel deformations of DSSYK from the RHP method in the continuum limit. But before looking at Christoffel deformations of DSSYK, let us look at the RHP of $R(z)$ for DSSYK itself. Here, the continuum limit is a bit subtle, and we will need to take the continuum limit in a slightly different way -- we first define the $\alpha$-continuum limit in the following way:
\beq
n= \frac{1}{\lm}\log\left(\frac{1}{\lm^{2+\al}}\right) +\frac{\ell}{\lm},\;\;\;\cdots\;\;\;(\alpha-\text{continuum limit}).
\eeq
In the end, we will take the limit $\alpha \to 0^+$. The purpose of introducing $\alpha$ will become clear in what follows. Notice that we can read off the jump matrix for DSSYK from equation \eqref{eq:finiteDelta}, with $Q=\infty, \ P=0$ and $R=0$, this gives:
\beq
\Delta(s)=-i \ \phi(s)^{-2n-1} \prod_{i=1}^{\infty} \frac{(\phi(q_i)^2-\phi(s)^{-2})^2}{4 \ \phi(q_i)^2 (q_i^2-s^2)} \mathcal{A}(s),
\eeq
where, $q_k=\text{cosh}\left( \frac{k \lm}{2}\right)$. Now taking the $\alpha$-continuum limit, we get:
\beq
\Delta(t)=-i \lm^{(2+\al)t-1} e^{-t \ell} \prod_{i=1}^{\infty} \frac{1-e^{-\lm t} e^{-i \lm}}{1-e^{\lm t} e^{-i\lm}} \frac{1}{ t} \mat{i}{1}{1}{-i}.
\eeq
It is easy to check that the infinite product above is given by
\beq
\prod_{i=1}^{\infty} \frac{1-e^{-\lm t} e^{-i \lm}}{1-e^{\lm t} e^{-i\lm}} = \frac{(\bq^{1+t};\bq)_{\infty}}{(\bq^{1-t};\bq)_{\infty}},
\eeq
where, $\bq=e^{-\lm}$. Using the definition of q-Gamma functions, $\Gamma_\bq(x)=\frac{(\bq;\bq)_{\infty}}{(\bq^x;\bq)_{\infty}} (1-\bq)^{1-x}$, in the $\lm\to 0$ limit, we get:
\beq
\prod_{i=1}^{\infty} \frac{1-e^{-\lm t} e^{-i \lm}}{1-e^{\lm t} e^{-i\lm}} = \lm^{-2 t} \frac{\Gamma_\bq(1-t)}{\Gamma_\bq(1+t)},
\eeq
and therefore in the limit $\lm \to 0$ we get:
\beq
\Delta(t)=-i \lm^{\al t-1} e^{-t \ell} \ \frac{\Gamma(1-t)}{\Gamma(1+t)} \frac{1}{t} \mat{i}{1}{1}{-i}.
\eeq
The $-1$ in the power of $\lm$ is spurious because $\Delta(t)$ is a density and we need $dt \Delta(t)=ds \Delta(s)$, so in the end we can see that we get an overall suppression of $\lm^{\al t}$. Note that the perturbative approach to solving the RHP works in this case as long as $\alpha > 0$; if we had set $\alpha=0$, then perturbation theory would not be justified. With $\alpha > 0$, the suppression of $\lm^{\al t}$ is greater for larger values of $t$, therefore, the contribution to $R^{(1)}(z)$ still comes from the nearest pole. And hence, $R(z)$ in this case is just given by taking the continuum limit of the $R(z)$ obtained from the large-n limit. 

Now, consider Christoffel deformtions of DSSYK:
\beq
\om(E)=\om_{\text{DSSYK}} (E)\, \wh(E),
\eeq
where $\wh(E)=\mathcal{N} \frac{\prod_{i=1}^{P}(p_i^2-E^2)}{\prod_{i=1}^{R}(r_i^2-E^2)}$, and $\mathcal{N}$ is an overall normalization constant. Using the factorization property of the Szeg\H{o} function, we get
\beq
D_h(z)=D_{h_{\text{DSSYK}}} D_{\wh}(z).
\eeq
Using this and equation \eqref{eq:def delta}, we find
\beq
\Delta(s)=\frac{\prod_{i=1}^P k(p_i,s)}{\prod_{i=1}^R k(r_i,s)} \ \mathcal{D}\cdot \Delta_{\text{DSSYK}}(s) \cdot\mathcal{D}^{-1},
\eeq
where
\beq
\mathcal{D}=\mat{D_{\wh,\infty}}{0}{0}{D_{\wh,\infty}^{-1}},
\eeq
with $D_{\wh,\infty}$ being the value of the function $D_{\wh}(z)$ as $z\to \infty$. The matrix $\mathcal{D}$ comes from the $D_{\infty}$ factors in the off-diagonal terms in $\mathcal{A}$ and the factor of $\prod_{i=1}^P k(p_i,s)$ comes from $D_h(s)^2/h(s)$.

We now introduce a new matrix $\wR$ defined as $\wR(z)=\mathcal{D}^{-1} R(z) \mathcal{D}$. It is simple to see that this new matrix more or less has the same RHP as $R(z)$, but with a new jump matrix. The new jump matrix is given by:
\beq
\widetilde{V}(s)=\mathbb{1}+\left(\frac{\prod_{i=1}^P k(p_i,s)}{\prod_{i=1}^R k(r_i,s)}\right) \ \Delta_{\text{DSSYK}}(s).
\eeq
This can of course be rewritten as:
\beq \label{eq:jump}
\widetilde{V}(s)=\mathbb{1}+\Delta_{\text{DSSYK}}(s)+ \mathcal{K}(s) \Delta_{\text{DSSYK}}(s),
\eeq
where, 
$$\mathcal{K}(s)=\left(\frac{\prod_{i=1}^P k(p_i,s)}{\prod_{i=1}^R k(r_i,s)}\right)-1.$$ 
Since, $\prod_{i=1}^P k(p_i,s) \sim 1+O(\lm),$ it is clear that to the leading order in $\lm$ the solution to the RHP of $R(z)$ is unchanged when we add a finite number of additional unscaled poles and zeros to the DSSYK spectral density. Therefore, the Lanczos coefficients must also remain unchanged under such deformations.
To calculate the leading order correction to these quantities, notice that in the $\alpha$-continuum limit, we get
\beq
R(z)=\mathbb{1}+\frac{1}{2 \pi i} \oint_{\Gamma} ds \frac{\Delta_{\text{DSSYK}}(s)(1+\mathcal{K}(s))}{s-z}.
\eeq
Picking up the contribution from the poles at $s=\pm \bq$, we get:
\beqn
\tilde{R}(z)&=&\mathbb{1}+R^{(1)}_{\text{DSSYK}}(z) (1+\mathcal{K}(\bq)) \\
&=&\mathbb{1}+R^{(1)}_{\text{DSSYK}}(z) \frac{\prod_{i=1}^{P}\left[\frac{(e^{\kappa_i}-e^{-\lm})^2}{4 e^{\kappa_i}\left(\cosh\left(\frac{\kappa_i}{2}\right)^2- \cosh\left(\frac{\lm_i}{2}\right)^2\right)}\right]}{\prod_{i=1}^{R}\left[\frac{(e^{\mu_i}-e^{-\lm})^2}{4 e^{\mu_i}\left(\cosh\left(\frac{\mu_i}{2}\right)^2- \cosh\left(\frac{\lm_i}{2}\right)^2\right)}\right]},
\eeqn
where we have parameterized $p_i\equiv \cosh \left(\frac{\kappa_i}{2}\right)$ and $r_i\equiv \cosh \left(\frac{\mu_i}{2}\right)$. Expanding the additional factor in powers of $\lm$, we get:
\beq
\tilde{R}(z)=\mathbb{1}+R^{(1)}_{\text{DSSYK}}(z) \frac{\prod_{i=1}^P\left(1+\frac{2 \lm}{e^{\kappa_i}-1}\right)}{\prod_{i=1}^R\left(1+\frac{2 \lm}{e^{\mu_i}-1}\right)}
\eeq
and therefore:
\beq \label{eq:R for Christ}
R(z)=\mathbb{1}+\mathcal{D}R^{(1)}_{\text{DSSYK}}(z) \mathcal{D}^{-1}\frac{\prod_{i=1}^P\left(1+\frac{2 \lm}{e^{\kappa_i}-1}\right)}{\prod_{i=1}^R\left(1+\frac{2 \lm}{e^{\mu_i}-1}\right)}.
\eeq

\subsection{Asymptotics for Lanczos coefficients} \label{app:asymp_lanc}
We will now calculate the Lanczos coefficients. 
\subsubsection{Large-$n$ limit}

We know that up to $O(e^{-c\,n})$ corrections, $R(z)=\mathbb{1}$. Using this, we can go back to the original matrix $Y(z)$ by going through through the sequence of inverse transformations $R \to S \to T \to Y$. From the first element of the matrix $Y(Z)$, we get the following:
\beq\label{eq:monicpi}
\pi_n(z)=\left(\frac{\phi(z)}{2}\right)^n \left((\frac{a(z)+a(z)^{-1})}{2} \frac{D_{\infty}}{D(z)} R_{11}(z) + i \frac{(a(z)-a(z)^{-1})}{2} \frac{R_{12}(z)}{D_{\infty} D(z)}\right).
\eeq
Hence, to leading order we get:
\beq\label{eq:leadpi}
\pi_n^{(0)}(z)=\left(\frac{\phi(z)}{2} \right)^{n} \frac{D_{\infty}}{D(z)} \frac{(a(z)+a(z)^{-1})}{2}.
\eeq
We can insert this into the recursion relations described in equation \eqref{eq:rec pi}. This gives:
\beq\label{eq:rec phi}
z \ \left(\frac{\phi(z)}{2} \right)^{n}= \left(\frac{\phi(z)}{2} \right)^{n+1}+\overline{a}_n \left(\frac{\phi(z)}{2} \right)^{n}+\overline{b}_n^2 \left(\frac{\phi(z)}{2} \right)^{n-1}.
\eeq
Expanding
$\left(\frac{\phi(z)}{2} \right)^{n}$ around $z = \infty$ as 
$$\left(\frac{\phi(z)}{2}\right)^{n}=z^n+c_nz^{n-1}+d_n z^{n-2}+O(z^{n-3}),$$ 
we get $c_n=0$ and $d_n=-n/4$. Substituting this in equation \eqref{eq:rec phi}, we find:
\beqn\label{eq:lanczos}
\overline{a}_n&=&c_n-c_{n+1}, \\
\overline{b}_n^2&=&d_n-d_{n+1}-a_nc_n.
\eeqn
So, to the leading order we have, $\overline{a}_n=0$ and $\overline{b}_n^2=\frac{1}{4}$. 


At the next order, we take the first correction $R^{(1)}$ given in equation \eqref{eq:exp R} and push it through the sequence of inverse transformations $R\to S\to T\to Y$ to obtain:
\beq
\pi_n^{(1)}(z)=\left(\frac{\phi(z)}2 \right)^n \frac{D_{\infty}}{D(z)} \left( \frac{(a(z)-a(z)^{-1})}{2 \ D_{\infty}^2} \ i \  R^{(1)}_{12}(z) + \frac{(a(z)+a(z)^{-1})}{2 } R^{(1)}_{11}(z) \right).
\eeq
Substituting equation \eqref{eq:exp R}, we get:
\beq\label{eq:corrpi}
\pi_n^{(1)}(z)=e^{-2cn} \sum_{s_j=\pm s_0} g(s_j) \left(\frac{\phi(z)}2 \right)^n \frac{D_{\infty}}{D(z)} \left(\frac{a(z)}{a(s_j)}+ \frac{a(s_j)}{a(z)} \right) \frac{(a(s_j)-a(s_j)^{-1})}{s_j-z} ,
\eeq
where,
\beq\label{eq:def g}
g(s_j)=\text{Res}_{s=s_j} \left(\frac1{h(s)}\right) \frac{\ D_h(s_j)^2}{4 \ \phi(s_j)}.
\eeq
Expanding the polynomial $\pi_n(z)$ around $z=\infty$ and substituting in the recursion relations, we get:
\beqn
\overline{a}_n&=&- e^{-2cn} (1-e^{-2c}) \sum_{s_j} g(s_j) \left(a(s_j)^2-\frac1{a(s_j)^2} \right),  \\
\overline{b}_n^2&=&\frac14-e^{-2cn} (1-e^{-2c}) \sum_{s_j} g(s_j)\left[\frac{(a(s_j)-a(s_j)^{-1})^2}{2}+s_j \left(a(s_j)^2-a(s_j)^{-2}\right)\right].
\eeqn
For the case of a symmetric potential, it is clear from equation \eqref{eq:def g} that:
\beq
g(-s)=g(s).
\eeq
Since, $D_h(s_j)$ is an even function for a symmetric potential as shown in section \ref{app:conventions} whereas $\text{Res}_{s=s_j} \left(\frac1{h(s)}\right)$ and $\phi(s_j)$ are odd functions under $s_j \to -s_j$. Also:
\beq
a(-s)=a(s)^{-1}
\eeq
Using these properties, it is clear that $\overline{a}_n=0$ (as should be the case for an even spectral density) and
\beqn 
\overline{b}_n^2 = \frac14 - 4 \ |g(s_0)| \ e^{-2cn}.
\eeqn

Let us consider some examples:

\noindent\textbf{Example 1 (Two zeros):} Consider the spectral density:
\beq
\omega(E)=\mathcal{N} \sqrt{1-E^2} (q^2-E^2).
\eeq
For this theory we get:
\beq
\frac{D_h(s)^2}{h(s)}=\frac{ \left(\phi(q)^2-\phi(s)^{-2}\right)^2}{4 \ \phi(q)^2 (q^2-s^2)}.
\eeq
Substituting this into equation (\ref{eq:def g}), we get:
\beq
g(s_0)= - \frac{\left( e^{\lm}-e^{-\lm}\right)^2}{32 \ e^{\frac{3\lm}{2}} \ \text{cosh}\left(\frac{\lm}{2}\right)},
\eeq
where
\beq
q = \text{cosh}\left(\frac{\lambda}{2}\right).
\eeq
Thus,
\beq
\overline{b}_n=\frac12-\frac{\left( e^{\lm}-e^{-\lm}\right)^2}{8 \ e^{\frac{3\lm}{2}} \ \text{cosh}\left(\frac{\lm}{2}\right)} e^{-\lm n}+\cdots.
\eeq

\noindent \textbf{Example 2 (Multiple zeros):} Consider a spectral density with multiple poles and zeros:
\beq
\om(E)=\mathcal{N} \sqrt{1-E^2} \frac{\prod_{i=1}^Q (q_i^2-E^2)}{\prod_{i=1}^R (r_i^2-E^2)}.
\eeq
In this case, we have:
\beq
\frac{D_h(s)^2}{h(s)}=\frac{\prod_{i=1}^N \left[\frac{ \left(\phi(q_i)^2-\phi(s)^{-2}\right)^2}{4 \ \phi(q_i)^2 (q_i^2-s^2)}\right]}{\prod_{j=1}^M \left[\frac{ \left(\phi(p_j)^2-\phi(s)^{-2}\right)^2}{4 \ \phi(p_i)^2 (p_i^2-s^2)}\right]}.
\eeq
Substituting this in equation (\ref{eq:def g}), we get:
\beq
g(s_0)= - \frac{\left( e^{\lm_1}-e^{-\lm_1}\right)^2}{32 \ e^{\frac{3\lm_1}{2}} \ \text{cosh}\left(\frac{\lm_1}{2}\right)} \ \frac{\prod_{i=2}^N \left(\frac{1-e^{-\lm_1} e^{-\lm_i}}{1-e^{\lm_1} e^{-\lm_i}}\right)}{\prod_{j=1}^M \left(\frac{1-e^{-\lm_1} e^{-\eta_j}}{1-e^{\lm_1} e^{-\eta_j}}\right)},
\eeq
where,
\beqn
q_i & = & \text{cosh}\left(\frac{\lambda_i}{2}\right), \\
p_j & = & \text{cosh}\left(\frac{\eta_j}{2}\right).
\eeqn
Therefore,
\beq
\overline{b}_n = \frac12-\frac{\left( e^{\lm_1}-e^{-\lm_1}\right)^2}{8 \ e^{\frac{3\lm_1}{2}} \ \text{cosh}\left(\frac{\lm_1}{2}\right)} \ \frac{\prod_{i=2}^N \left(\frac{1-e^{-\lm_1} e^{-\lm_i}}{1-e^{\lm_1} e^{-\lm_i}}\right)}{\prod_{j=1}^M \left(\frac{1-e^{-\lm_1} e^{-\eta_j}}{1-e^{\lm_1} e^{-\eta_j}}\right)} e^{- \lm_1 n} + \cdots.
\eeq

\noindent\textbf{Example 3 (Double scaled SYK):} The DSSYK spectral density also has a similar form, but with infinitely many zeros at $q_k = \pm \cosh(\frac{k\lm}{2})$ with $k > 1$. Using the formulas from the previous example, we get:
\beq
g(s_0)= - \frac{\left( e^{\lm}-e^{-\lm}\right)^2}{32 \ e^{\frac{3\lm}{2}} \ \text{cosh}\left(\frac{\lm}{2}\right)} \prod_{k=2}^{\infty} \left(\frac{1- e^{-(k+1) \lm}}{1- e^{-(k-1) \lm}}\right).
\eeq
The infinite product gets simplified due to cancellations between the numerator and denominator:
\beq
\prod_{k=2}^{\infty} \left(\frac{1- e^{-(k+1) \lm}}{1- e^{-(k-1) \lm}}\right)= \frac{1}{(1-e^{-\lm}) (1-e^{-2 \lm})}.
\eeq
From here, we find
\beq
g(s_0)=-\frac1{16},
\eeq
and therefore, 
\beqn
\overline{b}^2_n &=& \frac{1}{4}-\frac{1}{4} e^{-\lm n}+\cdots.
\eeqn
This agrees with the exact form of the Lanczos coefficients in DSSYK:
\beq 
\overline{b}_n^2 = \frac{1}{4}(1-\mathfrak{q}^n),\;\;\mathfrak{q}=e^{-\lm}.
\eeq
Remarkably, from the exact formula, we see that the leading contribution computed using the RHP method above is actually exact and does not receive further corrections at higher orders in perturbation theory.

\subsubsection{Continuum limit, case (i)}
Now, let us consider the case of finite number of scaled zeros. We can use the solution for $R(z)$ given in equation (\ref{eq:scaled Rsol}) to obtain:
\beq
\pi_n^{(1)}(z)=\left(\frac{\phi(z)}2 \right)^n \frac{D_{\infty}}{D(z)} \left( \frac{a(z)-a(z)^{-1}}{2 \ D_{\infty}^2} \ i \  R^{(1)}_{12}(z) + \frac{a(z)+a(z)^{-1}}{2 } R^{(1)}_{11}(z) \right)
\eeq
Since this is similar in structure as the unscaled zero case, we get:
\beq
\overline{b}_n = \frac12-4 \ |g(s_0)| \lm^2 e^{-\ell}
\eeq
where:
\beq
g(s_0)= - \frac{\left( e^{\lm}-e^{-\lm}\right)^2}{32 \ e^{\frac{3\lm}{2}} \ \text{cosh}\left(\frac{\lm}{2}\right)} \ \prod_{i=2}^Q \left(\frac{f(i)+1}{f(i)-1} \right) (1+O(\lm)).
\eeq
This shows that in the $\lm \to 0$ limit, $g(s_0)$ is itself suppressed by a power of $\lm^2$ and therefore to $O(\lm^2)$, we get:
\beq
\overline{b}_n= \frac12.
\eeq
\subsubsection{Continuum limit, case (ii)}
Applying the inverse transformations on $\tilde{R}(z)$ obtained in equation \eqref{eq:R for Christ}, we get:
\beq
\pi_n(z)=D_{\tilde{h},\infty} \left(\pi^{(0)}_{DSSYK}(z)+\mathcal{K}(\bq) \pi^{(1)}_{DSSYK}(z)\right). 
\eeq
Notice that the polynomial only differs by the additional factor of $\mathcal{K}(\bq)$, therefore, this effect would show up in the first correction to the lanczos coefficient as a multiplicative factor, so, we get:
\beqn
\overline{b}_n=\frac12-\mathcal{K}(\bq) \frac{\lm^{2+\alpha}}{4} e^{-\ell}.
\eeqn
Taking the $\al\to0^{+}$ limit, we get:
\beqn
\overline{b}_n=\frac12-\frac{\lm^2}{4} \mathcal{K}(\bq) e^{-\ell}.
\eeqn
But in the $\lambda \to 0$ limit, we get:
\beq
\mathcal{K}(\bq)=1+O(\lambda),
\eeq
so
\beqn
\overline{b}_n=\frac12-\frac{\lm^2}{4} e^{-\ell}.
\eeqn
To get the first correction due to the unscaled poles and zeros, we expand $\mathcal{K}(\bq)$ to leading order in $\lm$, which gives:
\beq
\mathcal{K}(\bq) = 1+2 \lm \left(\sum_{i=1}^P \frac{1}{\phi(p_i)^2-1}-\sum_{i=1}^R \frac{1}{\phi(r_i)^2-1} \right) .
\eeq
This shows that that the correction due to the unscaled poles and zeros come at $O(\lm^3)$ and is given by:
\beq\label{eq:conf}
\overline{b}_n=\frac12-\frac14 \lm^2 e^{-\ell}-\frac{\lm^3 e^{-\ell}}{2}\left(\sum_{i=1}^P \frac{1}{\phi(p_i)^2-1}-\sum_{i=1}^R \frac{1}{\phi(r_i)^2-1} \right),
\eeq
which precisely agrees with the corrections obtained using the other method in equations \eqref{eq:CDzeros} and \eqref{eq:CDpoles}.
\section{Rational Christoffel deformations} \label{app:RationalCD}
A rational Christoffel deformation of $\omega$ is given by
\begin{equation}
    \widetilde{\omega}(E) = \mathcal{N} \frac{\prod_{i=1}^k(p_i - E)}{\prod_{j=1}^l(q_j - E)}\, \omega(E),
\end{equation}
where the points $p_i$ and $q_j$ lie outside the support $[-1, 1]$ and are chosen so that $\widetilde{\omega}$ is real and positive on $[-1, 1]$. In this appendix we state the relevant theorems pertaining to such deformations. The proofs can be found in Appendix \ref{app:ChristoffelProof}.
\begin{theorem}[Rational Christoffel deformations] \label{thm:RationalCD}
    Let $\widetilde{\omega}$ be a Christoffel deformation of $\omega$ that introduces a collection of simple zeros at the points $\{p_i\}_{i=1}^k$ with $p_i > 1$, and simple poles at the points $\{q_j\}_{j=1}^l$ with $q_j > 1$. Let $\widetilde{\pi}_n(E)$ be the $n^{th}$ (monic) orthogonal polynomial with respect to $\widetilde{\omega}$. Then, if $l \leq n$, we have
    \begin{equation} \label{eq:RationalCDPoly}
        \widetilde{\pi}_n(E) = \frac{(-1)^k}{ \prod_{i=1}^k (p_i - E)}
        \frac{\det\begin{pmatrix}
            C[\pi_{n-l} \omega](q_1) & C[\pi_{n-l+1}\omega](q_1) & \dots &C[\pi_{n+k}\omega](q_1) \\
            \vdots & \vdots & \ddots & \vdots \\
            C[\pi_{n-l}\omega](q_l) & C[\pi_{n-l+1}\omega](q_l) & \dots &C[\pi_{n+k}\omega](q_l) \\
            \pi_{n-l}(p_1) & \pi_{n-l+1}(p_1) & \cdots & \pi_{n+k}(p_1) \\
            \vdots & \vdots & \ddots & \vdots \\
            \pi_{n-l}(p_k) & \pi_{n-l+1}(p_k) & \cdots & \pi_{n+k}(p_k) \\
            \pi_{n-l}(E) & \pi_{n-l+1}(E) & \cdots & \pi_{n+k}(E) 
            \end{pmatrix}}{\det
            \begin{pmatrix}
            C[\pi_{n-l}\omega](q_1) & C[\pi_{n-l+1}\omega](q_1) & \dots &C[\pi_{n+k-1}\omega](q_1) \\
            \vdots & \vdots & \ddots & \vdots \\
            C[\pi_{n-l}\omega](q_l) & C[\pi_{n-l+1}\omega](q_l) & \dots &C[\pi_{n+k-1}\omega](q_l) \\
            \pi_{n-l}(p_1) & \pi_{n-l+1}(p_1) & \cdots & \pi_{n+k-1}(p_1) \\
            \vdots & \vdots & \ddots & \vdots \\
            \pi_{n-l}(p_k) & \pi_{n-l+1}(p_k) & \cdots & \pi_{n+k-1}(p_k) \\
            \end{pmatrix}}.
    \end{equation}
    where $C[\pi_a \omega](z)$ is the Cauchy transform of $\pi_a \omega$:
    \begin{equation}
        C[\pi_a](z) = \int_{-1}^{1} dE\, \frac{\pi_a(E)\omega(E)}{z - E}.
    \end{equation}
\end{theorem}
Note that the above theorem holds only if the number of new poles does not exceed $n$ so that $C[\pi_{n-l} \omega]$ is well defined. There is a version of this theorem for the case $l > n$, but since we are primarily interested in the large-$n$ limit with $k, l$ fixed, we will not state it here. For further details see \cite{ismail2005classical}.
\begin{corollary}[Even rational Christoffel deformations] \label{cor:Cor3}
    Let $\widetilde{\omega}$ be an even rational Christoffel deformation of $\omega$ with simple zeros and poles at $\{\pm p_i\}_{i=1}^{k}$ and $\{\pm q_j\}_{j=1}^l$ respectively. Then, if $2l \leq n$, the deformed orthogonal polynomials take the form
    \begin{equation} \label{eq:RationalCDEven}
        \widetilde{\pi}_n(E) = \frac{(-1)^k}{\prod_{i=1}^k (p_i^2 - E^2)} \frac{\det M^{(n)}_{k,l}(E)}{\det \mathcal{D}^{(n)}_{k,l}}
    \end{equation}
    where we have defined
    \begin{equation}
        M^{(n)}_{k,l}(E) = 
        \begin{pmatrix}
            C[\pi_{n-2l} \omega](q_1) & C[\pi_{n-2(l+1)}\omega](q_1) & \dots &C[\pi_{n+2k}\omega](q_1) \\
            C[\pi_{n-2l} \omega](-q_1) & C[\pi_{n-2(l+1)}\omega](-q_1) & \dots &C[\pi_{n+2k}\omega](-q_1) \\
            \vdots & \vdots & \ddots & \vdots \\
            C[\pi_{n-2l}\omega](-q_l) & C[\pi_{n-2(l+1)}\omega](-q_l) & \dots &C[\pi_{n+2k}\omega](-q_l) \\
            \pi_{n-2l}(p_1) & \pi_{n-2(l+1)}(p_1) & \cdots & \pi_{n+2k}(p_1) \\
            \pi_{n-2l}(-p_1) & \pi_{n-2(l+1)}(-p_1) & \cdots & \pi_{n+2k}(-p_1) \\
            \vdots & \vdots & \ddots & \vdots \\
            \pi_{n-2l}(-p_k) & \pi_{n-2(l+1)}(-p_k) & \cdots & \pi_{n+2k}(-p_k) \\
            \pi_{n-2l}(E) & \pi_{n-2(l+1)}(E) & \cdots & \pi_{n+2k}(E) 
            \end{pmatrix}
    \end{equation}
    and
    \begin{equation}
        \mathcal{D}^{(n)}_{k,l} = \begin{pmatrix}
            C[\pi_{n-2l} \omega](q_1) & C[\pi_{n-2(l+1)} \omega](q_1) & \dots &C[\pi_{n+2(k-1)} \omega](q_1) \\
            C[\pi_{n-2l} \omega](-q_1) & C[\pi_{n-2(l+1)} \omega](-q_1) & \dots &C[\pi_{n+2(k-1)} \omega](-q_1) \\
            \vdots & \vdots & \ddots & \vdots \\
            C[\pi_{n-2l} \omega](-q_l) & C[\pi_{n-2(l+1)} \omega](-q_l) & \dots &C[\pi_{n+2(k-1)} \omega](-q_l) \\
            \pi_{n-2l}(p_1) & \pi_{n-2(l+1)}(p_1) & \cdots & \pi_{n+2(k-1)}(p_1) \\
            \pi_{n-2l}(-p_1) & \pi_{n-2(l+1)}(-p_1) & \cdots & \pi_{n+2(k-1)}(-p_1) \\
            \vdots & \vdots & \ddots & \vdots \\
            \pi_{n-2l}(-p_k) & \pi_{n-2(l+1)}(-p_k) & \cdots & \pi_{n+2(k-1)}(-p_k) \\
            \end{pmatrix}.
    \end{equation}
\end{corollary}
Now consider the case $k=1$ and $l=1$ so that the deformed density takes the form
\begin{equation}
    \widetilde{\omega}(E) = \mathcal{N} \frac{p^2 - E^2}{q^2 - E^2} \omega(E).
\end{equation}
This deformation introduces a pair of new zeros at $\pm p$ and a pair of new poles at $\pm q$. Using eq.\ \eqref{eq:RationalCDEven}, it is easy to check that the deformed orthogonal polynomials (for $n \geq 2$) are given by
\begin{equation} \label{eq:DefPolRationalEven}
    \widetilde{\pi}_n(E) = \frac{1}{E^2 - p^2} \left[ \pi_{n+2}(E) - \frac{\tau_{n-2, n+2}(p, q)}{\tau_{n-2,n}(p, q)} \pi_n(E) + \frac{\tau_{n,n+2}(p, q)}{\tau_{n-2, n}(p, q)} \pi_{n-2}(E) \right],
\end{equation}
where the functions $\tau_{a,b}(p,q)$ are defined as
\begin{equation}
    \tau_{a,b}(p,q) = C[\pi_a \omega](q) \pi_b(p) - C[\pi_b \omega](q) \pi_a(p).
\end{equation}
Let $\bar{\mathsf{b}}_n$ and $\bar{b}_n$ be the deformed and undeformed Lanczos coefficients respectively. Using the recursion relations\ \eqref{eq:DefLanczos} and \eqref{eq:UndefLanczos}, and eq.\ \eqref{eq:DefPolRationalEven} we can express $\bar{\mathsf{b}}_n$ in terms of $\bar{b}_n$. The strategy is the same as in the proof of corollary \ref{cor:Cor2} (see Appendix \ref{app:ChristoffelProof}). The final result is
\begin{equation} \label{eq:RationalCDb}
    \bar{\mathsf{b}}_n^2 = \frac{\tau_{n-2, n+2}(p,q) \tau_{n-3, n-1}(p,q)}{\tau_{n-2, n}(p,q) \tau_{n-3, n+1}(p,q)} \bar{b}_n^2 + \frac{\tau_{n-3, n-1}(p,q)}{\tau_{n-3, n+1}(p,q)} \left[ \frac{\tau_{n+1, n+3}(p,q)}{\tau_{n-1, n+1}(p,q)} - \frac{\tau_{n, n+2}(p,q)}{\tau_{n-2, n}(p,q)} \right]
\end{equation}
More generally, for $k,l > 1$ one can show that $\bar{\mathsf{b}}_n^2$ and $\bar{b}_n^2$ are linearly related, i.e. $$\bar{\mathsf{b}}_n^2 = \alpha_n(p,q)\, \bar{b}_n^2 + \beta_n(p,q).$$ Just as in the polynomial case (corollary \ref{cor:Cor2}), the coefficients $\alpha_n(p,q)$ and $\beta_n(p,q)$ can be expressed in terms of the appropriate sub-determinants of eq.\ \eqref{eq:RationalCDEven}.

\section{Proofs of the Christoffel deformation theorem and its corollaries} \label{app:ChristoffelProof}
This appendix contains the proofs of the Christoffel deformation theorems of section \ref{sec:CD} and Appendix \ref{app:RationalCD}. Interested readers can find further details in \cite{szeg1939orthogonal} (see also \cite{ismail2005classical}).
\subsection*{Proof of theorem \ref{thm:Christoffel}}
We closely follow \cite{szeg1939orthogonal}. We consider a Christoffel deformation of a compactly supported weight function $\omega(E)$,
\begin{equation}
    \widetilde{\omega}(E) = \mathcal{N} \prod_{i=1}^k (p_i - E)\, \omega(E),
\end{equation}
where all the points $p_i$ are distinct, and $p_i \in \mathbb{R}-\Sigma_0$. Let $\pi_n$ and $\widetilde{\pi}_n$ be the undeformed and the deformed orthogonal polynomials respectively. We wish to show that
\begin{equation} \label{eq:CDPolyApp}
        \widetilde{\pi}_n(E) = \frac{(-1)^k}{ \prod_{i=1}^k (p_i - E)}
        \frac{\det\begin{pmatrix}
            \pi_n(p_1) & \pi_{n+1}(p_1) & \cdots & \pi_{n+k}(p_1) \\
            \vdots & \vdots & \ddots & \vdots \\
            \pi_{n}(p_k) & \pi_{n+1}(p_k) & \cdots & \pi_{n+k}(p_k) \\
            \pi_n(E) & \pi_{n+1}(E) & \cdots & \pi_{n+k}(E) 
            \end{pmatrix}}{\det \left[ \pi_{n + j -1}(p_i) \right]_{i,j = 1}^k}.
\end{equation}
To begin with, note that the numerator on the right hand side is a polynomial in $E$ (of degree at most $n+k$) which vanishes at the points $\{ p_i \}$. Therefore, $\widetilde{\pi}_n(E)$ must be divisible by each $(p_i - E)$, and so $\widetilde{\pi}_n$ is actually a degree $n$ polynomial. Orthogonality follows from the fact that $\widetilde{\pi}_n$ is a linear combination of the undeformed orthogonal polynomials $\pi_{n+l}$ for $l \geq 0$, i.e.,
\begin{equation}
    \widetilde{\pi}_n(E) = \frac{1}{\prod_{i=1}^k (p_i - E)}\sum_{l=0}^k c_l\, \pi_{n+l}(E),
\end{equation}
and therefore for any monomial $E^m$ with $m < n$, we find
\begin{align*}
    \int_{-1}^{1} dE\, \widetilde{\omega}(E)\, E^m \widetilde{\pi}_n(E) &= \sum_{l=0}^k c_l \mathcal{N}\int_{-1}^{1} dE\, \omega(E)\, E^m \pi_{n+l} (E) \\
    &= 0.
\end{align*}
The coefficient of the leading term in $\widetilde{\pi}_n$ is the coefficient of the degree-$(n+k)$ term on the right hand side, which is 1 on account of the division by $\det[ \pi_{n+j-1}(p_i)]_{i,j=1}^k$ in equation \eqref{eq:CDPolyApp}. Since $\pi_{n+k}$ is monic, it follows that $\widetilde{\pi}_n$ must also be monic, provided that $\det[ \pi_{n+j-1}(p_i)]_{i,j=1}^k$ is non-vanishing. To confirm that this determinant is indeed non-zero, let us assume to the contrary that it vanishes. Consider now the determinant $Q(E)$ defined as
\begin{equation}
    Q(E) = \det \begin{pmatrix}
        \pi_n(E) & \pi_{n+1}(E) & \cdots & \pi_{n+k-1}(E) \\
        \pi_n(p_2) & \pi_{n+1}(p_2) & \cdots & \pi_{n+k-1}(p_2) \\
        \vdots & \vdots & \ddots & \vdots \\
        \pi_n(p_k) & \pi_{n+1}(p_k) & \cdots & \pi_{n+k-1}(p_k)
    \end{pmatrix}.
\end{equation}
If $Q(p_1)=0$, then $Q(E)$ vanishes at each $p_i$. Therefore, $Q(E)$ must be divisible by $\prod_{i=1}^k (p_i - E)$ and we may write $Q(E) = \prod_{i=1}^k (p_i - E) q_{n-1}(E)$, where $q_{n-1}$ is a degree-$(n-1)$ polynomial. As before, observe that $Q(E)$ is a linear combination of the undeformed polynomials $\pi_{n+l}$ for $0 \leq l \leq k-1$, and therefore the overlap between $Q(E)$ and $q_{n-1}(E)$ vanishes. In other words,
\begin{align*}
    \int_{-1}^1 dE\, \omega(E) \,Q(E)\, q_{n-1}(E) \propto \int_{-1}^{1} dE\, \widetilde{\omega}(E) q_{n-1}^2(E) = 0,
\end{align*}
and therefore $q_{n-1} \equiv 0$, a contradiction. This completes the proof.

\subsection*{Proof of corollary \ref{cor:Cor1}}
Assume that both $\omega$ and its deformation $\widetilde{\omega}$ are even. We wish to show that
    \begin{equation} \label{eq:CDPolyEvenApp}
    \widetilde{\pi}_n(E) = \frac{(-1)^k}{\prod_{i=1}^k (p_i^2 - E^2)} \frac{\det 
    \begin{pmatrix}
        \pi_n(p_1) & \pi_{n+2}(p_1) & \cdots & \pi_{n+2k}(p_1) \\
        \vdots & \vdots & \ddots & \vdots \\
        \pi_n(p_k) & \pi_{n+2}(p_k) & \cdots & \pi_{n+2k}(p_k) \\
        \pi_n(E) & \pi_{n+2}(E) & \cdots & \pi_{n+2k}(E) 
    \end{pmatrix}}{\det[ \pi_{n + 2(j-1)}(p_i)]_{i,j=1}^k}.
\end{equation}
We begin by proving the following lemma.
\begin{lemma} \label{lm:PolSymm}
    Let $\pi_n$ be a monic orthogonal polynomial with respect to an even density $\omega$. Then $\pi_n(E) = (-1)^n \pi_n(-E)$.
\end{lemma}
\begin{proof}
    Since $\omega$ is even, $\pi_n(-E)$ is also an orthogonal polynomial, and since it has the same degree as $\pi_n(E)$, there exists some constant $c$ such that $\pi_n(-E) = c\, \pi_n(E)$. Now $\pi_n(-E) = (-1)^n E^n + \cdots$ and therefore $c = (-1)^n$.
\end{proof}

Let us now rewrite the determinant in the Christoffel formula as follows:
\begin{equation}
    \det\begin{pmatrix}
            \pi_n(p_1) & \pi_{n+1}(p_1) & \cdots & \pi_{n+2k}(p_1) \\
            \vdots & \vdots & \ddots & \vdots \\
            \pi_{n}(-p_k) & \pi_{n+1}(-p_k) & \cdots & \pi_{n+2k}(-p_k) \\
            \pi_n(E) & \pi_{n+1}(E) & \cdots & \pi_{n+2k}(E) 
            \end{pmatrix} =
    \det \begin{pmatrix}
        R_n^+(p_1) \\
        R_n^-(p_1) \\
        \vdots \\
        R_n^+(p_k) \\
        R_n^-(p_k) \\
        R_n(E) 
    \end{pmatrix},
\end{equation}
where each row is given by
\begin{align}
    R_{n}^\pm (E) &= \begin{pmatrix}
        \pi_n(\pm p_i) & \pi_{n+1}(\pm p_i) & \cdots & \pi_{n+k}(\pm p_i)
    \end{pmatrix} \\
    R_n(E) &= \begin{pmatrix}
        \pi_n(E) & \pi_{n+1}(E) & \cdots & \pi_{n+k}(E)
    \end{pmatrix}
\end{align}
We then perform a sequence of row operations that send $R_n^{\pm}(p_i) \to R_n^+(p_i) \pm R_n^-(p_i)$. Using lemma \ref{lm:PolSymm} and a sequence of column operations it is easy to verify that the determinant splits up into an even and an odd block, with the even block being the one shown in eq.\ \eqref{eq:CDPolyEvenApp}. Performing similar operations on the denominator of eq.\ \eqref{eq:CDPolyEvenApp}, we see that the odd blocks and the excess numerical pre-factors cancel out and we are only left with the determinants of the even blocks.

\subsection*{Proof of corollary \ref{cor:Cor2}}
Let $\bar{\mathsf{b}}_n$ and $\bar{b}_n$ be the deformed and undeformed Lanczos coefficients for an even Christoffel deformation. Our goal is to compute the ratio $(\bar{\mathsf{b}}_n / \bar{b}_n)^2$. For notational convenience, let us define the deformation factor $\Delta(E) = \prod_{i=1}^k (E^2 - p_i^2)$. Let us rewrite eq.\ \eqref{eq:CDPolyEvenApp} as
\begin{equation} \label{eq:DefPolyLC}
    \Delta(E) \widetilde{\pi}_n(E) = \pi_{n+2k}(E) + \sum_{j=0}^{k-1} c^{(n)}_j \pi_{n+2j}(E),
\end{equation}
where each coefficient $c_j^{(n)}$ can be expressed in terms of the appropriate sub-determinant in eq.\ \eqref{eq:CDPolyEvenApp}. Now we multiply both sides by $E$:
\begin{equation} \label{eq:DefPolyLC2}
    E\Delta(E) \widetilde{\pi}_n(E) = E\,\pi_{n+2k}(E) + \sum_{j=0}^{k-1} c^{(n)}_j \,E\,\pi_{n+2j}(E).
\end{equation}
Using the recursion relation \eqref{eq:DefLanczos} for $\widetilde{\pi}_n$ on the left hand side of equation \eqref{eq:DefPolyLC2} and then using equation \eqref{eq:DefPolyLC} again, we get
\begin{align}
    E \Delta(E) \widetilde{\pi}_n(E) = \pi_{n+1+2k}(E) + \sum_{j=0}^{k-1} c^{(n+1)}_j \pi_{n+1+2j}(E) + \bar{\mathsf{b}}_n^2 \left(\pi_{n-1 + 2k}(E)  + \sum_{j=0}^{k-1} c_j^{(n-1)} \pi_{n-1+2j}(E) \right). \label{eq:R1}
\end{align}
On the other hand, using the recursion relations \eqref{eq:UndefLanczos} on the right hand side of equation \eqref{eq:DefPolyLC2}, we obtain
\begin{equation}
    E \Delta(E) \widetilde{\pi}_n(E) = \pi_{n+1+2k}(E) + \bar{b}_{n+2k}^2 \pi_{n-1+2k}(E) + \sum_{j=0}^{k-1} c_j^{(n)}\left(\pi_{n+1+2j}(E) + \bar{b}_{n+2j}^2 \pi_{n-1+2j}(E)\right) .\label{eq:R2}
\end{equation}
Comparing the coefficients of $\pi_{n-1}$ in eqs.\ \eqref{eq:R1} and \eqref{eq:R2} we find the following relation between $\bar{\mathsf{b}}_n$ and $\bar{b}_n$: $(\bar{\mathsf{b}_n} / \bar{b}_n)^2 = c^{(n)}_0 / c^{(n-1)}_0$. Expressing the coefficients $c_0^{(n)}$ and $c_0^{(n-1)}$ in terms of the sub-determinants in eq.\ \eqref{eq:CDPolyEvenApp} leads us to the desired result.

\subsection*{Proof of theorem \ref{thm:RationalCD}}
This proof is a modified version of the proof given in \cite{ismail2005classical}. Consider a rational Christoffel deformation
\begin{equation}
    \widetilde{\omega}(E) = \mathcal{N} \frac{\prod_{i=1}^k (p_i - E)}{\prod_{j=1}^l (q_j - E)}\, \omega(E),
\end{equation}
where all the points $p_i$ and $q_j$ are distinct and lie outside the support of $\omega$. We wish to prove eq.\ \eqref{eq:RationalCDPoly}. The strategy is similar to the proof of theorem \ref{thm:Christoffel}. First note that $\widetilde{\pi}_n(E)$ as defined in eq.\ \eqref{eq:RationalCDPoly} is indeed a degree-$n$ polynomial since the determinant on the right hand side has zeros at each $p_i$. To show orthogonality, we need to prove that for every $m < n$
\begin{equation} \label{eq:ortho}
    \int_{-1}^{1} dE\, \widetilde{\omega}(E) E^m \widetilde{\pi}_n(E) = 0.
\end{equation}
For notational convenience, we define the following: 
\begin{align}
    A(E) &= \prod_{i=1}^k (p_i - E) \\
    B(E) &= \prod_{j=1}^l (q_j - E).
\end{align}
Now let us rewrite the deformed orthogonal polynomials as a linear combination of the undeformed ones:
\begin{equation}
    \widetilde{\pi}_n(E) = \frac{1}{A(E)} \sum_{j=-l}^{k} c_j \pi_{n+j}(E) \equiv \frac{S(E)}{A(E)},
\end{equation}
where $S = \sum_{j=-l}^k c_j \pi_{n+j}$ and the coefficients $c_j$ are determined by eq.\ \eqref{eq:RationalCDPoly}. We then have
\begin{align}
    \int_{-1}^{1} dE\, \widetilde{\omega}(E) E^m \widetilde{\pi}_n(E) = \mathcal{N} \int_{-1}^{1}dE\, \omega(E) \frac{E^m}{B(E)} S(E).
\end{align}
Performing polynomial division on the ratio $E^m / B(E)$, we can find two new polynomials $Q(E)$ and $R(E)$ such that
\begin{equation}
    E^m = B(E) Q(E) + R(E),
\end{equation}
where $\deg Q = m - l$ and $\deg R \leq l-1$.\footnote{If $m < \deg B = l$, then $Q=0$ and $R = E^m$.} Note that
\begin{equation}
    \int_{-1}^{1} dE\, \omega(E) Q(E) S(E) = 0,
\end{equation}
since $\deg Q = m - l < n-l$. Therefore, the only potential non-zero contribution comes from the remainder term. We have
\begin{equation} \label{eq:O1}
    \int_{-1}^{1} dE\, \widetilde{\omega}(E) E^m \widetilde{\pi}_n(E) = \mathcal{N} \int_{-1}^{1} dE\, \omega(E) \frac{R(E)}{B(E)} S(E).
\end{equation}
Since $\deg R \leq \deg B - 1$, the ratio $R/B$ decays to zero as $|E| \to \infty$ and has simple poles at the points $\{q_j\}_{j=1}^l$. Therefore, we can always find numbers $\{ \mu_j\}_{j=1}^l$ such that
\begin{equation}
    \frac{R(E)}{B(E)} = \sum_{j=1}^l \frac{\mu_j}{q_j - E}.
\end{equation}
Substituting the above in to eq.\ \eqref{eq:O1} we get
\begin{align}
    \begin{split}
        \int_{-1}^{1} dE\, \widetilde{\omega}(E) E^m \widetilde{\pi}_n(E) &= \mathcal{N} \sum_{i=1}^l \int_{-1}^{1} dE\, \omega(E) \frac{\mu_i}{q_i - E} S(E) \\
        &= \mathcal{N} \sum_{i=1}^l \mu_i \sum_{j=-l}^k c_j \int_{-1}^{1} dE\, \omega(E) \frac{\pi_{n+j}(E)}{q_i - E} \\
        &= \mathcal{N} \sum_{i=1}^l \mu_i \sum_{j=-l}^k c_j C[\pi_{n+j} \omega](q_i).
    \end{split}
\end{align}
Now, since $c_j$ is precisely the minor corresponding to $\pi_{n+j}(E)$ in the numerator of eq.\ \eqref{eq:RationalCDPoly}, for each $q_i$ the combination
\begin{equation}
    \sum_{j=-l}^k c_j C[\pi_{n+j} \omega](q_i) =0.
\end{equation}
This proves that $\widetilde{\pi}_n(E)$ is indeed orthogonal with respect to the deformed measure $\widetilde{\omega}$.

Finally, we would like to show that the denominator in eq.\ \eqref{eq:RationalCDPoly} is non-vanishing. To that end, consider the determinant
\begin{equation}
    D(E) = \det
    \begin{pmatrix}
        C[\pi_{n-l}\omega](q_1) & C[\pi_{n-l+1}\omega](q_1) & \dots &C[\pi_{n+k-1}\omega](q_1) \\
        \vdots & \vdots & \ddots & \vdots \\
        C[\pi_{n-l}\omega](q_l) & C[\pi_{n-l+1}\omega](q_l) & \dots &C[\pi_{n+k-1}\omega](q_l) \\
        \pi_{n-l}(E) & \pi_{n-l+1}(E) & \cdots & \pi_{n+k-1}(E) \\
        \pi_{n-l}(p_2) & \pi_{n-l+1}(p_2) & \cdots & \pi_{n+k-1}(p_2) \\
        \vdots & \vdots & \ddots & \vdots \\
        \pi_{n-l}(p_k) & \pi_{n-l+1}(p_k) & \cdots & \pi_{n+k-1}(p_k) 
    \end{pmatrix}
    \equiv \sum_{j=-l}^{k-1} d_j \pi_{n+j}(E).
\end{equation}
Assume to the contrary that the denominator vanishes. This implies that $D(E)$ has zeros at each $p_i$ and therefore there exists a degree-$(n-1)$ polynomial $q_{n-1}(E)$ such that $D(E) = A(E) q_{n-1}(E)$. Now consider
\begin{equation} \label{eq:O2}
    \int_{-1}^{1} dE\, \widetilde{\omega}(E) q^2_{n-1}(E) = \mathcal{N} \int_{-1}^{1} dE\, \omega(E) D(E) \frac{q_{n-1}(E)}{B(E)}.
\end{equation}
Now, $\deg B = l \leq n-1$. Therefore, we can find polynomials $h(E)$ and $r(E)$ such that
\begin{equation}
    q_{n-1}(E) = B(E) h(E) + r(E).
\end{equation}
Plugging this into eq.\ \eqref{eq:O2}, we observe that the term proportional to $h$ vanishes since $\deg h = n-l -1$ and $D(E) = \sum_{j=-l}^{k-1} d_j \pi_{n+j}(E)$. To process the term proportional to $r$, note that since $\deg r \leq l-1$, we can find numbers $\alpha_i$ such that
\begin{equation}
    \frac{r(E)}{B(E)} = \sum_{i=1}^l \frac{\alpha_i}{q_i - E}.
\end{equation}
Therefore, we have
\begin{align*}
    \int_{-1}^{1} dE\, \widetilde{\omega}(E) q^2_{n-1}(E) &= \mathcal{N} \sum_{i=1}^l \alpha_i \sum_{j=-l}^{k-1} d_j \int_{-1}^{1} dE\, \frac{\pi_{n+j}(E) \omega(E)}{q_i - E} \\
    &= \mathcal{N} \sum_{i=1}^l \alpha_i \sum_{j=-l}^{k-1} d_j C[\pi_{n+j} \omega](q_i).
\end{align*}
Since, the coefficients $d_j$ are the minors corresponding to $\pi_{n+j}$ in $D(E)$ the combination $\sum_{j=-l}^{k-1} d_j C[\pi_{n+j}\omega](q_i)$ vanishes for each $q_i$, and therefore
\begin{equation}
    \int_{-1}^{1} dE\, \widetilde{\omega}(E) q_{n-1}^2(E) = 0,
\end{equation}
a contradiction. This completes the proof.

\subsection*{Proof of corollary \ref{cor:Cor3}}
The proof is almost identical to that of corollary \ref{cor:Cor1}. The idea is to use the relations
\begin{align}
    \pi_n(-E) &= (-1)^n \pi_n(E) \\
    C[\pi_n \omega](-E) &= (-1)^n C[\pi_n \omega](E)
\end{align}
to cast the determinants in eq.\ \eqref{eq:RationalCDPoly} into block form and then observe that the odd block drops out. We leave it to the reader to fill in the details.

\section{Scaling limit of rational Christoffel deformations of DSSYK} \label{sec:RationalCDScaling}
We consider a rational Christoffel deformation that introduces a single pair of zeros at $\pm p$ and a single pair of poles at $\pm q$:
\begin{equation}
    \widetilde{\omega}(E) = \mathcal{N} \frac{p^2 - E^2}{q^2 - E^2}\, \omega(E).
\end{equation}
From Appendix \ref{app:RationalCD} (see eq.\ \eqref{eq:DefPolRationalEven}), we know that the deformed monic orthogonal polynomials can be written as
\begin{equation}
    (E^2 - p^2) \widetilde{\pi}_n(E) = \pi_{n+2}(E) - \eta_n \pi_n(E) + \delta_n \pi_{n-2}(E)
\end{equation}
where we have defined
\begin{align}
    \eta_n = \frac{\tau_{n-2, n+2}(p,q)}{\tau_{n-2, n}(p,q)}, \qquad
    \delta_n = \frac{\tau_{n, n+2}(p,q)}{\tau_{n-2, n}(p,q)}.
\end{align}
and $\tau_{a,b}(p,q) = C[\pi_a \omega](q) \pi_b(p) - (a \leftrightarrow b)$. Let $\bar{\mathsf{b}}_n$ and $\bar{b}_n$ be the deformed and undeformed Lanczos coefficients. The relation between the two is given in eq.\ \eqref{eq:RationalCDb} which, in this notation, becomes
\begin{equation} \label{eq:Rationalbn}
    \bar{\mathsf{b}}_n^2 = \frac{\eta_n}{\eta_{n-1}} \bar{b}_n^2 + \frac{\delta_{n+1} - \delta_n}{\eta_{n-1}},
\end{equation}
To compute the scaling limit of $\bar{\mathsf{b}}_n$, it is useful to introduce the ratios
\begin{align}
    r_n(p) = \frac{\pi_{n+1}(p)}{\pi_n(p)}, \qquad s_n(q) = \frac{C[\pi_{n+1} \omega](q)}{C[\pi_{n} \omega](q)}.
\end{align}
Then, the quantities $\eta_n$ and $\delta_n$ can be completely expressed in terms of the ratios $r_n$ and $s_n$. Indeed, note that the determinant $\tau_{a,b}$ can be written as
\begin{equation}
    \tau_{a,b}(p,q) = C[\pi_a \omega](q) \pi_a(p) \left( \prod_{j=a}^{b-1} r_j(p) - \prod_{j=a}^{b-1} s_j(q) \right).
\end{equation}
Plugging this into the definitions of $\eta_n$ and $\delta_n$, we get
\begin{align}
    \eta_n &= \frac{\prod_{j=-2}^1 r_{n+j}(p) - \prod_{j=-2}^{1} s_{n+j}(q)}{r_{n-2}(p) r_{n-1}(p) - s_{n-2}(q) s_{n-1}(q)}, \label{eq:etars} \\
    \delta_n &= r_{n-2}(p) r_{n-1}(p) s_{n-2}(q) s_{n-1}(q) \left( \frac{ r_n(p) r_{n+1}(p) - s_n(q) s_{n+1}(q)}{r_{n-2}(p) r_{n-1}(p) - s_{n-2}(q) r_{n-1}(q) } \right). \label{eq:deltars}
\end{align}
Therefore, in order to compute the triple scaling limits of $\eta_n$ and $\delta_n$ it suffices to compute the triple scaling limits of $r_n$ and $s_n$.

Now, recall from section \ref{sec:continuum} that $r_n(p)$ satisfies the recursion relation
\begin{equation}
    p = r_n(p) + \frac{\bar{b}_n^2}{r_{n-1}(p)}.
\end{equation}
To find the recursion relation satisfied by $s_n$, note that the Cauchy transforms satisfy the recursion relation (for $n \geq 1$)
\begin{equation}
    q C[\pi_n \omega](q) = C[\pi_{n+1} \omega](q) + \bar{b}_n^2 C[\pi_{n-1} \omega](q).
\end{equation}
Dividing both sides by $C[\pi_n \omega](q)$ we get the following recursion relation for $s_n$:
\begin{equation}
    q = s_n(q) + \frac{\bar{b}_n^2}{s_{n-1}(q)}.
\end{equation}
As before, we assume that $r_n(p)$ and $s_n(q)$ have well defined scaling limits $r(\ell, p)$ and $s(\ell, q)$ with $p$ and $q$ fixed. In this limit, the recursion relations become
\begin{align}
    p &= r(\ell, p) + \frac{1}{4}( 1 - \lambda^2 e^{-\ell}) \frac{1}{r(\ell - \lambda, p)}, \label{eq:rRec} \\
    q &= s(\ell, q) + \frac{1}{4} (1 - \lambda^2 e^{-\ell}) \frac{1}{s(\ell - \lambda, q)}. \label{eq:sRec}
\end{align}
As before, the idea now is to expand $r$ and $s$ in powers of $\lambda$ and solve for the coefficients order by order using eqs.\ \eqref{eq:rRec} and \eqref{eq:sRec}. Since the recursion relations are non-singular at $\lambda = 0$, we may postulate the following expansions:
\begin{align}
    r(\ell, p) &= \sum_{n=0}^{\infty} r_{(n)}(\ell, p) \lambda^n, \\
    r(\ell - \lambda, p) &= \sum_{n, m = 0}^{\infty} \frac{(-1)^n}{n!} \lambda^{n+m} \partial_{\ell}^n r_{(m)}(\ell, p), \\
    s(\ell, p) &= \sum_{n=0}^{\infty} s_{(n)}(\ell, p) \lambda^n, \\
    s(\ell - \lambda, p) &= \sum_{n, m = 0}^{\infty} \frac{(-1)^n}{n!} \lambda^{n+m} \partial_{\ell}^n s_{(m)}(\ell, p).
\end{align}
Although formally the eqs.\ \eqref{eq:rRec} and \eqref{eq:sRec} look the same, they have distinct perturbative solutions. This is because $r$ and $s$ have different asymptotic behaviors for large values of $p$ and $q$. It is clear from definition that $r(\ell, p) \sim p$ as $|p| \to \infty$. However, $s(\ell, q) \sim 1/4q$ as $|q| \to \infty$ (and $\lambda \to 0$). This follows from the large-$|q|$ expansion of the Cauchy transform and the orthogonality of $\pi_n$. Indeed, we have
\begin{align}
\begin{split}
    \int_{-1}^{1} dE\, \frac{\pi_n(E) \omega(E)}{q - E} &= \sum_{k=0}^{\infty} \frac{1}{q^{k+1}} \int_{-1}^{1} dE\,  E^k \pi_n(E) \omega(E) \\
    &= \frac{\gamma_n^2}{q^{n+1}} + O(q^{-n-2}),
\end{split}
\end{align}
where in the second line we have used orthogonality of $\pi_n$ and $\gamma_n$ is the norm of $\pi_n$ with respect to $\omega$. Therefore, asymptotically we have
\begin{equation}
    s_n(q) \underset{|q| \to \infty}{\sim} \left(\frac{\gamma_{n+1}}{\gamma_n}\right)^2 \frac{1}{q}.
\end{equation}
Now, for monic orthogonal polynomials $\gamma_{n+1} = \bar{b}_{n+1} \gamma_n$. Using the fact that $\bar{b}^2_n \to 1/4$ as $n\to \infty$ (and $\lambda \to 0$), we find that the leading order contribution to $s(\ell, q)$ at large $q$ and small $\lambda$ is given by $s(\ell, q) \sim 1/4q$.

With this distinction in mind, we can now compute the coefficients $r_{(n)}$ and $s_{(n)}$. As we saw in section \ref{sec:CDTS}, the asymptotic condition $r(\ell, p) \sim p$ implies that the solution for $r$ takes the form
\begin{equation} \label{eq:rO2}
    r(\ell, p) = \frac{\phi(p)}{2} \left( 1 + \lambda^2 \frac{1}{2(\phi^2(p) - 1)} e^{-\ell} \right) + O(\lambda^3)
\end{equation}
up to $O(\lambda^2)$. We can similarly find the solution for $s$. At $O(\lambda^0)$ we have
\begin{equation}
    s_{(0)}^2 - q s_{(0)} + \frac{1}{4} = 0.
\end{equation}
The physical solution with the correct asymptotic behavior is the negative branch and we have
\begin{equation}
    s_{(0)}(\ell, q) = \frac{q - \sqrt{q^2 - 1}}{2} = \frac{1}{2\phi(q)}.
\end{equation}
Repeating the arguments of section \ref{sec:CDTS} we see that $s_{(1)} = 0$. Then, using the above form of $s_{(0)}$, we get the following solution for $s_{(2)}$:
\begin{equation}
    s_{(2)}(\ell, q) = - \frac{\phi(q)}{2 (\phi(q)^2 - 1)} e^{-\ell}.
\end{equation}
Note that it has the same form as $r_{(2)}$ except for an overall minus sign. The perturbative solution for $s$ up to $O(\lambda^2)$ is thus given by
\begin{equation} \label{eq:sO2}
    s(\ell, q) = \frac{1}{2\phi(q)} \left(1 - \lambda^2 \frac{\phi^2(q)}{ \phi^2(q) - 1} \right) + O(\lambda^3) 
\end{equation}
Just as the polynomial case, the $O(\lambda^3)$ corrections to $r$ and $s$ do not contribute to the $O(\lambda^3)$ correction to the Lanczos ratio. To see this, let us assume that $\eta_n$ and $\delta_n$ have well-defined scaling limits $\eta(\ell)$ and $\delta(\ell)$. Since $r(\ell)$ and $s(\ell)$ have well-defined expansions around $\lambda = 0$, we may expand $\eta$ and $\delta$ as
\begin{align}
    \eta(\ell) &= \sum_{n=0}^{\infty} \lambda^n \eta_{(n)}(\ell) \label{eq:etaExp} \\
    \delta(\ell) &= \sum_{n=0}^{\infty} \lambda^n \delta_{(n)}(\ell). \label{eq:deltaExp}
\end{align}
Note that $\eta_{(0)}$ and $\delta_{(0)}$ are independent of $\ell$ and $\eta_{(1)} = \delta_{(1)} = 0$. This follows from eqs.\ \eqref{eq:etars} and \eqref{eq:deltars} and the fact that $r(\ell)$ and $s(\ell)$ have the above properties. Using the expansions \eqref{eq:etaExp} and \eqref{eq:deltaExp}, we see that in the triple scaling limit
\begin{align}
    \frac{\eta_n}{\eta_{n-1}} &\to \frac{\eta(\ell)}{\eta(\ell - \lambda)} = 1 - \lambda^3 \frac{\partial_\ell \eta_{(2)}(\ell)}{\eta_{(0)}(\ell)} + O(\lambda^4), \\
    \frac{\delta_{n+1} - \delta_n}{\eta_{n-1}} &\to \frac{\delta(\ell + \lambda) - \delta(\ell)}{\eta(\ell - \lambda)} = \lambda^3 \frac{\partial_{\ell}\delta_{(2)}(\ell)}{\eta_{(0)}(\ell)} + O(\lambda^4).
\end{align}
Plugging the above into eq.\ \eqref{eq:Rationalbn} and using the fact that $\bar{b}_n^2 = 1/4 + O(\lambda^2)$, we get
\begin{equation} \label{eq:LRCDR}
    \left( \frac{\bar{\mathsf{b}}_n^2}{\bar{b}_n^2} \right) = 1 + \lambda^3 \left( \frac{\partial_{\ell} \eta_{(2)}(\ell) + 4 \partial_{\ell} \delta_{(2)}(\ell)}{\eta_{(0)}} \right) + O(\lambda^4).
\end{equation}
Therefore it suffices to compute $\eta$ and $\delta$ (and therefore $r$ and $s$) up to $O(\lambda^2)$.

From eq.\ \eqref{eq:etars} it follows that in the triple scaling limit we have
\begin{equation}
    \eta(\ell) = \frac{\prod_{j=-2}^{1} r(\ell + j \lambda, p) - \prod_{j=-2}^{1} s(\ell + j \lambda, q) }{r(\ell - 2\lambda,p) r(\ell -\lambda, p) - s(\ell - 2\lambda, q) s(\ell - \lambda, q)}.
\end{equation}
Now, note that we can replace each factor $r(\ell + j\lambda)$ (and likewise $s(\ell + j \lambda)$) by $r(\ell)$ (and similarly $s(\ell)$) since the corrections coming from the shifts $\ell \to \ell + j \lambda$ appear at $O(\lambda^3)$:
\begin{align}
\begin{split}
    r(\ell + j\lambda, p) &=r(\ell) + j \lambda \partial_\ell r(\ell) + \cdots \\
    &= r(\ell) + j\lambda^3 \partial_{\ell} r_{(2)}(\ell) + \cdots,
\end{split}
\end{align}
where in the second line we have used the fact that $r_{(0)}$ is independent of $\ell$. Therefore the above expression for $\eta(\ell)$ can be simplified to
\begin{align} \label{eq:etaO2}
    \eta(\ell) &= \frac{r^4(\ell, p) - s^4(\ell, q)}{r^2(\ell, p) - s^2(\ell, q)} + O(\lambda^3) \\
    &= r^2(\ell, p) + s^2(\ell, q) + O(\lambda^3)
\end{align}
Similar manipulations on $\delta(\ell)$ lead us to the following:
\begin{equation} \label{eq:deltaO2}
    \delta(\ell) = r^2(\ell, p) s^2(\ell, q) + O(\lambda^3).
\end{equation}
Now, using eqs.\ \eqref{eq:rO2}, \eqref{eq:sO2}, \eqref{eq:etaO2}, \eqref{eq:deltaO2}, \eqref{eq:LRCDR} and some elementary algebra, one can show that the Lanczos ratio takes the form
\begin{equation}
    \left( \frac{\bar{\mathsf{b}}_n}{\bar{b}_n} \right)^2 = 1 - 2\lambda^3 e^{-\ell} \left( \frac{1}{\phi^2(p) - 1} - \frac{1}{\phi^2(q) - 1} \right) + O(\lambda^4).
\end{equation}
This precisely coincides with eq.\ \eqref{eq:conf} with $P=1$ and $Q=1$.

We can generalize the above to arbitrary but finite values of $P$ and $Q$ using induction. Consider a rational Christoffel deformation with $P$ pairs of zeros and $Q$ pairs of poles:
\begin{equation}
    \widetilde{\omega}^{(P,Q)}(E) = \mathcal{N} \frac{\prod_{i=1}^P(p_i^2 - E^2)}{\prod_{j=1}^Q (q_j^2 - E^2)}\, \omega(E).
\end{equation}
Let $\pi_n^{(P,Q)}(E)$ be the corresponding orthogonal polynomials and $\bar{\mathsf{b}}_n^{(P,Q)}$ be the associated Lanczos coefficients. Let us assume that $\bar{\mathsf{b}}_n^{(P,Q)}$ is given by
\begin{equation}\label{eq:CDpoles}
    \left( \frac{\bar{\mathsf{b}}_n^{(P,Q)}}{\bar{b}_n} \right)^2 = 1 - 2\lambda^3 e^{-\ell} \left( \sum_{i=1}^P \frac{1}{\phi^2(p_i) - 1} - \sum_{j=1}^Q \frac{1}{\phi^2(q_j) - 1} \right) + O(\lambda^4).
\end{equation}
We can derive eq.\ \eqref{eq:conf} by performing induction on $P$ and $Q$ separately. Now, consider adding in a new pair of zeros at a pair of new points $\pm p$ via the deformation $$\widetilde{\omega}^{(P,Q)}(E) \to \widetilde{\omega}^{(P+1,Q)}(E) \propto (p^2 - E^2) \widetilde{\omega}^{(P,Q)}(E).$$ Using arguments similar to the one given in section \ref{sec:CDTS}, one can show that
\begin{equation}
    \left( \frac{\bar{\mathsf{b}}^{(P+1,Q)}_n}{\bar{\mathsf{b}}^{(P,Q)}_n} \right)^2 = 1 - 2\lambda^3 e^{-\ell} \left( \frac{1}{\phi^2(p) - 1} \right) + O(\lambda^4),
\end{equation}
which completes the induction on $P$. For the induction step on $Q$, consider a deformation at introduces a pair of new poles at the points $\pm q$:
$$\widetilde{\omega}^{(P,Q)}(E) \to \widetilde{\omega}^{(P,Q+1)}(E) \propto \frac{1}{q^2 - E^2} \widetilde{\omega}^{(P,Q)}(E).$$
Then using the Christoffel deformation formula for poles and a similar set of manipulations, one can show that
\begin{equation}
    \left( \frac{\bar{\mathsf{b}}^{(P,Q+1)}_n}{\bar{\mathsf{b}}^{(P,Q)}_n} \right)^2 = 1 + 2\lambda^3 e^{-\ell} \left( \frac{1}{\phi^2(q) - 1} \right) + O(\lambda^4),
\end{equation}
which completes the induction on $Q$.

The above analysis shows that for any rational Christoffel deformation of the DSSYK density of states, corrections to the Liouville Hamiltonian appear only at $O(\lambda^3)$, and that the corrections are in agreement with eq.\ \eqref{eq:conf}. This serves as a consistency check between the Christoffel deformation method and the RHP method for computing deformations of Lanczos coefficients.

\section{$r(\ell,p)$ in DSSYK from q-Hermite polynomials}\label{sec:Hermite}
In this appendix, we wish to give a separate derivation of the formulas for $r(\ell, p)$ in the DSSYK model obtained in the main text using the recursion method. Here, we will directly use the fact that the monic polynomials in the DSSYK model are the continuous $\bq$-Hermite polynomials \cite{Berkooz:2024lgq}:
\beqn\label{eq:QHP}
\pi_{n}(p)&=& \frac{H_n(\cosh(x)|\bq)}{2^n}\nonumber\\
&=&\frac{1}{2^{n-1}}\sum_{k=0}^{[n/2]} \binom{n}{k}_\bq \cosh[(n-2k)x]\nonumber\\
&=&\frac{1}{2^{n}}\sum_{k=0}^{n} \binom{n}{k}_\bq e^{x(n-2k)} ,
\eeqn
where we have defined
\beq 
p= \cosh(x),\;\;\;e^{x} = \phi(p),
\eeq
and
\beq
\binom{n}{k}_\bq=\frac{(\bq;\bq)_n}{(\bq;\bq)_k (\bq;\bq)_{n-k}},
\eeq 
with $(a;\bq)_n$ being the $\bq$-Pochhammer symbol
\beq
(a;\bq)_n = \prod_{k=0}^{n-1}(1-a\,\bq^{k}).
\eeq 
 The second line in equation \eqref{eq:QHP} makes it clear that these are indeed polynomials in $p = \cosh(x)$. We wish to take the continuum limit (i.e., the triple scaling limit) of these monic polynomials, i.e.,  $\lm \to 0$, $n = \frac{1}{\lm}\log(1/\lm^2) + \frac{\ell}{\lm},$ with $\ell$ and $p=\cosh(x)$ fixed. We first write the polynomials as:
\beq
\pi_n(p) = \left(\frac{\phi(p)}{2}\right)^n g_n(x),\;\;g_n(x) = \sum_{k=0}^n \frac{(\bq;\bq)_n}{(\bq;\bq)_k (\bq;\bq)_{n-k}}\,e^{-2kx}.
\eeq
We now use the Euler-McLaurin formula to approximate the sum over $k$ in $g_n$ by an integral:
\beq
g_n(x) \sim \frac{1}{\lm}\int_0^{\log(1/\lm^2)+\ell} dy\,e^{\frac{1}{\lm}S(x,y,\ell)}+\cdots,
\eeq 
where we have defined $y = k\lm$ and the end-point contributions are exponentially suppressed in $\frac{1}{\lm}$. To write out the action, we need some way of writing the q-Pochhammer symbols that is amenable to the continuum limit. We will again use the Euler-Mclaurin formula. Consider, for instance:
\beqn \label{eq:Sk}
\log\,(\bq;\bq)_k &=& \sum_{j=1}^{k}\log(1-\bq^j)\nonumber\\
&\sim & \frac{1}{\lm}\int_{\lm}^{y}ds\log(1-e^{-s}) + \frac{1}{2}\log(1-e^{-\lm})+\frac{1}{2}\log(1-e^{-y}) + \cdots
\nonumber\\
&=& \frac{1}{\lm}\text{Li}_2(e^{-y}) + \frac{1}{2}\log(1-e^{-y}) + \cdots,
 \eeqn
where in the last line we have dropped terms which give rise to an overall $n$-independent coefficient in $g_n$, and so cancel out in the ratio $\frac{g_{n+1}}{g_n}$. The second term in the last line above comes from an endpoint contribution in the Euler-Mclaurin formula, and there are of course, higher order corrections of this type as well. Fortunately, it turns out that such endpoint contributions give rise to $O(\lm^3)$ corrections to the ratio $r_n$  -- this happens because in $(\bq;\bq)_k$, these endpoint terms do not directly depend on $\ell$, but only depend on $\ell$ indirectly through the saddle point value $y_*$, whose dependence on $\ell$ turns out to enter at $O(\lm^2)$, as we will see below. Consequently the contribution of the endpoint terms to the ratio enters at $O(\lm^3).$ Evaluating the other Pochhammer symbols similarly, we find:
\beq 
\log \frac{(\bq;\bq)_n}{(\bq;\bq)_{n-k}} \sim \frac{1}{\lm}\left(\text{Li}_2(\lm^2 e^{-\ell}) -\text{Li}_2(\lm^2 e^{y-\ell})\right) + \frac{1}{2}\log(1-\lm^2 e^{-\ell}) - \frac{1}{2}\log(1-\lm^2 e^{y-\ell})+\cdots.
\eeq 
Here, note that the endpoint contributions enter at $O(\lm^2)$ and so their contribution to the ratio enters at $O(\lm^3)$; so we can once again drop the endpoint contributions if we only care about the ratio $r_n$. Thus, the action becomes:
\beq 
S(x,y,\ell) = \text{Li}_2(\lm^2 e^{-\ell}) - \text{Li}_2(\lm^2 e^{y-\ell})-\text{Li}_2(e^{-y}) - 2xy+ \cdots,
\eeq 
up to terms that contribute at $O(\lm^3)$ in the ratio. We will now perform the integral over $y$ using the saddle point method. The equation of motion for $y$ takes the form:\footnote{Note that the endpoint contribution to the action in equation \eqref{eq:Sk} enters the equation of motion at $O(\lm)$. However, it changes the solution by an $O(\lm)$ term that is $\ell$-independent, and so modifies the on-shell action at $O(\lm^2)$ in an $\ell$-independent way. So this contribution also enters the ratio at $O(\lm^3)$.}
\beq 
(e^{2x}-1)e^y - e^{2x}+O(\lm)=0.
\eeq 
In the $\lm \to 0$ limit, the solution is given by
\beq 
e^{y_*} = \frac{e^{2x}}{(e^{2x}-1)}+ O(\lm).
\eeq 
Note that $y_*$ does not depend on $\ell$ at leading order in $\lm$; this dependence only enters at $O(\lm^2)$. Evaluating the on-shell action, we get 
\beq 
S_{\text{o.s}}(x,\ell) = s_0(x) -\frac{\lm^2 e^{-\ell}}{e^{2x}-1} + \cdots,
\eeq  
where $s_0$ does not depend on $\ell$. Computing the desired ratio, we get
\beq
r_n = \frac{\phi}{2}\frac{g_{n+1}}{g_n} 
\simeq   \frac{\phi}{2} e^{\partial_{\ell}S_{\text{o.s}}} = \frac{\phi}{2}\left(1+\lm^2 \frac{e^{-\ell}}{(\phi^2 -1)}+\cdots\right).
\eeq
This precisely agrees with the formulas we obtained from the recursion relation, see equations \eqref{eq:assumpt} and \eqref{eq:r_2}. The one-loop correction about the saddle point contributes at $O(\lm^3)$. 

\bibliographystyle{JHEP}
\bibliography{Reference_jt.bib}

@article{Berkooz:2024lgq,
    author = "Berkooz, Micha and Mamroud, Ohad",
    title = "{A cordial introduction to double scaled SYK}",
    eprint = "2407.09396",
    archivePrefix = "arXiv",
    primaryClass = "hep-th",
    doi = "10.1088/1361-6633/ada889",
    journal = "Rept. Prog. Phys.",
    volume = "88",
    number = "3",
    pages = "036001",
    year = "2025"
}

@article{Berkooz:2018jqr,
    author = "Berkooz, Micha and Isachenkov, Mikhail and Narovlansky, Vladimir and Torrents, Genis",
    title = "{Towards a full solution of the large N double-scaled SYK model}",
    eprint = "1811.02584",
    archivePrefix = "arXiv",
    primaryClass = "hep-th",
    doi = "10.1007/JHEP03(2019)079",
    journal = "JHEP",
    volume = "03",
    pages = "079",
    year = "2019"
}

@article{Berkooz:2018qkz,
    author = "Berkooz, Micha and Narayan, Prithvi and Simon, Joan",
    title = "{Chord diagrams, exact correlators in spin glasses and black hole bulk reconstruction}",
    eprint = "1806.04380",
    archivePrefix = "arXiv",
    primaryClass = "hep-th",
    doi = "10.1007/JHEP08(2018)192",
    journal = "JHEP",
    volume = "08",
    pages = "192",
    year = "2018"
}

@article{Lin:2022rbf,
    author = "Lin, Henry W.",
    title = "{The bulk Hilbert space of double scaled SYK}",
    eprint = "2208.07032",
    archivePrefix = "arXiv",
    primaryClass = "hep-th",
    doi = "10.1007/JHEP11(2022)060",
    journal = "JHEP",
    volume = "11",
    pages = "060",
    year = "2022"
}

@article{Rabinovici:2023yex,
    author = "Rabinovici, E. and S{\'a}nchez-Garrido, A. and Shir, R. and Sonner, J.",
    title = "{A bulk manifestation of Krylov complexity}",
    eprint = "2305.04355",
    archivePrefix = "arXiv",
    primaryClass = "hep-th",
    doi = "10.1007/JHEP08(2023)213",
    journal = "JHEP",
    volume = "08",
    pages = "213",
    year = "2023"
}

@article{Ambrosini:2024sre,
    author = "Ambrosini, Marco and Rabinovici, Eliezer and S{\'a}nchez-Garrido, Adri{\'a}n and Shir, Ruth and Sonner, Julian",
    title = "{Operator K-complexity in DSSYK: Krylov complexity equals bulk length}",
    eprint = "2412.15318",
    archivePrefix = "arXiv",
    primaryClass = "hep-th",
    reportNumber = "CERN-TH-2025-040",
    doi = "10.1007/JHEP08(2025)059",
    journal = "JHEP",
    volume = "08",
    pages = "059",
    year = "2025"
}

@article{Sachdev_1993,
   title={Gapless spin-fluid ground state in a random quantum Heisenberg magnet},
   volume={70},
   ISSN={0031-9007},
   url={http://dx.doi.org/10.1103/PhysRevLett.70.3339},
   DOI={10.1103/physrevlett.70.3339},
   number={21},
   journal={Physical Review Letters},
   publisher={American Physical Society (APS)},
   author={Sachdev, Subir and Ye, Jinwu},
   year={1993},
   month=May, pages={3339–3342} }

@misc{Kitaev2,
author = {Kitaev, Alexei},
title = {A simple model of holography 2},
journal = {Talk at KITP},
year= 2015,
 howpublished = {\url{https://online.kitp.ucsb.edu/online/entangled15/kitaev2/}},
}

@misc{Kitaev1,
  author = {Kitaev, Alexei},
  title = {A simple model of quantum holography 1},
  journal = {Talk at KITP},
  year = 2015,
  howpublished = {\url{https://online.kitp.ucsb.edu/online/entangled15/kitaev/}},
}

@article{Maldacena:2016hyu,
    author = "Maldacena, Juan and Stanford, Douglas",
    title = "{Remarks on the Sachdev-Ye-Kitaev model}",
    eprint = "1604.07818",
    archivePrefix = "arXiv",
    primaryClass = "hep-th",
    doi = "10.1103/PhysRevD.94.106002",
    journal = "Phys. Rev. D",
    volume = "94",
    number = "10",
    pages = "106002",
    year = "2016"
}

@article{Almheiri:2014cka,
    author = "Almheiri, Ahmed and Polchinski, Joseph",
    title = "{Models of AdS$_{2}$ backreaction and holography}",
    eprint = "1402.6334",
    archivePrefix = "arXiv",
    primaryClass = "hep-th",
    doi = "10.1007/JHEP11(2015)014",
    journal = "JHEP",
    volume = "11",
    pages = "014",
    year = "2015"
}

@article{Maldacena:2016upp,
    author = "Maldacena, Juan and Stanford, Douglas and Yang, Zhenbin",
    title = "{Conformal symmetry and its breaking in two dimensional Nearly Anti-de-Sitter space}",
    eprint = "1606.01857",
    archivePrefix = "arXiv",
    primaryClass = "hep-th",
    doi = "10.1093/ptep/ptw124",
    journal = "PTEP",
    volume = "2016",
    number = "12",
    pages = "12C104",
    year = "2016"
}

@article{Saad:2019lba,
    author = "Saad, Phil and Shenker, Stephen H. and Stanford, Douglas",
    title = "{JT gravity as a matrix integral}",
    eprint = "1903.11115",
    archivePrefix = "arXiv",
    primaryClass = "hep-th",
    month = "3",
    year = "2019"
}

@article{Kitaev:2017awl,
    author = "Kitaev, Alexei and Suh, S. Josephine",
    title = "{The soft mode in the Sachdev-Ye-Kitaev model and its gravity dual}",
    eprint = "1711.08467",
    archivePrefix = "arXiv",
    primaryClass = "hep-th",
    doi = "10.1007/JHEP05(2018)183",
    journal = "JHEP",
    volume = "05",
    pages = "183",
    year = "2018"
}

@article{Sarosi:2017ykf,
    author = "S{\'a}rosi, G{\'a}bor",
    title = "{AdS$_{2}$ holography and the SYK model}",
    eprint = "1711.08482",
    archivePrefix = "arXiv",
    primaryClass = "hep-th",
    doi = "10.22323/1.323.0001",
    journal = "PoS",
    volume = "Modave2017",
    pages = "001",
    year = "2018"
}

@article{Cotler:2016fpe,
    author = "Cotler, Jordan S. and Gur-Ari, Guy and Hanada, Masanori and Polchinski, Joseph and Saad, Phil and Shenker, Stephen H. and Stanford, Douglas and Streicher, Alexandre and Tezuka, Masaki",
    title = "{Black Holes and Random Matrices}",
    eprint = "1611.04650",
    archivePrefix = "arXiv",
    primaryClass = "hep-th",
    reportNumber = "SU-ITP-16-19, SU-ITP-16/19, YITP-16-124",
    doi = "10.1007/JHEP05(2017)118",
    journal = "JHEP",
    volume = "05",
    pages = "118",
    year = "2017",
    note = "[Erratum: JHEP 09, 002 (2018)]"
}

@article{Bagrets:2016cdf,
    author = "Bagrets, Dmitry and Altland, Alexander and Kamenev, Alex",
    editor = "Unno, Yoshinobu and Ohsugi, Takashi and Hou, Suen and Sadrozinski, Hartmut F. -W. and Lou, Xinchou and Zhu, Hongbo and Ouyang, Qun",
    title = "{Sachdev{\textendash}Ye{\textendash}Kitaev model as Liouville quantum mechanics}",
    eprint = "1607.00694",
    archivePrefix = "arXiv",
    primaryClass = "cond-mat.str-el",
    doi = "10.1016/j.nuclphysb.2016.08.002",
    journal = "Nucl. Phys. B",
    volume = "911",
    pages = "191--205",
    year = "2016"
}

@article{Harlow:2018tqv,
    author = "Harlow, Daniel and Jafferis, Daniel",
    title = "{The Factorization Problem in Jackiw-Teitelboim Gravity}",
    eprint = "1804.01081",
    archivePrefix = "arXiv",
    primaryClass = "hep-th",
    doi = "10.1007/JHEP02(2020)177",
    journal = "JHEP",
    volume = "02",
    pages = "177",
    year = "2020"
}

@article{Jafferis:2022wez,
    author = "Jafferis, Daniel Louis and Kolchmeyer, David K. and Mukhametzhanov, Baur and Sonner, Julian",
    title = "{Jackiw-Teitelboim gravity with matter, generalized eigenstate thermalization hypothesis, and random matrices}",
    eprint = "2209.02131",
    archivePrefix = "arXiv",
    primaryClass = "hep-th",
    doi = "10.1103/PhysRevD.108.066015",
    journal = "Phys. Rev. D",
    volume = "108",
    number = "6",
    pages = "066015",
    year = "2023"
}

@article{Okuyama:2023kdo,
    author = "Okuyama, Kazumi",
    title = "{Discrete analogue of the Weil-Petersson volume in double scaled SYK}",
    eprint = "2306.15981",
    archivePrefix = "arXiv",
    primaryClass = "hep-th",
    doi = "10.1007/JHEP09(2023)133",
    journal = "JHEP",
    volume = "09",
    pages = "133",
    year = "2023"
}

@article{Okuyama:2023aup,
    author = "Okuyama, Kazumi and Suyama, Takao",
    title = "{Solvable limit of ETH matrix model for double-scaled SYK}",
    eprint = "2311.02846",
    archivePrefix = "arXiv",
    primaryClass = "hep-th",
    doi = "10.1007/JHEP04(2024)094",
    journal = "JHEP",
    volume = "04",
    pages = "094",
    year = "2024"
}

@article{Okuyama:2024eyf,
    author = "Okuyama, Kazumi",
    title = "{Baby universe operators in the ETH matrix model of double-scaled SYK}",
    eprint = "2408.03726",
    archivePrefix = "arXiv",
    primaryClass = "hep-th",
    doi = "10.1007/JHEP10(2024)249",
    journal = "JHEP",
    volume = "10",
    pages = "249",
    year = "2024"
}

@article{Miyaji:2025ucp,
    author = "Miyaji, Masamichi and Mori, Soichiro and Okuyama, Kazumi",
    title = "{Finite N bulk Hilbert space in ETH matrix model for double-scaled SYK. Null states, state-dependence and Krylov state complexity}",
    eprint = "2505.13194",
    archivePrefix = "arXiv",
    primaryClass = "hep-th",
    reportNumber = "YITP 25-71",
    doi = "10.1007/JHEP08(2025)084",
    journal = "JHEP",
    volume = "08",
    pages = "084",
    year = "2025"
}

@article{Blommaert:2024ymv,
    author = "Blommaert, Andreas and Mertens, Thomas G. and Papalini, Jacopo",
    title = "{The dilaton gravity hologram of double-scaled SYK}",
    eprint = "2404.03535",
    archivePrefix = "arXiv",
    primaryClass = "hep-th",
    doi = "10.1007/JHEP06(2025)050",
    journal = "JHEP",
    volume = "06",
    pages = "050",
    year = "2025"
}

@article{Blommaert:2025avl,
    author = "Blommaert, Andreas and Levine, Adam and Mertens, Thomas G. and Papalini, Jacopo and Parmentier, Klaas",
    title = "{Wormholes, branes and finite matrices in sine dilaton gravity}",
    eprint = "2501.17091",
    archivePrefix = "arXiv",
    primaryClass = "hep-th",
    doi = "10.1007/JHEP09(2025)123",
    journal = "JHEP",
    volume = "09",
    pages = "123",
    year = "2025"
}

@incollection{kuijlaars2003riemann,
  title={Riemann-Hilbert analysis for orthogonal polynomials},
  author={Kuijlaars, Arno BJ},
  booktitle={Orthogonal Polynomials and Special Functions: Leuven 2002},
  pages={167--210},
  year={2003},
  publisher={Springer}
}

@article{Fokas:1991za,
    author = "Fokas, A. S. and Its, A. R. and Kitaev, A. V.",
    title = "{The Isomonodromy approach to matrix models in 2-D quantum gravity}",
    reportNumber = "INS-183",
    doi = "10.1007/BF02096594",
    journal = "Commun. Math. Phys.",
    volume = "147",
    pages = "395--430",
    year = "1992"
}

@misc{bleher,
      title={Semiclassical asymptotics of orthogonal polynomials, Riemann-Hilbert problem, and universality in the matrix model}, 
      author={Pavel Bleher and Alexander Its},
      year={1999},
      eprint={math-ph/9907025},
      archivePrefix={arXiv},
      primaryClass={math-ph},
      url={https://arxiv.org/abs/math-ph/9907025}, 
}

@article{Deift_et_al,
author = {Deift, P. and Kriecherbauer, T. and McLaughlin, K. T-R and Venakides, S. and Zhou, X.},
title = {Strong asymptotics of orthogonal polynomials with respect to exponential weights},
journal = {Communications on Pure and Applied Mathematics},
volume = {52},
number = {12},
pages = {1491-1552},
doi = {https://doi.org/10.1002/(SICI)1097-0312(199912)52:12<1491::AID-CPA2>3.0.CO;2-\#},
url = {https://onlinelibrary.wiley.com/doi/abs/10.1002/%28SICI%291097-0312%28199912%2952%3A12%3C1491%3A%3AAID-CPA2%3E3.0.CO%3B2-%23},

year = {1999}
}

@book{Mehta,
author= {Mehta, Madal Lal},
title = {Random Matrices},
publisher = {Academic Press},
edition = {Revised and Enlarged Second Edition},
address = {San Diego},
pages = {xv},
year = {1991}
}

@article{DiFrancesco:1993cyw,
    author = "Di Francesco, P. and Ginsparg, Paul H. and Zinn-Justin, Jean",
    title = "{2-D Gravity and random matrices}",
    eprint = "hep-th/9306153",
    archivePrefix = "arXiv",
    reportNumber = "LA-UR-93-1722, SACLAY-SPH-T-93-061",
    doi = "10.1016/0370-1573(94)00084-G",
    journal = "Phys. Rept.",
    volume = "254",
    pages = "1--133",
    year = "1995"
}

@article{Aguilar-Gutierrez:2026ogo,
    author = "Aguilar-Gutierrez, Sergio E.",
    title = "{Deforming the Double-Scaled SYK and Reaching the Stretched Horizon From Finite Cutoff Holography}",
    eprint = "2602.06113",
    archivePrefix = "arXiv",
    primaryClass = "hep-th",
    doi = "10.1002/prop.70112",
    journal = "Fortsch. Phys.",
    volume = "74",
    number = "5",
    pages = "e70112",
    year = "2026"
}

@article{Aguilar-Gutierrez:2026nmd,
    author = "Aguilar-Gutierrez, Sergio E. and Kukolj, Trivko and Seitz, Josef",
    title = "{q-Askey Deformations of Double-Scaled SYK}",
    eprint = "2605.13956",
    archivePrefix = "arXiv",
    primaryClass = "hep-th",
    month = "5",
    year = "2026"
}

@article{Gibbons:1976ue,
    author = "Gibbons, G. W. and Hawking, S. W.",
    title = "{Action Integrals and Partition Functions in Quantum Gravity}",
    reportNumber = "PRINT-76-0995 (CAMBRIDGE)",
    doi = "10.1103/PhysRevD.15.2752",
    journal = "Phys. Rev. D",
    volume = "15",
    pages = "2752--2756",
    year = "1977"
}

@article{Lewkowycz:2013nqa,
    author = "Lewkowycz, Aitor and Maldacena, Juan",
    title = "{Generalized gravitational entropy}",
    eprint = "1304.4926",
    archivePrefix = "arXiv",
    primaryClass = "hep-th",
    doi = "10.1007/JHEP08(2013)090",
    journal = "JHEP",
    volume = "08",
    pages = "090",
    year = "2013"
}

@article{Eynard:2015aea,
    author = "Eynard, Bertrand and Kimura, Taro and Ribault, Sylvain",
    title = "{Random matrices}",
    eprint = "1510.04430",
    archivePrefix = "arXiv",
    primaryClass = "math-ph",
    month = "10",
    year = "2015"
}

@article{Marolf:2017kvq,
    author = "Marolf, Donald and Parrikar, Onkar and Rabideau, Charles and Izadi Rad, Ali and Van Raamsdonk, Mark",
    title = "{From Euclidean Sources to Lorentzian Spacetimes in Holographic Conformal Field Theories}",
    eprint = "1709.10101",
    archivePrefix = "arXiv",
    primaryClass = "hep-th",
    doi = "10.1007/JHEP06(2018)077",
    journal = "JHEP",
    volume = "06",
    pages = "077",
    year = "2018"
}

@article{Balasubramanian:2022tpr,
    author = "Balasubramanian, Vijay and Caputa, Pawel and Magan, Javier M. and Wu, Qingyue",
    title = "{Quantum chaos and the complexity of spread of states}",
    eprint = "2202.06957",
    archivePrefix = "arXiv",
    primaryClass = "hep-th",
    doi = "10.1103/PhysRevD.106.046007",
    journal = "Phys. Rev. D",
    volume = "106",
    number = "4",
    pages = "046007",
    year = "2022"
}

@article{Baiguera:2025dkc,
    author = "Baiguera, Stefano and Balasubramanian, Vijay and Caputa, Pawel and Chapman, Shira and Haferkamp, Jonas and Heller, Michal P. and Halpern, Nicole Yunger",
    title = "{Quantum complexity in gravity, quantum field theory, and quantum information science}",
    eprint = "2503.10753",
    archivePrefix = "arXiv",
    primaryClass = "hep-th",
    reportNumber = "YITP-25-39",
    doi = "10.1016/j.physrep.2025.11.001",
    journal = "Phys. Rept.",
    volume = "1159",
    pages = "1--77",
    year = "2026"
}

@article{Nandy:2024evd,
    author = "Nandy, Pratik and Matsoukas-Roubeas, Apollonas S. and Mart{\'\i}nez-Azcona, Pablo and Dymarsky, Anatoly and del Campo, Adolfo",
    title = "{Quantum dynamics in Krylov space: Methods and applications}",
    eprint = "2405.09628",
    archivePrefix = "arXiv",
    primaryClass = "quant-ph",
    reportNumber = "RIKEN-iTHEMS-Report-24",
    doi = "10.1016/j.physrep.2025.05.001",
    journal = "Phys. Rept.",
    volume = "1125-1128",
    pages = "1--82",
    year = "2025"
}

@article{Rabinovici:2025otw,
    author = "Rabinovici, Eliezer and S{\'a}nchez-Garrido, Adri{\'a}n and Shir, Ruth and Sonner, Julian",
    title = "{Krylov Complexity}",
    eprint = "2507.06286",
    archivePrefix = "arXiv",
    primaryClass = "hep-th",
    reportNumber = "CERN-TH-2025-128",
    month = "7",
    year = "2025"
}

@article{Parker:2018yvk,
    author = "Parker, Daniel E. and Cao, Xiangyu and Avdoshkin, Alexander and Scaffidi, Thomas and Altman, Ehud",
    title = "{A Universal Operator Growth Hypothesis}",
    eprint = "1812.08657",
    archivePrefix = "arXiv",
    primaryClass = "cond-mat.stat-mech",
    doi = "10.1103/PhysRevX.9.041017",
    journal = "Phys. Rev. X",
    volume = "9",
    number = "4",
    pages = "041017",
    year = "2019"
}

@book{viswanath1994recursion,
  title={The recursion method: application to many-body dynamics},
  author={Viswanath, VS and M{\"u}ller, Gerhard},
  year={1994},
  publisher={Springer}
}

@article{Papadodimas:2013wnh,
    author = "Papadodimas, Kyriakos and Raju, Suvrat",
    title = "{Black Hole Interior in the Holographic Correspondence and the Information Paradox}",
    eprint = "1310.6334",
    archivePrefix = "arXiv",
    primaryClass = "hep-th",
    reportNumber = "ICTS-2013-20, ICTS/2013/20",
    doi = "10.1103/PhysRevLett.112.051301",
    journal = "Phys. Rev. Lett.",
    volume = "112",
    number = "5",
    pages = "051301",
    year = "2014"
}

@article{Papadodimas:2013jku,
    author = "Papadodimas, Kyriakos and Raju, Suvrat",
    title = "{State-Dependent Bulk-Boundary Maps and Black Hole Complementarity}",
    eprint = "1310.6335",
    archivePrefix = "arXiv",
    primaryClass = "hep-th",
    reportNumber = "ICTS-2013-21, ICTS/2013/21",
    doi = "10.1103/PhysRevD.89.086010",
    journal = "Phys. Rev. D",
    volume = "89",
    number = "8",
    pages = "086010",
    year = "2014"
}

@article{Almheiri:2014lwa,
    author = "Almheiri, Ahmed and Dong, Xi and Harlow, Daniel",
    title = "{Bulk Locality and Quantum Error Correction in AdS/CFT}",
    eprint = "1411.7041",
    archivePrefix = "arXiv",
    primaryClass = "hep-th",
    reportNumber = "SU-ITP-14-30, SU-ITP-14/30",
    doi = "10.1007/JHEP04(2015)163",
    journal = "JHEP",
    volume = "04",
    pages = "163",
    year = "2015"
}

@article{Basu:2024tgg,
    author = "Basu, Ritam and Ganguly, Anirban and Nath, Souparna and Parrikar, Onkar",
    title = "{Complexity growth and the Krylov-Wigner function}",
    eprint = "2402.13694",
    archivePrefix = "arXiv",
    primaryClass = "hep-th",
    doi = "10.1007/JHEP05(2024)264",
    journal = "JHEP",
    volume = "05",
    pages = "264",
    year = "2024",
    note = "[Erratum: JHEP 03, 202 (2026)]"
}

@article{Basu:2025mmm,
    author = "Basu, Ritam and Chowdhury, Pratyusha and Ganguly, Anirban and Nath, Souparna and Parrikar, Onkar and Paul, Suprakash",
    title = "{Wigner negativity, random matrices and gravity}",
    eprint = "2506.02110",
    archivePrefix = "arXiv",
    primaryClass = "hep-th",
    doi = "10.1007/JHEP01(2026)106",
    journal = "JHEP",
    volume = "01",
    pages = "106",
    year = "2026"
}

@article{Kar:2021nbm,
    author = "Kar, Arjun and Lamprou, Lampros and Rozali, Moshe and Sully, James",
    title = "{Random matrix theory for complexity growth and black hole interiors}",
    eprint = "2106.02046",
    archivePrefix = "arXiv",
    primaryClass = "hep-th",
    doi = "10.1007/JHEP01(2022)016",
    journal = "JHEP",
    volume = "01",
    pages = "016",
    year = "2022"
}

@article{Muck:2022xfc,
    author = {M{\"u}ck, Wolfgang and Yang, Yi},
    title = "{Krylov complexity and orthogonal polynomials}",
    eprint = "2205.12815",
    archivePrefix = "arXiv",
    primaryClass = "hep-th",
    doi = "10.1016/j.nuclphysb.2022.115948",
    journal = "Nucl. Phys. B",
    volume = "984",
    pages = "115948",
    year = "2022"
}

@article{Murugan:2026rfa,
    author = "Murugan, Jeff and Van Zyl, Hendrik J. R. and Watanabe, Masataka",
    title = "{Spectral Topology and Universal Krylov Dynamics}",
    eprint = "2608.07258",
    archivePrefix = "arXiv",
    primaryClass = "hep-th",
    month = "8",
    year = "2026"
}

@article{Balasubramanian:2022dnj,
    author = "Balasubramanian, Vijay and Magan, Javier M. and Wu, Qingyue",
    title = "{Tridiagonalizing random matrices}",
    eprint = "2208.08452",
    archivePrefix = "arXiv",
    primaryClass = "hep-th",
    doi = "10.1103/PhysRevD.107.126001",
    journal = "Phys. Rev. D",
    volume = "107",
    number = "12",
    pages = "126001",
    year = "2023"
}

@article{Witten:1998qj,
    author = "Witten, Edward",
    title = "{Anti de Sitter space and holography}",
    eprint = "hep-th/9802150",
    archivePrefix = "arXiv",
    reportNumber = "IASSNS-HEP-98-15",
    doi = "10.4310/ATMP.1998.v2.n2.a2",
    journal = "Adv. Theor. Math. Phys.",
    volume = "2",
    pages = "253--291",
    year = "1998"
}

@article{Gubser:1998bc,
    author = "Gubser, S. S. and Klebanov, Igor R. and Polyakov, Alexander M.",
    title = "{Gauge theory correlators from noncritical string theory}",
    eprint = "hep-th/9802109",
    archivePrefix = "arXiv",
    reportNumber = "PUPT-1767",
    doi = "10.1016/S0370-2693(98)00377-3",
    journal = "Phys. Lett. B",
    volume = "428",
    pages = "105--114",
    year = "1998"
}

@article{Erd_s_2014,
   title={Phase Transition in the Density of States of Quantum Spin Glasses},
   volume={17},
   ISSN={1572-9656},
   url={http://dx.doi.org/10.1007/s11040-014-9164-3},
   DOI={10.1007/s11040-014-9164-3},
   number={3-4},
   journal={Mathematical Physics, Analysis and Geometry},
   publisher={Springer Science and Business Media LLC},
   author={Erdős, László and Schröder, Dominik},
   year={2014},
   month=Dec, pages={441–464} }

@article{Dumitriu:2002ntg,
    author = "Dumitriu, Ioana and Edelman, Alan",
    title = "{Matrix models for beta ensembles}",
    eprint = "math-ph/0206043",
    archivePrefix = "arXiv",
    doi = "10.1063/1.1507823",
    journal = "J. Math. Phys.",
    volume = "43",
    number = "11",
    pages = "5830--5847",
    year = "2002"
}

@article{linalg2,
author = {Druskin, Vladimir and Knizhnerman, Leonid},
title = {Krylov subspace approximation of eigenpairs and matrix functions in exact and computer arithmetic},
journal = {Numerical Linear Algebra with Applications},
volume = {2},
number = {3},
pages = {205-217},
doi = {https://doi.org/10.1002/nla.1680020303},
url = {https://onlinelibrary.wiley.com/doi/abs/10.1002/nla.1680020303},
year = {1995}
}

@article{Jiang:2019pam,
    author = "Jiang, Jiaqi and Yang, Zhenbin",
    title = "{Thermodynamics and Many Body Chaos for generalized large q SYK models}",
    eprint = "1905.00811",
    archivePrefix = "arXiv",
    primaryClass = "hep-th",
    doi = "10.1007/JHEP08(2019)019",
    journal = "JHEP",
    volume = "08",
    pages = "019",
    year = "2019"
}

@article{Anninos:2022qgy,
    author = "Anninos, Dionysios and Galante, Dami{\'a}n A. and Sheorey, Sameer U.",
    title = "{Renormalisation group flows of deformed SYK models}",
    eprint = "2212.04944",
    archivePrefix = "arXiv",
    primaryClass = "hep-th",
    doi = "10.1007/JHEP11(2023)197",
    journal = "JHEP",
    volume = "11",
    pages = "197",
    year = "2023"
}

@article{Yang:2018gdb,
    author = "Yang, Zhenbin",
    title = "{The Quantum Gravity Dynamics of Near Extremal Black Holes}",
    eprint = "1809.08647",
    archivePrefix = "arXiv",
    primaryClass = "hep-th",
    doi = "10.1007/JHEP05(2019)205",
    journal = "JHEP",
    volume = "05",
    pages = "205",
    year = "2019"
}

@article{Berkooz:2024ofm,
    author = "Berkooz, Micha and Brukner, Nadav and Jia, Yiyang and Mamroud, Ohad",
    title = "{Path integral for chord diagrams and chaotic-integrable transitions in double scaled SYK}",
    eprint = "2403.05980",
    archivePrefix = "arXiv",
    primaryClass = "hep-th",
    doi = "10.1103/PhysRevD.110.106015",
    journal = "Phys. Rev. D",
    volume = "110",
    number = "10",
    pages = "106015",
    year = "2024"
}

@article{Berkooz:2024ifu,
    author = "Berkooz, Micha and Frumkin, Ronny and Mamroud, Ohad and Seitz, Josef",
    title = "{Twisted times, the Schwarzian and its deformations in DSSYK}",
    eprint = "2412.14238",
    archivePrefix = "arXiv",
    primaryClass = "hep-th",
    doi = "10.1007/JHEP05(2025)080",
    journal = "JHEP",
    volume = "05",
    pages = "080",
    year = "2025"
}

@article{Balasubramanian:2026klv,
    author = "Balasubramanian, Vijay and Caputa, Pawel and Parrikar, Onkar and Singh, Vivek",
    title = "{Wigner negativity in Krylov space and emergent semiclassicality}",
    eprint = "2607.01351",
    archivePrefix = "arXiv",
    primaryClass = "hep-th",
    month = "7",
    year = "2026"
}

@article{linalg0,
author = {Saad, Y.},
title = {Analysis of Some Krylov Subspace Approximations to the Matrix Exponential Operator},
journal = {SIAM Journal on Numerical Analysis},
volume = {29},
number = {1},
pages = {209-228},
year = {1992},
doi = {10.1137/0729014},

URL = { 
    
        https://doi.org/10.1137/0729014
    
    

},
eprint = { 
    
        https://doi.org/10.1137/0729014
    
    

}

}

@article{linalg,
author = {Hochbruck, Marlis and Lubich, Christian},
title = {On Krylov Subspace Approximations to the Matrix Exponential Operator},
journal = {SIAM Journal on Numerical Analysis},
volume = {34},
number = {5},
pages = {1911-1925},
year = {1997},
doi = {10.1137/S0036142995280572},

URL = { 
    
        https://doi.org/10.1137/S0036142995280572
    
    

},
eprint = { 
    
        https://doi.org/10.1137/S0036142995280572
    
    

}
}

@article{kuijlaars2004riemann,
  title={The Riemann--Hilbert approach to strong asymptotics for orthogonal polynomials on [- 1, 1]},
  author={Kuijlaars, Arno BJ and McLaughlin, KT-R and Van Assche, Walter and Vanlessen, Maarten},
  journal={Advances in mathematics},
  volume={188},
  number={2},
  pages={337--398},
  year={2004},
  publisher={Elsevier}
}

@article{deift1993steepest,
  title={A steepest descent method for oscillatory Riemann--Hilbert problems. Asymptotics for the MKdV equation},
  author={Deift, Percy and Zhou, Xin},
  journal={Annals of Mathematics},
  pages={295--368},
  year={1993},
  publisher={JSTOR}
}

@book{szeg1939orthogonal,
  title={Orthogonal polynomials},
  author={Szeg{\H o}, Gabor},
  volume={23},
  year={1939},
  publisher={American Mathematical Soc.}
}

@book{ismail2005classical,
  title={Classical and quantum orthogonal polynomials in one variable},
  author={Ismail, Mourad},
  volume={13},
  year={2005},
  publisher={Cambridge university press}
}

\end{document}